\documentclass[11pt]{article}

\newif\ifblind
\blindfalse

\newif\ifdraft
\draftfalse

\usepackage{rpmacros}
\usepackage[dvipsnames,svgnames,x11names,hyperref]{xcolor}
\usepackage{stmaryrd}
\usepackage{amsthm}
\usepackage{xfrac}
\usepackage{mathpazo}
\usepackage{gitinfo}
\usepackage{enumitem}
\usepackage{fullpage}
\usepackage{thmtools}
\usepackage[T1]{fontenc} 
\usepackage[colorlinks]{hyperref}
\hypersetup{
  linkcolor=[rgb]{0,0,0.4},
  citecolor=[rgb]{0, 0.4, 0},
  urlcolor=[rgb]{0.6, 0, 0}
}
\usepackage{lipsum}
\usepackage{mathrsfs,bbm}
\usepackage{bm}
\usepackage{todonotes}
\usepackage{tikz}
\usetikzlibrary{decorations.pathreplacing,backgrounds}
\usepackage{multirow}
\usepackage{setspace}

\usepackage{microtype}

\usepackage[linesnumbered,vlined,ruled]{algorithm2e}

\usepackage{amsthm}
\usepackage{thmtools,thm-restate}
\usepackage[nameinlink,capitalise]{cleveref}

\numberwithin{equation}{section}
\declaretheoremstyle[bodyfont=\it,qed=\qedsymbol]{noproofstyle}

\declaretheorem[numberlike=equation]{observation}

\declaretheorem[name=Observation,numbered=no]{observation*}

\declaretheorem[numberlike=equation]{theorem}

\declaretheorem[name=Theorem,numbered=no]{theorem*}

\declaretheorem[numberlike=equation]{lemma}
\declaretheorem[name=Lemma,numbered=no]{lemma*}

\declaretheorem[numberlike=equation]{corollary}
\declaretheorem[name=Corollary,numbered=no]{corollary*}

\declaretheorem[name=Proposition,numbered=no]{proposition*}

\declaretheorem[numberlike=equation]{claim}
\declaretheorem[name=Claim,numbered=no]{claim*}

\declaretheorem[name=Conjecture,numbered=no]{conjecture*}

\declaretheorem[numberlike=equation]{question}
\declaretheorem[name=Question,numbered=no]{question*}

\declaretheoremstyle[bodyfont=\it,qed=$\lozenge$]{defstyle} 

\declaretheorem[numberlike=equation,style=defstyle]{definition}
\declaretheorem[unnumbered,name=Definition,style=defstyle]{definition*}

\declaretheorem[unnumbered,name=Example,style=defstyle]{example*}

\declaretheorem[unnumbered,name=Notation=defstyle]{notation*}

\declaretheorem[unnumbered,name=Construction,style=defstyle]{construction*}

\declaretheorem[numberlike=equation,style=defstyle]{remark}
\declaretheorem[unnumbered,name=Remark,style=defstyle]{remark*}

\renewcommand{\phi}{\varphi}
\renewcommand{\epsilon}{\varepsilon}

\newcommand{\size}{\operatorname{size}}

\newcommand{\SP}{\Sigma\Pi}
\newcommand{\SPsize}[1]{(\SP)^k\operatorname{-size}}

\newcommand{\PSymToESym}{\operatorname{PSymToESym}}
\newcommand{\ESymToPSym}{\operatorname{ESymToPSym}}
\newcommand{\DivTest}{\operatorname{DivTest}} 
\usepackage{nth}
\usepackage{intcalc}
\usepackage{etoolbox}
\usepackage{xstring}

\usepackage{ifpdf}
\ifpdf
\else
\usepackage[quadpoints=false]{hypdvips}
\fi

\newcommand{\ECCCurl}[2]{http://eccc.hpi-web.de/report/\ifnumcomp{#1}{>}{93}{19}{20}#1/#2/}

\newcommand{\shortECCC}[2]{\texttt{\href{http://eccc.hpi-web.de/report/\ifnumcomp{#1}{>}{93}{19}{20}#1/#2/}{eccc:TR#1-#2}}}

\newcommand{\parseECCC}[1]{
\StrSubstitute{#1}{TR}{}[\tmpstring]%
\IfSubStr{\tmpstring}{/}{ 
\StrBefore{\tmpstring}{/}[\ecccyear]%
\StrBehind{\tmpstring}{/}[\ecccreport]%
}{
\StrBefore{\tmpstring}{-}[\ecccyear]%
\StrBehind{\tmpstring}{-}[\ecccreport]%
}%
\shortECCC{\ecccyear}{\ecccreport}}

\newcommand{\homog}{\operatorname{Hom}}
\newcommand*\samethanks[1][\value{footnote}]{\footnotemark[#1]}

\ifdraft
\newcommand{\RPnote}[1]{\textcolor{WildStrawberry}{\guillemotleft RP: #1 \guillemotright}}
\newcommand{\MKnote}[1]{\textcolor{BlueGreen}{\guillemotleft Mrinal: #1 \guillemotright}}
\newcommand{\VRnote}[1]{\textcolor{Blue}{\guillemotleft VR: #1 \guillemotright}}
\newcommand{\SBnote}[1]{\textcolor{OliveGreen}{\guillemotleft SB: #1 \guillemotright}}
\newcommand{\SSnote}[1]{\textcolor{BrickRed}{\guillemotleft SS: #1 \guillemotright}}
\newcommand{\gitinfonotecolour}{Gray}
\newcommand{\easteregg}{}
\else
\newcommand{\RPnote}[1]{}
\newcommand{\MKnote}[1]{}
\newcommand{\VRnote}[1]{}
\newcommand{\SBnote}[1]{}
\newcommand{\SSnote}[1]{}
\newcommand{\gitinfonotecolour}{white}
\newcommand{\easteregg}{Easter-egg? Didn't factor that in...}
\fi

\newcommand{\ignore}[1]{}
\newcommand{\gitinfonote}{git info:~\gitAbbrevHash\;,\;(\gitAuthorIsoDate)\; \;\gitVtag}

\newcommand{\Res}[3]{\ensuremath{\operatorname{Res}_{#1}(#2,#3)}} 

\renewcommand{\vec}[1]{\ensuremath{\bm{#1}}}

\newcommand{\K}{\ensuremath{\mathbb{K}}}

\newcommand{\trunc}{\;\mathrm{trunc}\;}

\newcommand{\psym}{\mathbf{Psym}}
\newcommand{\esym}{\mathbf{Esym}}
\newcommand{\ki}{\mathbf{KI}}

\usepackage{environ}
\makeatletter
\newcommand{\voidenvironment}[1]{%
  \expandafter\providecommand\csname env@#1@save@env\endcsname{}%
  \expandafter\providecommand\csname env@#1@process\endcsname{}%
  \@ifundefined{#1}{}{\RenewEnviron{#1}{}}%
}
\makeatother

\title{Deterministic factorization of constant-depth algebraic circuits in subexponential time} 

\ifblind
\else
\author{
    {Somnath Bhattacharjee \thanks{University of Toronto, Canada. Email: \texttt{somnath.bhattacharjee@mail.utoronto.ca}}}
    \and
    {Mrinal Kumar \thanks{Tata Institute of Fundamental Research, Mumbai, India. Email: \texttt{\{mrinal, varun.ramanathan, ramprasad\}@tifr.res.in}.  Research supported by the Department of Atomic Energy, Government of India, under project number RTI400112, and in part by Google and SERB Research Grants. }}
    \and
    {Varun Ramanathan{\samethanks[2]}}
    \and
    {Ramprasad Saptharishi{\samethanks[2]}}
    \and
    {Shubhangi Saraf \thanks{University of Toronto, Canada. Email: \texttt{shubhangi.saraf@utoronto.ca}}}
}
\fi
\date{}

\begin{document}

\maketitle
\begin{abstract}
While efficient randomized algorithms for factorization of polynomials given by algebraic circuits have been known for decades, obtaining an even \emph{slightly non-trivial} deterministic algorithm for this problem has remained an open question of great interest. This is true even when the input algebraic circuit has additional structure, for instance, when it is a constant-depth circuit. Indeed, no efficient deterministic algorithms are known even for the seemingly easier problem of factoring sparse polynomials or even the problem of testing the irreducibility of sparse polynomials. 

In this work, we make progress on these questions: we design a deterministic algorithm that runs in subexponential time, and when given as input a constant-depth algebraic circuit $C$ over the field of rational numbers, it outputs  algebraic circuits (of potentially unbounded depth) for all the irreducible factors of $C$, together with their multiplicities. In particular, we give the first subexponential time deterministic algorithm for factoring sparse polynomials. 

For our proofs, we rely on a finer understanding of the structure of power series roots of constant-depth circuits and the analysis of the Kabanets-Impagliazzo  generator. In particular, we show that the Kabanets-Impagliazzo generator constructed using low-degree hard polynomials (explicitly constructed in the work of Limaye, Srinivasan \& Tavenas) preserves not only the non-zeroness of small constant-depth circuits (as shown by Chou, Kumar \& Solomon), but also their irreducibility and the irreducibility of their factors.  
\end{abstract}

\newpage
\tableofcontents

\section{Introduction} \label{sec:intro}
The main problem of interest in this work is that of polynomial factorization --- \emph{given a polynomial as input, output its decomposition into a product of irreducible polynomials.} 

For this paper, we work in the setting where the input is a multivariate polynomial which is specified by a (small) algebraic circuit computing it, and we are over the field $\Q$ of rational numbers. This problem saw significant progress starting in the 1980s where a sequence of results culminating in the works of Kaltofen \cite{K89} and Kaltofen \& Trager \cite{KT90} gave a randomized algorithm that when given a size $s$ algebraic circuit $C$ computing an $n$ variate degree $d$ polynomial over $\Q$, terminated in $\poly(s,d,n)$ time and output algebraic circuits for all the irreducible factors of $C$ (together with their multiplicities). A surprising fact that is implicit in these results is a \emph{closure result} for polynomials computable by small circuits - all irreducible factors of $C$ have algebraic circuits of size $\poly(s,d,n)$. This is necessary for us to be able to entertain any hopes of having a polynomial time algorithm for this problem since a priori it is not even clear if there is a description of the output (namely, the irreducible factors of $C$) that is polynomially bounded in the input parameters $s, d, n$. 

These polynomial factorization algorithms represent a significant landmark in our understanding of a fundamental problem in computational algebra on their own. However, in hindsight, the impact of this line of research seems to go far beyond this original context - these results and the techniques discovered in the course of their proofs have since found many diverse applications in algorithm design, pseudorandomness, coding theory and complexity theory, e.g. \cite{S97, GS99, Alekhnovich2005, bogdanov05, DGV24}. 

Among the most important problems in this broad area of polynomial factorization that continue to be open is that of derandomizing the results of Kaltofen \cite{K89} and Kaltofen \& Trager \cite{KT90}. In fact, even before randomized factoring algorithms were studied for general circuits, they were studied in the setting of sparse polynomials. A beautiful work by von zur Gathen and Kaltofen ~\cite{GathenKaltofen85} gave the first randomized factoring algorithm for sparse polynomials\footnote{The running time obtained was polynomial in the sparsity of the factors}. Even for sparse polynomials, the problem of derandomizing these factoring algorithms is of great interest and has received considerable attention over the last decade or two. 

We now know from the work of Shpilka \& Volkovich \cite{SV10} that efficient deterministic polynomial factorization is at least as hard as efficient deterministic polynomial identity testing- the reduction simply being that to check whether a multivariate $f(\vecx)$ is non-zero, we check if the polynomial $(f(\vecx) + yz)$ is irreducible. This connection clearly continues to hold for structured sub-classes for algebraic circuits like formulas, branching programs, sparse polynomials and constant-depth circuits. Thus, the task of derandomizing polynomial factorization for any such subclass of circuits must necessarily be preceded by a non-trivial deterministic polynomial identity testing algorithm for the subclass. 

Perhaps a bit surprisingly,  Kopparty, Saraf \& Shpilka \cite{KSS15} also showed a reduction in the other direction - efficient deterministic polynomial identity testing for algebraic circuits implies efficient deterministic polynomial factorization for algebraic circuits. This reduction between the problems of polynomial factorization and polynomial identity testing continue to hold in both the black box model (where we only have query access to the input algebraic circuits) and the white box model (where we can look inside the given circuits). An important aspect of the reduction in \cite{KSS15} is that even for the task of factorizing polynomials computed by very structured circuits like formulas, constant-depth circuits or even sparse polynomials, the PIT instances that we encounter on the way are for seemingly much more powerful circuit classes like algebraic branching programs. Thus, non-trivial PIT algorithms for a structured circuit class do not immediately yield a non-trivial deterministic algorithm for factorization of polynomials computed in this class. 

A very natural example of this phenomenon is that of sparse polynomials 
 (which are depth 2 circuits) and even general constant-depth circuits. 
 For the case of sparse polynomials, we have known polynomial-time PIT algorithms for a while from the work of Klivans \& Spielman \cite{KS01}, and for arbitrary constant-depth circuits, we also now have non-trivial deterministic algorithms for polynomial identity testing of constant-depth circuits. These follow from the lower bounds for constant-depth circuits in the recent work of Limaye, Srinivasan and Tavenas \cite{LST21} and the connections between hardness and derandomization for such circuits in the work of Chou, Kumar and Solomon \cite{ChouKS19}\footnote{An alternative deterministic subexponential time algorithm for PIT for constant-depth circuits  was given by Andrews and Forbes \cite{AF22} in a subsequent work.}. However,  these deterministic polynomial identity tests do not seem to immediately imply non-trivial deterministic factorization results for either sparse polynomials or polynomials computed by constant-depth circuits.  In fact, the following seemingly simpler question also appears to be open. 
\begin{question}\label{q:sparse-irreducibility}
  Design a subexponential-time deterministic algorithm that when given a sparse polynomial as input, decides if it is irreducible. 
\end{question}
More generally, as Forbes \& Shpilka mention (Questions 1.4 and 4.1 in \cite{Forbes-Shpilka-survey})  in their survey on polynomial factorization, very natural questions around derandomization of polynomial factorization algorithms are wide open. 
\subsection{Our results}
Our main result in this paper addresses this question and more generally, the question of deterministic polynomial factorization for constant-depth circuits. More precisely, we prove the following. 
\begin{theorem}[Informal version of \autoref{thm:main-algorithm}] \label{thm:main-intro}
    For every constant $\epsilon > 0$, constant $\Delta \in \N$ and sufficiently large $n$, the following statement is true. 
    
    There is a deterministic algorithm that takes as input an algebraic circuit $C$ over the field $\Q$ on $n$ variables with depth $\Delta$, and size and  degree $\poly(n)$, runs in time $\exp(O(n^{\epsilon}))$ and outputs algebraic circuits (of potentially unbounded depth) for all irreducible factors of $C$, together with their multiplicities. 
\end{theorem}

This result in particular also gives the first subexponential deterministic factoring algorithm for sparse polynomials (which is the case when $\Delta =2$).

At the heart of our proof is the deterministic construction of a variable reduction map that reduces the number of variables in a constant-depth circuit substantially, while preserving its \emph{factorization profile}. The following theorem is our main technical result.\footnote{The theorem follows immediately from \autoref{thm:irreducibility-preservation} and our deterministic implementation of some (fairly standard) preprocessing steps in the factorization algorithm. In fact, by a slight modification of our proofs, a stronger version of this theorem can be shown to be true where the time complexity of constructing the map $\Gamma_{\epsilon, \Delta}$ is quasipolynomially bounded in $n$. However, this does not improve the overall time complexity of the algorithm in \autoref{thm:main-intro} since some steps in the algorithm, for instance, that of deterministic factorization of $n^{\epsilon}$ variate polynomials of degree $\poly(n)$, still run in time $\exp(O(n^{\epsilon}))$. } 

\begin{theorem}[Informal version of \autoref{thm:irreducibility-preservation}] \label{thm:main-intro-irreducibility-preserving}
    For every $\epsilon > 0$, every constant $\Delta \in \N$ and all sufficiently large $n$, there is a polynomial map $\Gamma_{\epsilon, \Delta}:\Q^{n^{\epsilon}} \to \Q^n$ of degree $o(\log n)$ with the following properties. 
    \begin{itemize}
        \item Given $\epsilon, \Delta$, the map $\Gamma_{\epsilon, \Delta}$ can be constructed deterministically in time $\exp(O(n^{\epsilon}))$. 
        \item $\Gamma_{\epsilon, \Delta}$ is a variable reduction map that preserves the irreducibility of $n$-variate polynomials that can be computed by depth $\Delta$ circuits of size and degree $\poly(n)$. In other words, an $n$-variate algebraic circuit $C$ of size and degree $\poly(n)$ is irreducible if and only if the $n^{\epsilon}$-variate polynomial $\Hat{C}$ obtained from $C$ by composing it with $\Gamma_{\epsilon, \Delta}$ is irreducible. 
        \item More generally, $\Gamma_{\epsilon, \Delta}$ is a variable reduction map that preserves the irreducibility of the factors of $n$-variate polynomials that can be computed by depth $\Delta$ circuits of size and degree $\poly(n)$. 
    \end{itemize}
\end{theorem}

\paragraph{Obtaining \autoref{thm:main-intro} from \autoref{thm:main-intro-irreducibility-preserving}. }
The second item of the above theorem, when combined with known deterministic factorization algorithms for $n^{\epsilon}$-variate polynomials of degree at most $\poly(n)$ that run in time $n^{O(n^{\epsilon})}$ (polynomial in the size of the dense representation of such polynomials), essentially gives us a subexponential time deterministic irreducibility testing algorithm for sparse polynomials, and more generally for constant-depth circuits. Similarly, the third item of \autoref{thm:main-intro-irreducibility-preserving} can be combined with known ideas in algorithms for polynomial factorization (with some effort) to give \autoref{thm:main-intro}. 

We end this section with a brief discussion of some quantitative and qualitative aspects of the main theorems. 

\paragraph{Improvements in time complexity. } Given the reduction from polynomial identity testing to polynomial factorization in \cite{SV10}, we know that one cannot hope to improve the running time of the algorithm in \autoref{thm:main-intro} significantly for general constant-depth circuits, unless we have significantly better deterministic polynomial identity testing algorithms for these problems. However, for the case of sparse polynomials, we have known polynomial time PIT algorithms for more than two decades \cite{KS01}, and thus we can, in principle, expect faster factorization or irreducibility testing algorithms for sparse polynomials. It is a very natural question to explore further. 

\paragraph{Closure results for constant-depth circuits.} We do not know if irreducible factors of polynomials with small constant-depth circuits (or even the more restricted class of sparse polynomials) must be computable by small constant-depth circuits, and \autoref{thm:main-intro} does not appear to shine any light on such a closure result. As a consequence, the outputs of the algorithm in \autoref{thm:main-intro} are general unbounded depth circuits of polynomial size. In fact, we do not even know of non-trivial upper bounds on the size of constant-depth circuits for irreducible factors of sparse polynomials.  It would be very interesting to obtain such closure results or to gather some evidence that points towards such a statement being false.

\paragraph{Field dependence in \autoref{thm:main-intro}. } Our results in \autoref{thm:main-intro} are stated over the field of rational numbers. However, the results hold a little more generally - our proofs continue to work over any underlying field over which we have an efficient deterministic algorithm for factoring univariate polynomials, over which the lower bounds of Limaye, Srinivasan \& Tavenas \cite{LST21} for constant-depth circuits, the results of Andrews \& Wigderson \cite{AW24}, and the Taylor-expansion-based techniques of Chou, Kumar \& Solomon \cite{ChouKS19} continue to work. The field of rational numbers satisfies all these properties. The results also continue to hold over finite fields of moderately large characteristic (polynomially large in the degree and number of variables) since they satisfy all the properties stated above. However, for our presentation in this paper, we just work over the field of rational numbers and skip the minor technical changes needed to see the extension of the results to finite fields of moderately large characteristic. 

\subsection{Related prior work}
Following the subexponential time deterministic PIT algorithms for constant-depth algebraic circuits that follow from the lower bounds in the work of Limaye, Srinivasan \& Tavenas \cite{LST21}, there has been a renewed interest in obtaining deterministic algorithms for factorization of constant-depth circuits. This includes the results of Kumar, Ramanathan \& Saptharishi  \cite{KRS23}, who gave a subexponential time deterministic algorithm to compute all constant-degree factors of polynomials computed by constant-depth circuits and an alternative proof as well as a generalization of this result by Dutta, Sinhababu \& Thierauf \cite{DST24}. However, these results do not appear to give anything non-trivial for factors that are not low-degree. 

Another result that is very relevant here is a work of Kumar, Ramanathan, Saptharishi \& Volk \cite{KRSV} that gave a deterministic subexponential time algorithm that on input a constant-depth circuit outputs a list of circuits, of unbounded depth and with division gates, such that every irreducible factor of the input polynomial is computed by some circuit in this list. However, a significant drawback of this result is that the output list could contain circuits that do not correspond to any factor of the input, or aren't even valid circuits in the sense that they involve division by a circuit that computes an identically zero polynomial. In particular, their algorithm does not imply a deterministic subexponential time algorithm for testing irreducibility of a sparse polynomial since even when the input is irreducible, the algorithm could output a collection of circuits of tentative factors. However, in this case, none of the output circuits correspond to a true factor of the input polynomial. Unfortunately, since the depth of these output circuits is potentially large (and they contain division gates), we have no deterministic way of checking if there is a polynomial in the output list that corresponds to a true factor of the input. Thus, the question of pruning the output list deterministically in \cite{KRSV} to identify true factors is non-trivial and is essentially open. 

A recent work of Andrews \& Wigderson \cite{AW24} gives efficient parallel algorithms (essentially constant-depth circuits) for a variety of problems related to polynomial factorization, such as GCD and LCM computation of polynomials. As a consequence of their techniques, they show that the resultant and the discriminant polynomials can be computed by constant-depth circuits of polynomial size, and the squarefree part of a polynomial with small constant-depth circuits is a small constant-depth circuit. The techniques of Andrews \& Wigderson also give an alternative proof of the results in \cite{KRSV}. However, it suffers from the same drawback as \cite{KRSV} - some of the circuits in the output list might not correspond to any valid factors of the input, and might not even be valid circuits since they might have a division by an identically zero circuit built inside. 

Even though the aforementioned results represent some interesting progress towards the  question of obtaining deterministic factorization of constant-depth circuits continues, it seems fair to say that the original problem has largely remained open. In particular, these results do not say anything at all about even the question of sparse irreducibility testing (\autoref{q:sparse-irreducibility}), which seems like a natural and possibly simpler intermediate step on the way to obtaining deterministic factorization of constant-depth circuits.  

Factoring of sparse polynomials has been studied for a long time and it was initiated by the work of
von zur Gathen and Kaltofen almost four decades ago~\cite{GathenKaltofen85}. This work gave the first randomized factoring algorithm for sparse polynomials, and the running time obtained was polynomial in the sparsity of the factors. Ever since this work, it has been an interesting question to obtain deterministic factoring algorithms for sparse polynomials, especially since we have known how to derandomize PIT for this class for a while.

There are special cases of sparse polynomials for which we do know some results on irreducibility testing and factoring. The work of Bhargava, Saraf \& Volkovich~\cite{BSV18} gave a quasipolynomial time deterministic factoring algorithm for sparse polynomials where each variable has bounded degree. This extends prior works by Shpilka \& Volkovich~\cite{SV10} which gave deterministic factoring algorithms for multilinear sparse polynomials and Volkovich~\cite{Volkovich17} which deterministic factoring algorithms for multiquadratic sparse polynomials. 

As the main result of this paper (\autoref{thm:main-intro}), we  give a subexponential time deterministic algorithm for factorizing polynomials computable by small constant-depth circuits and, in particular, sparse polynomials (which are equivalent to depth-2 circuits). On the way to the proof, we obtain a deterministic subexponential time algorithm for testing the irreducibility of polynomials computable by small constant-depth circuits. In terms of techniques, we rely on some of the insights from prior work, e.g. \cite{ChouKS19, KRS23, KRSV, AW24} and introduce some new ideas on the way that might have further applications for problems of this nature. In particular, for our proofs, we crucially rely on  technical observations about the structure of power series roots of polynomials with small constant-depth circuits, and their interaction with tools from hardness-vs-randomness for algebraic circuits, e.g. the Kabanets-Impagliazzo generator. As our main technical statement, we show that this generator, when invoked with a low-degree polynomial that is hard for constant-depth circuits (as constructed  in \cite{LST21}) not only preserves the non-zeroness of constant-depth circuits\footnote{For this, we need to assume that the circuit has undergone some initial preprocessing, which again is done using a similar generator. } (as was already known), but also preserves their irreducibility (and the irreducibility of their factors). 

In the next section, we discuss these techniques in greater detail and give an overview of the proof of \autoref{thm:main-intro}. 

\section{Overview of proofs}\label{sec:proof-overview}

Almost all factorization algorithms proceed with some initial pre-processing to guarantee the following requirements:
\begin{itemize}\itemsep 0pt
\item $P(\vecx, z)$ is squarefree and monic in $z$.
\item We have that $P(\veczero, 0) = 0$ and $\partial_z(P)(\veczero, 0) \neq 0$. 
\end{itemize}
The above non-degeneracy conditions can typically be guaranteed relatively easily via a randomized algorithm. To do this deterministically, we use a recent result of Andrews \& Wigderson \cite{AW24} (see \autoref{thm:andrews-wigderson-squarefree-decomposition}) to get a squarefree decomposition of $P$. In fact, their theorem says something more that is useful for us: the squarefree parts of $P$ are computable by a small constant-depth circuit since $P$ is computable by a small constant-depth circuit. Moreover, we can obtain this decomposition via a deterministic subexponential time algorithm that uses the now known deterministic PIT algorithms for constant-depth circuits \cite{LST21, AF22}. 

We also translate the $\vecx$ variables appropriately (again using deterministic PIT for constant-depth circuits) to ensure that we are working with polynomials that are monic in $z$ and each of the squarefree parts remains squarefree even when all the $\vecx$ variables are set to zero. 

We then invoke a univariate factorization algorithm to find a root of $P(\veczero, z)$, which for this discussion we assume is zero. Furthermore, we replace each $x_i$ with $x_i T$, for a fresh variable $T$, to work with a polynomial of the form $P(T, \vecx, z)$. Note that $P(0, \vecx, z) \in \F[z]$; we will refer to such polynomials as $T$-regularized.\footnote{More generally, we say that a polynomial in $\F[T, \vecx, z]$ is $T$-regularized if every monomial with $\vecx$-degree at least one in the polynomial has $T$-degree at least one. } 

The general plan is to use Newton iteration, starting with $\phi_0(T, \vecx) = 0$, to get ``approximate roots'' $\phi_{k-1}(T, \vecx)$ satisfying $\deg_T(\phi_{k-1}) < k$ and $\phi_{k-1}(T, \vecx) = \phi_0(T, \vecx) \bmod{T}$ such that $P(T, \vecx, \phi_{k-1}(T, \vecx)) = 0 \bmod{T^k}$. We can obtain a small circuit for $\phi_{k-1}$, although not necessarily one of constant depth. The hope is that, if $k$ is large enough, then we can recover from $\phi_{k-1}$ a true factor of $P(T, \vecx, z)$. 

\subsection{Restricting to roots}

Let us start by focusing on just extracting factors of $P(T, \vecx, z)$ of the form $(z - f(T, \vecx))$. If $f(T, \vecx) = \phi_0(T, \vecx) \bmod{T}$, then uniqueness of Newton Iteration guarantees that $\phi_{k-1}(T, \vecx) = f(T, \vecx) \bmod{T^{k}}$ as well. This therefore leads to the following natural algorithm:
\begin{enumerate}\itemsep 0pt
\item Compute a circuit for $\phi_{k-1}(T, \vecx)$ for $k > \deg_T(P)$, with $\deg_T(\phi_{k-1}) <  k$, via Newton Iteration. 
\item Check if $P(T, \vecx, \phi_{k-1}(T, \vecx)) = 0$. 
\end{enumerate}
If it were the case that $\phi_{k-1}$ was actually computable by a constant-depth circuit, then the circuit $P(T, \vecx, \phi_{k-1})$ is also constant-depth circuit and we can check for zeroness via known subexponential time PIT for constant-depth circuits. Unfortunately, we do not (yet) know if $\phi_{k-1}(T, \vecx)$ is computable by small constant-depth circuits. \\

The main insight is to notice that the circuit for $\phi_{k-1}$ that we obtain is structured enough that perhaps we can show that $P(T, \vecx, \phi_{k-1}(T, \vecx)) \stackrel{?}{=} 0$ can be tested via the same PIT nevertheless. We show that we can analyze the standard Kabanets-Impagliazzo generator instantiated with a suitable hard polynomial for such structured circuits and prove that it does preserve nonzeroness. The details can be found in \autoref{sec:root-preservation}. 

\subsection{Analyzing the KI generator on such circuits}

The main structural lemma that drives our analysis is the following that shows that ``low degree parts'' of approximate roots can indeed be computed via not-too-large constant-depth circuits. 

\begin{lemma*}[Informal version of \autoref{lem:low-deg-w-components-of-roots-mod-T}]
Suppose $\phi_{k-1}(T, \vecx)$ with $\deg_T \phi_{k-1} < k$ is an approximate $z$-root of $P(T, \vecx, z)$ as above. Then, for every integer $\ell \geq 0$, there is a circuit $C_\ell(T, \vecx)$ of constant depth and ``not-too-large'' size that agrees with $\phi_{k-1}$ on all monomials of degree at most $\ell$. That is,
\[
C_\ell(T, \vecx) = \phi_{k-1}(T, \vecx) \bmod \inangle{\vecx}^\ell,
\]
where ``not-too-large'' is $\poly(\size(P), \deg(P)) \cdot (\log k)^{\poly(\ell)}$. 
\end{lemma*}

While this is reminiscent of a lemma of \cite{ChouKS19} (\cref{lem:CKS-updated}) that allows one to argue that low-degree components of a root have small constant-depth circuits, there is an important difference in our statement: $\phi_{k-1}(T,\vecx)$ is a root modulo $T^k$, whereas we want a small constant-depth circuit to compute the root modulo $\inangle{\vecx}^{\ell}$ i.e. a \emph{different} set of variables. This difference demands a little more care (see \cref{rmk:why-quadratic-convergence}), although morally we still use a Taylor-expansion based approach along the lines of \cite{ChouKS19}.

With the above lemma, we would be able to show that the Kabanets-Impagliazzo generator will indeed preserve the nonzeroness of $P(T, \vecx, \phi_{k-1}(T, \vecx))$ when instantiated with a hard polynomial for constant-depth circuits. 

\begin{theorem*}[Informal version of \autoref{thm:technical-theorem-1}]
Let $f$ be a polynomial that requires ``large'' constant-depth circuits. Then, 
\[
    P(T, \vecx, \phi_{k-1}(T, \vecx)) \neq 0 \Longrightarrow P(T, \vecx, \phi_{k-1}(T, \vecx)) \circ \ki_{f} \neq 0
\]
where $\ki_{f}: \F[\vecx] \rightarrow \F[\vecw]$ is the Kabanets-Impagliazzo generator instantiated with the polynomial $f$. 
\end{theorem*}

We now briefly outline why this is true. Let us define $R(T,\vecx) := P(T, \vecx, \phi_{k-1}(T, \vecx))$. 
To prove the above lemma, assume on the contrary that the generator $\ki_{f}$ did not preserve the non-zeroness of $R(T, \vecx)$. By the standard hybrid argument (with additional substitutions), there are ``easy-to-compute'' polynomials $f_i$'s such that 
\begin{align*}
R(f_1,f_2, \ldots, f_{i-1}, x_i, f_{i+1}, \ldots, f_n, T) & \neq 0\\
\text{but }R(f_1,f_2, \ldots, f_{i-1}, f, f_{i+1},\ldots, f_n, T) & = 0.
\end{align*}
($f_{i+1}, \ldots, f_n$ are just field constants). Therefore, the polynomial $x_i - f$ divides the polynomial $R(f_1,\ldots, f_{i-1}, x_i, f_{i+1},\ldots, f_n, T)$. At this point, we make use of the following result by Chou, Kumar and Solomon (stated informally).

\begin{theorem*}[Chou-Kumar-Solomon (follows immediately from \autoref{lem:CKS-updated})] Let $R(\vecw, y)$ be a non-zero polynomial with $R(\vecw, g(\vecw)) = 0$ for some polynomial $g$ of degree at most $\ell$, with $\partial_y R (0, g(\veczero)) \neq 0$. \\
For all $i \leq \ell$, suppose there are small circuits $C_i(\vecw, y)$ such that 
\[
C_i(\vecw, y) = \partial_{y^i} R \bmod{\inangle{\vecw}^\ell}. 
\]
If $\ell$ is ``small'', then  $g$ is a ``small'' composition of the circuits $C_1,\ldots, C_\ell$. In particular, if $\ell$ is small and each $C_i$ is a small, constant-depth circuit, then $g$ is also computable by a ``small-ish'', constant-depth circuit.
\end{theorem*}

For our setting, let $\tilde{R}(\vecw, y, T) = R(f_1,f_2, \ldots, f_{i-1}, y, f_{i+1}, \ldots, f_n, T)$ which has $y - f(\vecw)$ as a factor. Note that this is not a constant-depth circuit as $R(T, \vecx) = P(T, \vecx, \phi_{k-1}(T, \vecx))$ and $\phi_{k-1}$ is not known to have a constant-depth circuits. However, if $\deg(f)$ is small enough, then the above theorem of Chou, Kumar and Solomon asserts that we can construct a constant-depth circuit for $(y - f(\vecw))$ if we are able to build constant-depth circuits for $\homog_{\leq \deg(f)}\inparen{\partial_{y^i} \tilde{R}(\vecw, y, T)}$. Fortunately, our main structural lemma states that low-degree homogeneous parts of $\phi_{k-1}$ do have not-too-large constant-depth circuits. Putting this together, we are able to assert that $f(\vecw)$ must also be computable by a not-too-large constant-depth circuit, thus contradicting the hardness of $f$. 

\subsection{Computing general factors}

For computing general factors, at a high level, the proof proceeds in three broad steps. We first show a way of characterizing irreducibility of polynomials by (exponentially many) divisibility tests involving the power series roots of these polynomials. We then show that these divisibility tests can be reduced to PIT instances for circuits that might not be constant-depth,  but have some additional structure. And finally, in spite of being unable to show that these PIT instances are for constant-depth circuits, we manage to show that the Kabanets-Impagliazzo hitting set generator \cite{KI04} invoked with low degree hardness for constant-depth circuits from \cite{LST21} preserves the non-zeroness of these PIT instances. This part of the proof again builds upon the techniques from \cite{ChouKS19}. We now discuss these steps in a bit more detail and refer to \autoref{sec:irreducibility-preservation} for the full proof.

 The main technical insight of our proof is that the Kabanets-Impagliazzo generator invoked with low-degree hard polynomials for constant-depth circuits  in fact preserves irreducibility of small constant-depth circuits and their factors. Let us consider the polynomial $P(\veczero, 0, z)$ over the algebraic closure $\overline{\Q}$, and let us assume this splits as $(z - \zeta_1) \cdots (z - \zeta_d)$. Let $\phi^{(1)}(T, \vecx), \ldots, \phi^{(d)}(T, \vecx)$ be the approximate root lifted from $\zeta_i$ --- i.e., it satisfies $P(T, \vecx, \phi^{(i)}) = 0\bmod T^k$ and $\phi^{(i)}(\veczero,0) = \zeta_i$, where $k$ is large enough. Then, each true factor $Q(T, \vecx, z)$ of $P(T, \vecx, z)$ corresponds to some subset $S_Q$ of these approximate roots. That is, we must have 
\[
Q(T, \vecx, z) = \prod_{i\in S_Q} (z - \phi^{(i)}(T, \vecx)) \trunc T^k
\]
where ``$\trunc T^k$'' denotes the operation of discarding all monomials that have degree $k$ or more in $T$. As a consequence, if $P(T, \vecx, z)$ is indeed irreducible, then for every subset $\emptyset \neq S \subsetneq [d]$ we have $Q_S(T, \vecx, z)$ does not divide $P(T, \vecx, z)$ as a polynomial, where
\[
    Q_S(T, \vecx, z) = \prod_{i\in S} (z - \phi^{(i)}(T, \vecx)) \trunc T^k \, .
\]

Turns out, the non-divisibility of such $Q_S, P$ can be expressed as an appropriate PIT via standard reductions from divisibility testing to PIT \cite{Forbes15, AW24}. Once again, this is not a PIT of a constant-depth circuit due to the presence of the $\phi^{(i)}(T, \vecx)$ sub-circuits which are not known to have small constant-depth circuits. Nevertheless, since we are able to show that $\phi^{(i)}(T, \vecx)$ has enough structure to allow us to build small constant-depth circuits for their ``low-degree'' components, we are able to argue that the Kabanets-Impagliazzo generator, when instantiated with a sufficiently hard low-degree polynomial, will preserve the non-zeroness of such PIT instances. Furthermore, our algorithms do not proceed by directly derandomizing these PIT instances since there are exponentially many of them to deal with. We only use this reduction in the analysis of our algorithms. 

Overall, this yields a variable reduction (to $n^\epsilon$ variables, for any $\epsilon > 0$) that allows us to preserve the factorization pattern of any constant-depth circuit. From this, recovering circuits (although not necessarily constant-depth) for the factors is relatively straightforward, given the known algorithms for polynomial factorization.  

\subsection*{Organization}
The rest of the paper is organized as follows. We begin by recalling some standard notations and  preliminaries in \autoref{sec:prelims}, and some preliminaries related to polynomial factorization in \autoref{sec:prelims-for-factorization}. We describe the details of our algorithms in \autoref{sec:algorithms}, and discuss their analysis, assuming some technical results in \autoref{sec:analysis-of-algorithms}. We discuss the proof of these technical results in detail in \autoref{sec:technical}, \autoref{sec:root-preservation} and \autoref{sec:irreducibility-preservation}. 

For readers familiar with the general area of algebraic circuit complexity and some familiarity with  polynomial factorization algorithms, we recommend skipping the preliminaries and going to \autoref{sec:algorithms}, and then referring to \autoref{sec:prelims} and \autoref{sec:prelims-for-factorization} as and when needed. 

\section{Standard preliminaries}\label{sec:prelims}

\subsection{Notation}
\begin{itemize}
    \item Throughout the paper $\F$ denotes a field and $\Q$ denotes the field of rational numbers. 
    \item We use letters like $x, y, z$ to refer to formal variables and letters like $a, b, c$ as constants, as would be clear from the context. Boldface letters $\vecx, \vecy, \vecz, \veca, \vecb$ etc. refer to tuples of such objects. The arity of the tuple is generally specified unless it is clear from the context.  
    \item We say that a multivariate polynomial $P$ is \emph{monic} in a specific variable $y$ if the coefficient of the highest degree monomial in $y$ in $P$ is a non-zero field constant. 
    \item An algebraic circuit over a field $\F$ on variables $\vecx$ is a directed acyclic graph with internal nodes being labeled by product $(\times)$ or sum $(+)$, and the leaves (nodes of in-degree zero) being labeled by variables in $\vecx$ or constants from $\F$. All the fan-ins are unbounded. 
    
    Such a circuit computes a polynomial in a natural sense - a leaf computes the polynomial that is equal to its label, a sum gate computes the sums of its inputs and a product gate computes the product of its inputs. The size of an algebraic circuit equals the number of edges in it and the depth equals the length of the longest path from a leaf  to an output node (a node of out-degree zero). We refer to the surveys \cite{SY10, S15} for detailed discussions on algebraic circuits.  

    \item Polynomial identity testing or PIT refers to the decision problem where the input is an algebraic circuit, and the goal is to decide if the polynomial computed by the circuit is identically zero. 

    \item For a polynomial $P$ and a variable $y$, $\deg_y(P)$ refers to the degree of $P$ with respect to $y$. Similarly, for a tuple $\vecx$ of variables, $\deg_{\vecx}(P)$ refers to the total degree of $P$ with respect to variables in $\vecx$. 

\end{itemize}

\subsection{Truncation}
We start with the important, although non-standard definitions that are helpful in succinctly expressing our technical statements in the paper. 
\begin{definition}[Polynomial truncation]
    \label{defn:truncation}
    For a polynomial $Q(\vecx) \in R[\vecx]$ and a positive integer $k$, $Q\trunc \inangle{\vecx}^k$ denotes the unique polynomial $\tilde{Q}$ with $\vecx$-degree less than $k$ such that $\tilde{Q} \equiv Q \bmod\inangle{\vecx}^k$. We extend this notation to $Q(T, \vecx)\in \F[\vecx][T]$ (as well as power-series in $\F[\vecx]\llbracket T \rrbracket$) and use $Q \trunc T^k$ to denote the truncation interpreting $Q(T, \vecx) \in R[T]$ (or $R\llbracket T \rrbracket$) for $R = \F[\vecx]$.
\end{definition}

While it is common to slightly overload notation and use $\bmod \inangle{\vecx}^k$ to also denote $\trunc \inangle{\vecx}^k$, we choose to separate these notations for reasons of clarity.

\subsection{Polynomial Identity Lemma}
We now recall the statement of the polynomial identity lemma. 
\begin{lemma}[Polynomial Identity Lemma \cite{O22,DL78,S80,Z79}]\label{lem:SZ}
Let $P$ be an $n$-variate non-zero polynomial of degree at most $d$ over a field $\F$. And, let $S$ be any subset of $\F$. 

Then, the number of zeroes of $P$ on the product set $S \times S \times \cdots \times S$ is at most $d|S|^{n-1}$. In particular, if $|S| > d$, then $P$ is non-zero on at least one point on $S \times S \times \cdots \times S$. 
\end{lemma}

\subsection{Interpolation}
We now recall the standard interpolation lemma for extracting coefficients of univariates from their evaluations. 
    \begin{lemma}[Interpolation (cf. {\cite[Lemma 5.3]{S15}})]
        Let $R$ be a commutative ring that contains a field $\F$ of at least $d+1$ elements, and let $\alpha_0,\ldots, \alpha_d$ be distinct field elements in $\F$. Then, for each $i \in \set{0,\ldots, d}$, there are constants $\beta_{i0}, \ldots, \beta_{id} \in \F$ such that for every $f(x) = f_0 + f_1x + \cdots + f_d x^d \in R[x]$ we have
        \[
        f_i = \sum_{j=0}^d \beta_{ij} \cdot f(\alpha_{j}).
        \]
    \end{lemma}
The following corollary invokes this lemma in some of the contexts that appear in our proof. The proof is immediate from the lemma.  
    \begin{corollary}[Standard consequences of interpolation]
        \label{cor:interpolation-consequences}
        Let $\alpha_0,\ldots, \alpha_d$ be distinct elements in  $\F$. Then, 
        \begin{enumerate}\itemsep0pt
        \item \textbf{[Partial derivatives]} If $C(\vecx, y)$ has degree $d$ in the variable $y$, then the $i$-th order partial derivative of $C$ with respect to $y$ can be expressed as an $\F[y]$-linear combination of $\setdef{C(\vecx, \alpha_j)}{j\in \set{0,\ldots, d}}$. That is, there are polynomials $\mu_0(y), \ldots, \mu_d(y)$ (not depending on $C$) of degree at most $d$ such that 
        \[
            \partial_{y^i} C(\vecx, y) = \mu_0(y) \cdot C(\vecx, \alpha_0) + \cdots + \mu_d(\vecy) \cdot C(\vecx, \alpha_d).
        \]
        \item \textbf{[Homogeneous components]} Let $C(\vecx)$ be a degree $d$ polynomial. Then, for any subset $\vecx_S \subseteq \vecx$ and any $i \in [d]$, the degree $i$ homogeneous part of $C$ with respect to $\vecx_S$, denoted by $\homog_{\vecx_S, i}(C)$, can be expressed as
        \[
            \homog_{\vecx_S, i}(C) = \sum_{j=0}^d \beta_{i,j} \cdot C(\alpha_j \cdot \vecx_S, \vecx_{\overline{S}})
        \]
        for some constants $\beta_{i,j} \in \F$ (not depending on $C$). 
        \item \textbf{[Truncation]} Let $C(\vecx)$ be a degree $d$ polynomial and let $\vecx_S \subseteq \vecx$ be a subset of variables. Then, for any $\ell \in [d]$, the truncation $C(\vecx) \trunc \inangle{\vecx_S}^{\ell+1}$ can be expressed as 
        \[
            C(\vecx) \trunc \inangle{\vecx_S}^{\ell+1} = \sum_{j=0}^d \gamma_{\ell, j} \cdot C(\alpha_j \cdot \vecx_S, \vecx_{\overline{S}})
        \]
        for some constants $\gamma_{\ell,j} \in \F$ (not depending on $C$). 
        \end{enumerate}
        In particular, if $C$ is computable by a size $s$, depth $\Delta$ circuit, then all of the above operations yield a circuit of size $\poly(d,s)$ and depth $\Delta + O(1)$. 
    \end{corollary}

\subsection{Strassen's division elimination}
We recall a classical theorem of Strassen for division elimination in algebraic circuits. 

\begin{theorem}[Division elimination in algebraic circuits \cite{Strassen1973}]\label{thm:div-elimination}
    Let $C$ be an algebraic circuit of size $s$ with division gates that computes a multivariate polynomial $P$ of degree $d$ over any sufficiently large field $\F$. 

    Then, there is an algebraic circuit $\Hat{C}$ of size at most $\poly(s,d)$ that computes the polynomial $P$ and does not have any division gates. 
\end{theorem}
As a matter of notation, algebraic circuits throughout this paper do not have division gates. We explicitly mention if we have to deal with circuits with division gates at any point in our proof. 

In our proofs, we rely on the following consequence of the above theorem. The proof easily follows from the standard proof of \autoref{thm:div-elimination}, for instance in \cite{SY10}. 
\begin{lemma}[Algorithmic division elimination]\label{lem:div-elimination-algorithm}
Let $\F$ be any field. 

There is a deterministic algorithm that takes as input algebraic circuits of size at most $s$ computing $n$ variate polynomials $A, B$ over $\F$, a point $\vecu \in \F^n$ and a degree parameter $d$ such that (a) $B$ divides $A$, (b) the quotient $A/B$ has degree at most $d$, and (c) $B(\vecu) \neq 0$, and outputs a (division free) algebraic circuit for the quotient $A/B$. Moreover, the algorithm runs in time $\poly(s,d)$. 
\end{lemma}

\subsection{Resultant and Discriminant}
The notion of resultant and its close connection to GCD of two univariate polynomials plays an important role in our proof (and in most polynomial factorization algorithms). We start by recalling the definition. 

\begin{definition}[Sylvester Matrix and Resultant]\label{def:resultant}
    Let $\F$ be any field, and let $P$ and $Q$ be univariates over $\F$ of degree equal to $a \geq 1$ and $b \geq 1$ respectively. 
    
    Let $\Gamma_{P,Q}: \F^{b} \times \F^{a} \to \F^{a+b}$ be the $\F$ linear map that maps a pair $(U, V)$ of univariates over $\F$ with degree of $U$ at most $(b-1)$ and degree of $V$ at most $(a-1)$ to the polynomial $(UP + VQ)$ of degree at most $(a+b-1)$.

    Then, the Sylvester matrix of $P$ and $Q$ is the $(a+b) \times (a+b)$ matrix for the $\F$-linear map $\Gamma_{P,Q}$, when the inputs and the outputs are represented as their coefficient vectors. 
    
    And, the Resultant of $P, Q$ is defined as the determinant of the Sylvester Matrix of $P$ and $Q$. 
\end{definition}
Clearly, the entries of the Sylvester matrix as defined above are the coefficients of $P$ and $Q$. 

In some applications in this paper, we end up invoking the notion of the resultant while working with multivariates $P$, $Q$. In these applications, we think of these multivariates as univariates in one of the variables, with the coefficients coming from the field of rational functions in the other variables. Unless otherwise clear from the context, we indicate the variable with respect to which these definitions are invoked. 

The resultant has a deep and extremely useful connection to the GCD of two polynomials as the following classical theorem indicates. We refer to Chapter 6 in the book \cite{GG13} for a proof. 

\begin{theorem}[Resultants and GCD {\cite[Corollary 6.20]{GG13}}]\label{thm:resultants-gcd}
        Let $\mathcal{R}$ be a unique factorization domain, and let $P,Q \in \mathcal{R}[z]$ be non-zero polynomials. Then, $\deg_z(\gcd(P,Q)) \geq 1$ if and only if $\Res{y}{P}{Q} = 0$, where $\gcd(P,Q) \in \mathcal{R}[z]$ and $\Res{z}{P}{Q} \in \mathcal{R}$. Moreover, there exist polynomials $A, B \in \mathcal{R}[z]$ satisfying $\Res{z}{P}{Q} = AP + BQ$. 
\end{theorem}
A special case of the resultant that is very natural to study is when we consider the resultant of a polynomial $P(z)$ and its derivative $\frac{dP}{dz}$. This resultant is referred to as the \emph{discriminant} of $P$, and unless the derivative vanishes for trivial reasons (for instance over fields of small characteristic), the discriminant captures the squarefreeness of $P$. More precisely, we have the following theorem. 
\begin{theorem}[Discriminant and squarefreeness {\cite[Lemma 12]{DSS22-closure}}]\label{thm:discriminant-squarefreeness}
        Let $\F$ be any field of characteristic zero, and let $P(z)$ be a univariate over $\F$ of degree at least one. Then, $P(z)$ is squarefree if and only if its discriminant $\Res{z}{P}{\frac{\partial P}{\partial z}}$ is non-zero. 
\end{theorem}

A beautiful result of Andrews and Wigderson shows that the resultant can be computed by a constant-depth circuit.
\begin{theorem}[Computing resultants via constant-depth circuits {\cite[Theorem 6.1]{AW24}}]\label{thm:resultant-constant-depth-andrews-wigderson}
    Let $\F$ be a field of characteristic zero. For a fixed $\Delta \in \N$, there is a family of depth-$\Delta$ circuits $\{C_n\}_{n\in \N}$ with size $\leq \poly(n)$ such that the following is true. If $f, g \in \F[z]$ are degree-$n$ univariate polynomials given by their coefficients, then $C_n$ takes the coefficients of $f$ and $g$ as input, and computes the resultant $\Res{z}{f}{g}$.
\end{theorem}

\subsection{Gauss' lemma}
The following basic lemma of Gauss is useful for us in our proof. 
\begin{lemma}[Gauss' lemma {\cite[Section 6.2, Corollary 6.10]{GG13}}]\label{lem:gauss}
Let $\mathcal{R}$ be a unique factorization domain with the field of fractions $\mathcal{K}$. Then, a monic polynomial $P(z)$ is irreducible in $\mathcal{R}[z]$ if and only if it is irreducible in $\mathcal{K}[z]$. \\
In particular, for any monic $P(z)$, the factorization of $P(z)$ into its irreducible factors in $\mathcal{R}[z]$ is identical to the factorization of $P(z)$ into its irreducible factors in $\mathcal{K}[z]$.
\end{lemma}
Moreover, for a monic $P(z) \in \mathcal{K}[z]$, its factors can also be assumed to be monic in $z$ without loss of generality.

\subsection{The Kabanets-Impagliazzo generator}
The Kabanets-Impagliazzo hitting set generator \cite{KI04} for algebraic circuits is an adaptation of the Nisan-Wigderson generator in classical complexity to the algebraic setting. It allows us to obtain non-trivial derandomization for PIT for algebraic circuits from sufficiently hard explicit polynomial families. More precisely, Kabanets \& Impagliazzo showed in \cite{KI04} that given an explicit polynomial family that requires exponential size algebraic circuits, there is a deterministic algorithm for PIT for algebraic circuits (with size and degree polynomially bounded in the number of variables) that runs in quasipolynomial time.  

This generator and the details of its analysis play an important role in the proofs in this paper. We start by recalling the notion of combinatorial designs, and then define the generator. 
\begin{definition}[Combinatorial designs]\label{def:designs}
Let $n, \sigma, \mu, \rho \in \N$. A family of subsets $\mathcal{S} = (S_1, \dots, S_n)$ is a $(n,\sigma,\mu,\rho)$ design if:
\begin{itemize}
    \item $\forall i\in[n]: \, S_i \subseteq [\mu]$
    \item $\forall i\in[n]: \, |S_i| = \sigma$
    \item For any $i,j \in [n]$ s.t. $i \neq j:$ $|S_i \cap S_j| < \rho$
\end{itemize}
\end{definition}
The following theorem gives an explicit construction of such designs. 
\begin{lemma}[Explicit construction of designs \cite{NW94}] \label{lem:explicit-designs}  For any positive integers $n,\sigma$ with $n<2^{\sigma}$, there exists an explicit $(n,\sigma,\mu,\rho)$-design with $\mu=\displaystyle\frac{\sigma^2}{\log n}$ and $\rho=\log n$. 

Moreover, each subset inside the design can be computed in $\poly(n,2^{\mu})$ time deterministically.

\end{lemma}

\begin{definition}[KI-generator $\ki_{g,\mathcal{S}}$ {\cite{KI04}}]\label{def:KI-generator}
    Let $\F$ be a field, $n,\sigma$ be positive integers with $n<2^\sigma$ and $g(\vecx)$ be a $\sigma$-variate polynomial. Let $\mathcal{S} = (S_1,\dots,S_n)$ be an explicit $(n,\sigma,\mu,\rho)$-design from \cref{lem:explicit-designs}. The Kabanets-Impagliazzo generator given by $g$ and $\mathcal{S}$ is a polynomial map $\ki_{g,\mathcal{S}}: \F^\mu \to \F^n$ defined by $\vecw \mapsto (g_1(\vecw), \dots, g_n(\vecw))$, where for each $i\in [n]$, $g_i(\vecw) := g(\vecw_{S_i})$ for $\vecw_{S_i}:= \{w_j: j \in S_i\}$. Thus, $\ki_{g,\mathcal{S}}$ also defines a homomorphism from $\F[x_1, \dots, x_n]$ to $\F[w_1, \dots, w_\mu]$ that maps each $x_i$ to $g_i(\vecw)$.
\end{definition}

\subsection{Lower bounds and PIT for constant-depth circuits}
We now recall the results of \cite{LST21} that prove superpolynomial lower bounds for constant-depth circuits. 
\begin{theorem}[Lower bounds for constant-depth circuits {\cite[Corollary 4]{LST21}}]
    \label{LST-hardness}  Suppose $n,d \in \N$ with $d \leq \log n/100$ and $\F$ is a field with $\operatorname{char}(\F) = 0$ or greater than $d$. Then, for any product-depth $\Delta\in \N$, there exists an explicit $n$-variate polynomial $P(x_1, \dots, x_n) \in \F[x_1,\dots,x_n]$ of degree $d$ such that any algebraic circuit of product-depth at most $\Delta$ must have size at least $n^{d^{\exp(-O(\Delta))}}$. 
\end{theorem}
In \cite{ChouKS19}, it was shown that explicit \emph{low degree} hard polynomials for constant-depth circuits as shown in the above theorem imply non-trivial deterministic PIT algorithms for constant-depth algebraic circuits. An alternative route to achieving the same result was shown later by Andrews \& Forbes. We recall this theorem below.  
\begin{theorem}[Subexponential time deterministic PIT for constant-depth circuits \cite{LST21, AF22}] \label{thm:lst-PIT}
    Let $\varepsilon>0$ be a real number and $\F$ a field of characteristic 0. Let $C$ be an algebraic circuit of size $s \leq \poly(n)$ and depth $\Delta = o(\log \log \log n)$, computing an $n$-variate polynomial . Then, there is a deterministic algorithm that can decide whether the polynomial computed by $C$ is identically zero or not, in time ${(s^{\Delta+1}\cdot n)^{O(n^{\varepsilon})}}$.
\end{theorem}

\section{Preliminaries for polynomial factorization}\label{sec:prelims-for-factorization}

\subsection{Regularized polynomials}

\begin{definition}[$T$-regularized polynomials and non-degenerate, truncated, approximate $z$-roots of order $k$]
    \label{defn:t-regularized-non-degenerate-approximate-root}

    Let $P(T, \vecx, z)$ be a polynomial in $\F[T, \vecx, z]$ that is monic in the variable $z$, and $\Phi(T, \vecx) \in \F[\vecx]\indsquare{T}$ be a power series. 
    \begin{itemize}\itemsep 0pt
        \item $\Phi(T,\vecx)$ is said to be an \emph{approximate $z$-root of order $k$ with respect to $T$} if 
        
        \[
        P(T, \vecx, \Phi(T, \vecx)) \equiv 0 \bmod{T^k}.
        \]
        Moreover, an approximate $z$-root $\Phi$ of order $k$ is \emph{truncated} if $\deg_T \Phi < k$. Throughout the paper, approximate roots will be defined modulo $T^k$. For the sake of brevity, we sometimes refer to an approximate $z$-root of order $k$ with respect to $T$ as just an approximate $z$-root of order $k$. 
        
        \item $P(T,\vecx,z)$ is \emph{$T$-regularized with respect to $z$} if $P(0, \vecx, z)  \in \F[z]$, that is, every monomial that depends on $\vecx$ is divisible by $T$. 
        \item $\Phi(T,\vecx)$ is a \emph{non-degenerate} approximate $z$-root of $P(T,\vecx,z)$ if 
        \[
        (\partial_z P)(0, \veczero, \alpha) \neq 0
        \]
        where $\alpha = \Phi(0, \veczero) \in \F$; that is, the constant term of $\Phi$ is not a repeated root of $P(0, \veczero, z)$. 
        \qedhere
    \end{itemize}
\end{definition}

The following simple observation follows immediately from the above definitions.

\begin{observation} \label{obs:approx-root-simple-obs}
    Let $P(T, \vecx, z)$ be a polynomial in $\F[T, \vecx, z]$ that is monic in the variable $z$, $T$-regularized, and let $\Phi(T, \vecx) \in \F[\vecx]\indsquare{T}$ be a power series. \\
    If $\Phi(T, \vecx)$ is an approximate $z$-root of $P(T,\vecx,z)$, then $\Phi(0, \vecx)$ -- a root of the univariate $P(0, \vecx, z) \in \F[z]$ --  is a scalar in $\F$, implying that $\Phi(0,\vecx) = \Phi(0,\veczero)$. Moreover, $(\partial_z P)(0, \veczero,z) = (\partial_z P)(0, \vecx, z)$.
\end{observation}

\subsection{Squarefree decomposition}
We say that a polynomial $P$ is \emph{squarefree} if it is not divisible by the square of another polynomial. In particular, every irreducible factor of $P$ appears with multiplicity one in the unique factorization of $P$.

We now define the notion of squarefree decomposition of a polynomial. 
\begin{definition}[Squarefree decomposition]
    Let $F \in \F[\vecx]$ be a polynomial such that $F(\vecx) = \prod_{i=1}^m G_i(\vecx)^{e_i}$. Let $r = \max_{i\in [m]} e_i$, where each $G_i$ is irreducible. Then the squarefree decomposition of $F$ is $(F_1, F_2, \dots, F_r)$, where for each $i\in [r]$, $F_i := \prod_{j\in[m]: e_j = i} G_j$.  
\end{definition}
The following theorem of Andrews \& Wigderson shows that squarefree parts of a polynomial computable by a small constant-depth circuit have small constant-depth circuits, and moreover, we can compute such a decomposition deterministically given an appropriate PIT oracle. 
\begin{theorem}[Squarefree decomposition \cite{AW24}]\label{thm:andrews-wigderson-squarefree-decomposition}
    Let $\F$ be a field of characteristic zero or characteristic greater than $D$. Let $\mathcal{O}$ be an oracle that solves polynomial identity testing for constant-depth circuits. Then, there is a deterministic polynomial-time algorithm with oracle access to $\mathcal{O}$ which does the following:
    \begin{enumerate}
        \item The algorithm receives as a input a constant-depth circuit that computes a polynomial $F$ of degree $D$.
        \item The algorithm outputs a collection of constant-depth circuits $C_1, \dots, C_r$ such that $C_i$ computes $F_i$, where $(F_1, \dots, F_r)$ is the squarefree decomposition of $F$.
    \end{enumerate}
\end{theorem}

\subsection{Newton iteration}
We now recall various flavors of Newton iteration that we use in our proofs. 
\begin{lemma}[Newton iteration with linear convergence {\cite[Lemma 5.1]{ChouKS19}}]
    \label{lem:newton-iteration-linear}
    Let $R = \F[\vecx]$ be a polynomial ring, and let $H(\vecx, z) \in R[z]$. 
    Suppose $\phi \in \F\llbracket\vecx\rrbracket$ is a power-series such that $H(\vecx, \phi) = 0 \bmod \inangle{\vecx}^{m}$ and $\partial_z H (\veczero, \phi(\veczero)) \neq 0$. Then, \[\phi' := \phi - \frac{H(\vecx, \phi)}{\partial_z H(\veczero,\varphi(\veczero))}\] satisfies $H(\vecx, \phi') = 0 \bmod{\inangle{\vecx}^{m+1}}$ and $\phi' = \phi \bmod{\inangle{\vecx}^{m}}$. Furthermore, such an extension $\phi'$ of $\phi$ is unique in the sense that any $\phi''$ that satisfies $H(\vecx, \phi'') = 0\bmod{\inangle{\vecx}^{m+1}}$ and $\phi'' = \phi \bmod{\inangle{\vecx}^{m}}$ must satisfy
    \[
    \phi' = \phi'' \bmod{\inangle{\vecx}^{m+1}}.
    \]
\end{lemma}

\begin{corollary}[{\cite[Corollary 5.5]{ChouKS19}}, {\cite[Lemma 3.1]{KRSV}}]
    \label{lem:newton-iteration-linear-circuit-and-uniqueness}
    Let $R = \F[\vecx]$ be a polynomial ring, and let $H(\vecx, z) \in R[z]$ be a polynomial of degree $D$ and a circuit of size $s$. Suppose $u \in \F$ such that:
    \begin{align*}
        H(\veczero,u) &= 0 \\
        \frac{\partial H}{\partial z}(\veczero,u) &\neq 0
    \end{align*}
    Then for every $k \in \N$, there is a unique truncated, non-degenerate, approximate $z$-root $\Phi_k(\vecx)$ of order $k$ with respect to $\vecx$ for $H(\vecx,z)$, satisfying $\Phi_k(\veczero) = u$. Moreover, there is a deterministic algorithm that runs in time $\poly(s,D,k)$ and outputs a circuit of size $\poly(s,D,k)$ for $\Phi_k(\vecx)$. 
\end{corollary}

\begin{lemma}[Newton iteration with quadratic convergence {\cite[Lemma 9.21, Lemma 9.27]{GG13}}]
    \label{lem:newton-iteration-quadratic-generic}
    Let $R = \F[\vecx]$ be a polynomial ring, and let $H(\vecx, z) \in R[z]$. 

    Suppose $\phi \in \F\llbracket\vecx\rrbracket$ is a power-series such that $H(\vecx, \phi) = 0 \bmod \inangle{\vecx}^{m}$ and $\partial_z H (\veczero, \phi(\veczero)) \neq 0$. Then, for any power series $\sigma \in \F\llbracket\vecx\rrbracket$ satisfying
    \begin{align*}
        \sigma(\vecx) & = \frac{1}{\partial_z H(\vecx, \phi)} \bmod{\inangle{\vecx}^{m}}
    \end{align*}
    we have that $\phi' := \phi - H(\vecx, \phi)\sigma$ satisfies $H(\vecx, \phi') = 0 \bmod{\inangle{\vecx}^{2m}}$ and $\phi' = \phi \bmod{\inangle{\vecx}^{m}}$. Furthermore, such an extension $\phi'$ of $\phi$ is unique in the sense that any $\phi''$ that satisfies $H(\vecx, \phi'') = 0\bmod{\inangle{\vecx}^{2m}}$ and $\phi'' = \phi \bmod{\inangle{\vecx}^{m}}$ must satisfy
    \[
    \phi' = \phi'' \bmod{\inangle{\vecx}^{2m}}.
    \]
\end{lemma}

\begin{lemma}[Quadratic-convergence Newton Iteration without divisions]
    \label{lem:newton-iteration-quadratic-div-free}
    Let $R = \F[\vecx]$ be a polynomial ring, and let $H(\vecx, z) \in R[z]$. Suppose there exists $\alpha \in \F$ such that $H(\veczero, \alpha) = 0$ and $\partial_z H (\veczero, \alpha) = \beta \neq 0$. For each $i \geq 0$, define polynomials $\phi_i, \sigma_i \in \F[\vecx]$ as follows:
    \begin{align*}
    \phi_0 & := \alpha, & \sigma_0 &:= (1/\beta),\\
    \text{For $i \geq 0$,}\quad \phi_{i+1} &:= \phi_i - H(\vecx, \phi_i) \cdot \sigma_i, & \sigma_{i+1} &:= 2\sigma_i - \sigma_i^2 \cdot \partial_zH (\vecx, \phi_{i+1}).
    \end{align*}
    Then $H(\vecx, \phi_i) = 0\bmod{\inangle{\vecx}^{2^i}}$, $\phi_{i+1} = \phi_i \bmod{\inangle{\vecx}^{2^i}}$, and $\sigma_i \cdot \partial_z H(\vecx,\varphi_i) = 1 \bmod \inangle{\vecx}^{2^i}$. 
    
\end{lemma}
\begin{proof}
    We prove this by induction on $i$. The base case $i=0$ follows by definition.\\
    Now suppose that for some $i \geq 0$,  $H(\vecx, \phi_i) = 0\bmod{\inangle{\vecx}^{2^i}}$ and ${\sigma_i \cdot \partial_z H(\vecx,\varphi_i)} = 1\bmod \inangle{\vecx}^{2^i}$. \\
    By \cref{lem:newton-iteration-quadratic-generic}, it follows that $\varphi_{i+1}:= \phi_i - H(\vecx, \phi_i) \cdot \sigma_i$ satisfies $H(\vecx,\varphi_{i+1}) = 0 \bmod \inangle{\vecx}^{2^{i+1}}$ and $\varphi_{i+1} = \varphi_i \bmod \inangle{\vecx}^{2^i}$. 
    \begin{align*}
        \sigma_{i+1} \cdot \partial_z H(\vecx, \phi_{i+1}) - 1 & = (2 \sigma_i - \sigma_i^2 \cdot \partial_z H(\vecx, \phi_{i+1})) \cdot \partial_z H(\vecx, \phi_{i+1}) - 1 \\
        & = \inparen{\sigma_i \cdot \partial_z H(\vecx, \phi_{i+1}) - 1} - \inparen{(\sigma_i \cdot \partial_z H(\vecx, \phi_{i+1}))^2 - \sigma_i\cdot \partial_z H(\vecx, \phi_{i+1})}\\
        & = - \inparen{\sigma_i \cdot \partial_z H(\vecx, \phi_{i+1}) - 1}^2\\
        & = 0 \bmod{\inangle{\vecx}^{2^{i+1}}}
    \end{align*}
    where the last equality follows because $\varphi_{i+1} = \varphi_i \bmod \inangle{\vecx}^{2^i}$ and $\sigma_{i} \cdot \partial_z H(\vecx, \phi_{i}) - 1 = 0\bmod{\inangle{\vecx}^{2^i}}.$
\end{proof}

\subsection{A lemma of Chou, Kumar \& Solomon}
The following technical lemma of Chou, Kumar \& Solomon \cite{ChouKS19} is used in the analysis of Kabanets-Impagliazzo generator for constant-depth circuits. The lemma, both as a blackbox and the technical ideas therein are important for our proofs. 
\begin{lemma}[Lemma 5.2 and Lemma 5.3 in \cite{ChouKS19}] \label{lem:CKS-updated}
    Let $P\in \F[\vecx,y]$ and let $R(\vecx)$ be polynomials such that $R$ is of degree at most $d$, $P(\vecx, R(\vecx)) \equiv 0$ and $\frac{\partial P}{\partial y}(\mathbf{0}, R(\mathbf{0}))$ is non-zero. 

    Then, there exists a $(d+1)$-variate polynomial $Q(\vecz)$ of degree at most $d$ such that 
    \[
        R(\vecx) \equiv  Q(h_0(\vecx), h_1(\vecx), \ldots, h_d(\vecx)) \mod \langle \vecx \rangle^{d+1} \, ,
    \]
    where for every $i \in \{0, 1, \ldots, d\}$, $h_i(\vecx)$ is defined as 
    \[
        h_i(\vecx) := \frac{\partial P}{\partial y^j}(\vecx, R(\mathbf{0})) - \frac{\partial P}{\partial y^j}(\mathbf{0}, R(\mathbf{0}))  \trunc \langle \vecx \rangle^{d+1} \, .
    \]
\end{lemma}
For our proofs, we end up invoking \autoref{lem:CKS-updated} in settings where the polynomial $P$ also depends on an additional variable $T$ (whereas $R$ does not). In this case, the quantity $\frac{\partial P}{\partial y}(T,\mathbf{0}, R(\mathbf{0}))$ potentially depends on the variable $T$. Since our root $R$ does not depend on the variable $T$, we can set $T$ to some field constant without disturbing the starting conditions for Newton iteration, and thus we can perform Newton iteration as usual. 

\begin{lemma}\label{lem:CKS-updated-rational-in-T-new}
    Let $P\in \F[T,\vecx,y]$ and let $R(\vecx)$ be polynomials such that $R$ is of degree at most $d$, $P(T, \vecx, R(\vecx)) \equiv 0$ and $\frac{\partial P}{\partial y}(T, \mathbf{0}, R(\mathbf{0}))$ is non-zero. 

    Then, there exists a $\kappa \in \F$ and a $(d+1)$-variate polynomial $Q(\vecz) \in \F[\vecz]$ of degree at most $d$ such that 
    \[
        R(\vecx) \equiv  Q(h_0(\kappa,\vecx), h_1(\kappa,\vecx), \ldots, h_d(\kappa,\vecx)) \mod \langle \vecx \rangle^{d+1} \, ,
    \]
    where for every $i \in \{0, 1, \ldots, d\}$, $h_i(\vecx)$ is defined as 
    \[
        h_i(T,\vecx) := \frac{\partial P}{\partial y^i}(T, \vecx, R(\mathbf{0})) - \frac{\partial P}{\partial y^i}(T, \mathbf{0}, R(\mathbf{0}))  \trunc \langle \vecx \rangle^{d+1} \, .
    \]
\end{lemma}
\begin{proof}
    Since $\frac{\partial P}{\partial y}(T,\veczero,R(\veczero)) = \delta(T)\in \F[T]$ for some $\delta(T) \not\equiv 0$, there exists a $\kappa$ such that $\delta(\kappa) = \delta_0 \neq 0$. Thus, the polynomial $\tilde{P}(\vecx,y) := P(\kappa,\vecx,y)$ satisfies $\tilde{P}(\vecx,R(\vecx)) \equiv 0$ and $\frac{\partial \tilde{P}}{\partial y}(\veczero,R(\veczero)) = \delta_0 \neq 0$, for some $\delta_0 \in \F$. Applying \Cref{lem:CKS-updated} on $\tilde{P}$ and $R(\vecx)$ tells us that there exists a $(d+1)$-variate polynomial $Q(\vecz)$ of degree at most $d$ such that 
    \[
        R(\vecx) \equiv  Q(h_0(\vecx), h_1(\vecx), \ldots, h_d(\vecx)) \mod \langle \vecx \rangle^{d+1} \, ,
    \]
    where for every $i \in \{0, 1, \ldots, d\}$, $h_i(\vecx)$ is defined as 
    \[
        h_i(\vecx) := \frac{\partial \tilde{P}}{\partial y^j}(\vecx, R(\mathbf{0})) - \frac{\partial \tilde{P}}{\partial y^j}(\mathbf{0}, R(\mathbf{0}))  \trunc \langle \vecx \rangle^{d+1} \, .
    \] 
    Since $\tilde{P}(\vecx,y) := P(\kappa,\vecx,y)$, the required statement follows.
\end{proof}
For most of this paper, the only operation we are allowed on the $T$-variable is scaling by a field element since we care about roots $\bmod T^k$, and we want to preserve the $T$-degree of monomials during any such operations/substitutions. But the above lemma will be invoked at a point when we are concerned about roots $\mod \inangle{\vecx}^d$, not $\bmod T^k$, which is why it will be okay to replace $T$ by a field element $\kappa$. 

\subsection{Deterministic factorization}
For our proof, we need the following classical theorem of Lenstra, Lenstra and Lovasz. 
\begin{theorem}[Factorizing polynomials with rational coefficients \cite{LLL82, GG13}]\label{thm:-LLL-univariate-factorization}
  Let $P \in \Q[x]$ be a monic polynomial of degree $d$. Then there is a deterministic algorithm computing all the irreducible factors of $P$ that runs in time $\poly(d, t)$, where $t$ is the maximum bit-complexity of the coefficients of $f$. 
\end{theorem}

For our proofs, we also rely on a deterministic algorithm for factoring $n$ variate degree $d$ polynomials that run in time $d^{O(n)}$. This is implicit in many known algorithms for polynomial factorization, for instance, in the results of Kopparty, Saraf, Shpilka \cite{KSS15}. Such a statement can also be inferred from our proofs in this paper. We recall a formal statement of this nature from a work of Lecerf below.  
\begin{theorem}[{\cite[Proposition 4]{lecerf2007}}]\label{thm:dense-rep-factorization}
    Suppose $P(x_1, \dots, x_n,y) \in \Q[\vecx,y]$ be a polynomial that is monic in $y$ with total degree $d$. Further, suppose that $P$ is squarefree and $P(\veczero,y)$ is squarefree. Then, there is a deterministic algorithm that takes $P$ as input in the dense representation and outputs each of its irreducible factors in time $\leq O(N^2)$, where $N = \binom{n+d+1}{n}$.
\end{theorem}
Intuitively, the proof of the theorem follows from standard Hensel Lifting or Newton Iteration based algorithms for polynomial factorization, and observing (as formally done in \cite{KSS15}) that at every stage of the algorithm, randomness is needed only for polynomial identity testing. Moreover, these PIT instances are all polynomials of degree $\poly(d)$ and on $O(n)$ variables, and hence by \autoref{lem:SZ} can be solved deterministically in time $d^{O(n)}$. 

\section{Algorithm} \label{sec:algorithms}

\subsection{Main theorems and high-level overview}

First, we describe our main structural result, which states that the Kabanets-Impagliazzo hitting-set generator, when instantiated with a sufficiently hard low-degree polynomial and an appropriate combinatorial design, preserves the irreducibility of the factors of constant-depth circuits. The precise statement requires a few more conditions, and these conditions are without loss of generality.

Throughout the paper, we will use $\mathcal{G}=\{g_m\}_{m\in \N}$ to denote a family of explicit polynomials such that for every $m\in \N$, $g_m \in \F[x_1, \dots, x_m]$, $d_m := \deg(g_m) \leq O(\log \log (m))$. Further, $\mathcal{G}$ has the property that for any depth $\Delta \in \N$, if $\mathcal{C} = \{C_m\}_{m\in \N}$ is a family of depth-$\Delta$ circuits computing $\mathcal{G}$, then $\mathcal{C}$ requires size $m^{d_m^{\exp(-O(\Delta))}}$, which is $m^{\omega(1)}$. \cref{LST-hardness} gives us such a family of explicit low-degree polynomials that are hard for constant-depth circuits.

\begin{restatable}[Irreducibility-preserving variable reduction]{theorem}{irredpreserve}\label{thm:irreducibility-preservation}
    Fix any $\Delta \in \N$ and $\varepsilon \in (0,0.5)$. For an absolute constant\footnote{If $a(n) = O(b(n))$, then there exists some $C$ such that for all $n>C$, $a(n) \leq b(n)$; the $C_{\Delta,\varepsilon}$ in our statement is for this purpose. The precise value of $C_{\Delta,\varepsilon}$ depends on the exact hardness of the polynomial in $\mathcal{G}$ and the upper bounds obtained in our proofs.} $C_{\Delta,\varepsilon} \in \N$, let $n \in \N, n \geq C_{\Delta,\varepsilon}$ and $\vecx:=(x_1, \dots, x_n)$. Let $P(T, \vecx, z)$ be a nonzero polynomial with the following properties.
    \begin{itemize}
        \item $P(T,\vecx,z)$ is computable by a size $s \leq \poly(n)$ and depth $\Delta$ circuit.
        \item $P(T,\vecx,z)$ is monic in $z$ and $T$-regularized, with $\deg(P) = D \leq \poly(n)$.
        \item $P(T,\vecx,z)$ and $P(0,\vecx,z) = P(0,\veczero,z)$ are squarefree.
    \end{itemize}
    Let $\sigma = O(n^\varepsilon),\mu = O(\frac{n^{2\varepsilon}}{\log(n)}),\rho = O(\log(n))$, and let $\mathcal{S}$ be an $(n,\sigma ,\mu ,\rho)$-design. Let $\ki_{g_{\sigma},\mathcal{S}}: \F^{\mu} \to \F^{n}$ be the polynomial map in \autoref{def:KI-generator} defined using the design $\mathcal{S}$ and the polynomial $g_{\sigma}$ from the family of hard polynomials  $\mathcal{G}$.
    Then, the following is true. 

    A polynomial $F(T,\vecx,z)$ is an irreducible factor of $P(T,\vecx,z)$ if and only if $F(T,\ki_{g_\sigma,\mathcal{S}}(\vecw), z)$ is an irreducible factor of $P(T,\ki_{g_\sigma,\mathcal{S}}(\vecw),z)$.
\end{restatable}

A few remarks regarding the choice of parameters in \cref{thm:irreducibility-preservation}:
\begin{itemize}
    \item The family $\mathcal{G}$ has degree $d_m \leq O(\log \log (m))$ so that $(\log(m))^{\poly(d_m)}$ is $\poly(m)$. We can work with any $d_m \leq O((\log m)^{\alpha})$ for some small enough $\alpha$ depending on the exponent in $\poly(d_m)$, but $\log \log (m)$ works in every case and makes it simpler to state the theorems.
    \item The design $\mathcal{S}$ has $\sigma = O(n^\varepsilon)$ because $\mathcal{G}$ is guaranteed to be superpolynomially hard for constant-depth circuits. Stronger hardness guarantees can be used with smaller $\sigma$ to get algorithms with better time complexity, as is usually the case in hardness-vs-randomness results.
\end{itemize}

We now describe our main algorithmic result. Informally, the result states that for every choice of depth $\Delta \in \N$, there is a \emph{deterministic} algorithm $\mathcal{A}_{\Delta}$ such that $\mathcal{A}_{\Delta}$ takes a depth-$\Delta$ circuit for a polynomial $P$ as input, and outputs small (but unbounded depth) circuits along with multiplicity information for each irreducible factor of $P$. Moreover, $\mathcal{A}_{\Delta}$ runs in time subexponential in the input size. 

\begin{restatable}[Deterministic subexponential time algorithm for factorization of constant-depth circuits]{theorem}{mainalgotheorem}\label{thm:main-algorithm}
Fix any $\Delta \in \N$ and $\varepsilon \in (0,0.5)$. There exists an algorithm $\mathcal{A}_{\Delta, \varepsilon}$ which, for all sufficiently large $n$, 
\begin{itemize}
    \item takes as input a polynomial $P(\vecx) \in \Q[x_1, \dots, x_n]$ of degree $D \leq \poly(n)$ with a depth-$\Delta$, size $s\leq \poly(n)$ circuit;
    \item outputs $\poly(s,D)$-sized circuits for each irreducible factor of $P$, along with the multiplicity of each such factor; and
    \item runs in time $\poly(s, D)^{O(n^{2\varepsilon})}$.
\end{itemize}    
\end{restatable} 
\subsubsection{High-level overview of the algorithm}\label{sec:algo-overview}
\begin{enumerate}
    \item Suppose $P(\vecx)$ is the polynomial that we would like to factor. We first use the algorithm by Andrews and Wigderson (\cref{thm:andrews-wigderson-squarefree-decomposition}) to compute the squarefree decomposition of $P$. Now, we deal with each squarefree part separately.
    \item For a specific squarefree part $P_r(\vecx)$, we perform some of the standard transformations ($x_i \mapsto T\cdot x_i + a_i \cdot z + b_i$) so that the polynomial $P_r(T,\vecx,z)$ is monic in $z$ variable, $T$-regularized, and $P_r(0,\vecx,z) = P_r(0,\veczero,z)$ is squarefree.
    \item We apply $\ki_{g,\mathcal{S}}(\vecw)$ on $P_r(T,\vecx,z)$ (instantiated with an appropriately chosen low-degree hard polynomial $g$ and a design $\mathcal{S})$. This maintains the irreducibility of each factor (\cref{thm:irreducibility-preservation}). At this point, we are working with an $n^\varepsilon$-variate polynomial for some $\varepsilon\in (0,1)$. Thus, we can use a brute-force / dense-representation factorization algorithm to factorize $P_r(T,\ki_{g,\mathcal{S}}(\vecw),z)$ in subexponential time (\cref{thm:dense-rep-factorization}). 
    \item Since $P_r(0,\vecx,z) = P_r(0,\ki_{g,\mathcal{S}}(\vecw),z) = P_r(0,\veczero,z)$, taking an irreducible factor of $P_r(T,\ki_{g,\mathcal{S}}(\vecw),z)$ and setting $T$ to 0 precisely tells us $G(0,\vecx,z)$ for each irreducible factor $G(T,\vecx,z)$ of $P_r(T,\vecx,z)$. Thus, we now have $G(0,\vecx,z) = G(0,\veczero,z)$ for each $G$ that is an irreducible factor of $P$, and we can deal with each $G$ separately.\footnote{At this point, with the right univariate projections $G(0,\veczero,z)$ for each irreducible factor $G(T,\vecx,z)$, one way to get the irreducible factors is to perform \emph{Hensel lifting} (for instance, see \cite[Lecture 7]{Sudan98} or \cite{ST20}). While Hensel lifting usually ends with a \emph{reconstruction step} that involves solving a linear system, the clean-up for us would just be a truncation since Hensel lifting guarantees that if we start off with the right univariate projection, we will retrieve the right factor, modulo higher degree terms that might have accumulated in the process.}
    \item Suppose $G(T,\vecx,z)$ is an irreducible factor, and we have access to the univariate $G(0,\vecx,z) = G(0,\veczero,z)$. We factorize $G(0,\veczero,z)$ using \cref{thm:-LLL-univariate-factorization}. $G(0,\veczero,z)$ could have a linear factor of the form $(z-\alpha)$; in this case, we use Newton iteration (\cref{lem:newton-iteration-linear-circuit-and-uniqueness}) to lift the root $\alpha$ to a truncated, approximate $z$-root $\varphi(T,\vecx)$ of $P(T,\vecx,z)$ of sufficiently high accuracy. But all the factors of $G(0,\veczero,z)$ might be non-linear. In this situation, we artificially \emph{add a root} $u$ of $G(0,\veczero,z)$ to the field; more precisely, if $H(z)$ is an arbitrary irreducible factor of $G(0,\veczero,z)$, we work over the field $\frac{\Q[u]}{H(u)}$ so that $u$ is now a root of $G(0,\veczero,z)$. We can efficiently simulate arithmetic in this field. Thus, we can lift the root $u$ to a truncated, approximate $z$-root $\varphi(T,\vecx)$ of $P(T,\vecx,z)$ of sufficiently high accuracy.
    \item (\cref{lem:unique-minimal-polynomial}) Given a truncated, approximate $z$-root $\varphi$ of $P(T,\vecx,z)$, we set up a linear system whose solution will give us a minimal polynomial of the root $\varphi$. Since we know $G(0,\vecx,z)$, we know the $z$-degree of $G$; this is essentially the information we need to set up the linear system in a way that ensures that the solution is the same as the irreducible factor $G$.
    \item (\cref{thm:algo-circuit-min-poly-of-approx-root}) The solution to the linear system can be represented as a small circuit using Cramer's rule, but this involves a division by a determinant. We shall now use Strassen's division elimination to represent the solution as an arithmetic circuit without division gates; this requires a point where the denominator evaluates to a nonzero value. To this end, we again compose the denominator with $\ki_{g,\mathcal{S}}$ to get an $n^\varepsilon$-variate polynomial; since this denominator essentially captures the uniqueness of the minimal polynomial, we can prove that $\ki_{g,\mathcal{S}}$, by maintaining irreducibility of factors, also maintains the non-zeroness of the denominator. We can now use a brute-force derandomization of the Polynomial Identity Lemma to find a point where the denominator evaluates to a nonzero value, and hence, carry out Strassen's division elimination. Thus, we have division-free circuits for each irreducible factor $G(T,\vecx,z)$ of $P(T,\vecx,z)$. 
    \item Finally, we can undo our initial transformations and output circuits for the irreducible factors of the input polynomial $P(\vecx)$, along with their multiplicities. 
\end{enumerate} 

\subsection{The algorithm}

Fix any $\Delta \in \N$ and $\varepsilon \in (0,0.5)$. The following algorithm is the ``outermost'' wrapper for our complete algorithm, and it can be thought of as the algorithm $\mathcal{A}_{\Delta,\varepsilon}$ in \cref{thm:main-algorithm}. The algorithm computes the squarefree decomposition of the input polynomial, and then uses \cref{alg:factors-squarefree} on each squarefree part of the decomposition.

{
\small
\begin{algorithm}[H]
\onehalfspacing
\caption{All factors of depth-$\Delta$ circuits, parameter $\varepsilon \in (0,0.5)$}
\label{alg:all-factors}
\SetKwInOut{Input}{Input}\SetKwInOut{Output}{Output} \SetKwInOut{Prereq}{Prerequisites}
\Input{
A depth-$\Delta$ circuit $C_P(\vecx)$ of size $s \in \N$, computing a degree-$D$ polynomial $P(\vecx) = \prod_{i=1}^{m}G_i(\vecx)^{e_i} \in \Q[\vecx]$, where $\vecx = (x_1, \dots, x_n)$.
}
\Output{A list $L = \{(C_{G_1}(\vecx),e_1), \dots, (C_{G_m}(\vecx),e_m)\}$, such that each $C_{G_i}(\vecx)$ is a circuit of size $\poly(s,D)$ and depth $\poly(D)$, which computes the irreducible factor $G_i(\vecx)$, and $e_i$ is the multiplicity of $G_i$ in $P$.}
\BlankLine

Run Andrews-Wigderson (\autoref{thm:andrews-wigderson-squarefree-decomposition}) on $C_P(\vecx)$ to get circuits $C_{P_1}(\vecx), \dots, C_{P_r}(\vecx)$ of depth $\Delta' = \Delta + O(1)$, computing $P_1(\vecx), \dots, P_r(\vecx)$ where $r = \max_{i\in [m]} e_i$ and $P_i(\vecx) = \prod_{j \in S_i}G_j(\vecx)$ for $S_i = \{j \in [m]: e_j = i\}$ \label{algo1-line:andrews-wigderson-squarefree}

Initialize $L \leftarrow \emptyset$

\ForAll{$i \in [r]$}{
    Run \Cref{alg:factors-squarefree} for parameters $\Delta'$ and $\varepsilon$ on $C_{P_i}(\vecx)$ to get a list $L_i := \{C_{G_j}(\vecx): j \in S_i\}$ \label{algo1-line:run-algo2-squarefree}
    
    For each $C_{G_j}(\vecx) \in L_i$, add $(C_{G_j}(\vecx),i)$ to $L$. \label{algo1-line:add-multiplicity-info}
}

\Return{$L$}.

\end{algorithm}
}

\subsection{Algorithm for factors of squarefree polynomials}

In this section, we will describe the algorithm for steps 3 and 4 of the overview in \cref{sec:algo-overview}. Fix any $\Delta \in \N$. 

Let $\mathcal{G}=\{g_m\}_{m\in \N}$ be a family of polynomials such that for every $m\in \N$, $g_m \in \F[x_1, \dots, x_m]$, $d_m := \deg(g_m) \leq O(\log \log (m))$. Further, $\mathcal{G}$ has the property that for any depth $\Delta \in \N$, if $\mathcal{C} = \{C_m\}_{m\in \N}$ is a family of depth-$\Delta$ circuits computing $\mathcal{G}$, then $\mathcal{C}$ requires size $m^{d_m^{\exp(-O(\Delta))}}$, which is $m^{\omega(1)}$. \cref{LST-hardness} gives us such a family of explicit low-degree polynomials that are hard for constant-depth circuits.

Let $\varepsilon \in (0,0.5)$. Let $\sigma = O(n^\varepsilon),\mu = O(\frac{n^{2\varepsilon}}{\log(n)}),\rho = O(\log(n))$, and let $\mathcal{S}$ be an $(n,\sigma ,\mu ,\rho)$-design. Let $\ki_{g_{\sigma},\mathcal{S}}: \F^{\mu} \to \F^{n}$ be the polynomial map in \autoref{def:KI-generator} defined using the design $\mathcal{S}$ and the polynomial $g_{\sigma}$ from the family of hard polynomials  $\mathcal{G}$.

{
\small
\begin{algorithm}[H]
\onehalfspacing
\caption{Factors of squarefree polynomial with depth-$\Delta$ circuits, parameter $\varepsilon \in (0,0.5)$}
\label{alg:factors-squarefree}
\SetKwInOut{Input}{Input}\SetKwInOut{Output}{Output} \SetKwInOut{Prereq}{Prerequisites}
\Input{
A depth-$\Delta$ size-$s$ circuit $C_{P}$ computing the squarefree degree-$D$ polynomial $P(\vecx)= \prod_{j \in [m]} G_j(\vecx)$, where the $G_j$s are the irreducible factors of $P$.
}
\Output{A list $L = \{C_{G_1}(\vecx), \dots, C_{G_m}(\vecx)\}$, such that each $C_{G_j}(\vecx)$ is a circuit of size $\poly(s,D)$ and depth $\poly(D)$, which computes the irreducible factor $G_j(\vecx)$.}
\BlankLine

Compute $\veca \in \Q^n$ such that $\delta = \homog_D[P](\veca) \neq 0$ using \cref{thm:lst-PIT}. Let $\Hat{P}(\vecx,z) := P(\vecx + (\veca \cdot z))/\delta$ and for each $j \in [m]$, $\Hat{G_j}(T,\vecx,z) := {G_j}(\vecx + (\veca \cdot z))/\homog_{\deg(G_j)}[G_j](\veca)$ \label{algo2-line:monic-shift}

\tcp{Ensures that $\hat{P}$ is monic in $z$}

Compute $\vecb \in \Q^n$ such that $\operatorname{Disc}_z(\Hat{P})(\vecb) \neq 0$ using \cref{thm:lst-PIT}. Let $\tilde{P}(T,\vecx,z) := \Hat{P}((T\cdot \vecx) + \vecb, z)$ and for each $j \in [m]$, $\tilde{G_j}(T,\vecx,z) := \Hat{G_j}((T\cdot \vecx) + \vecb, z)$ \label{algo2-line:discriminant-shift-and-T-regularize} 

\tcp{Ensures that $\tilde{P}(T, \vecx, z)$ is $T$-regularised, and $P(0, \veczero, z)$ is squarefree.}

Construct the map $\ki_{g_\sigma,\mathcal{S}}$ using \cref{lem:explicit-designs} and $g_\sigma \in \mathcal{G}$ as given by \cref{LST-hardness}. \label{algo2-line:ki-map-construct}

Factorize $C_{\tilde{P}}(T,\ki_{g,\mathcal{S}}(\vecw),z)$ so that for each $j \in [m]$, $C_{\tilde{G}_j}(T,\ki_{g,\mathcal{S}}(\vecw),z)$ is a circuit  that computes $\tilde{G}_j(T,\ki_{g,\mathcal{S}}(\vecw),z)$, satisfying the property that $\tilde{G}_j(0,\ki_{g,\mathcal{S}}(\vecw),z) = \tilde{G}_j(0,\veczero,z)$. \label{algo2-line:factorize-n^eps-variate}
    
Run \Cref{alg:univariate-to-factor} with parameters $\Delta+2$ and $\varepsilon$ on $C_{\tilde{P}}(T,\vecx,z)$ and $\{\tilde{G}_j(0,\veczero,z): j \in [m]\}$ to get a list $L' = \{C_{\tilde{G}_1}(T,\vecx,z), \dots, C_{\tilde{G}_m}(T,\vecx,z)\}$. \label{algo2-line:run-algo3} 

For each $C_{\tilde{G}_j}(T,\vecx,z) \in L'$, compute $C_{G_j}(\vecx) := C_{\tilde{G}_j}(1,\vecx-\vecb,0)$. \label{algo2-line:undo-preprocessing}

Output  $L = \{C_{G_1}(\vecx), \dots, C_{G_m}(\vecx)\}$

\end{algorithm}
}

\subsection{Algorithm for obtaining irreducible factors from the right univariate projections}

In this section, we describe the algorithm for steps 5-7 of the overview in \cref{sec:algo-overview}. Fix any $\Delta \in \N$ and $\varepsilon \in (0,0.5)$.

{
\small
\begin{algorithm}[H]
\onehalfspacing
\caption{Factors of monic, $T$-regularized, squarefree polynomial with depth-$\Delta$ circuit, parameter $\varepsilon \in (0,0.5)$}
\label{alg:univariate-to-factor}
\SetKwInOut{Input}{Input}\SetKwInOut{Output}{Output}\SetKwInOut{Prereq}{Prerequisites}

\Input{
\begin{itemize}
    \item A depth-$\Delta$ size-$s$ circuit $C_{\tilde{P}}$ computing a squarefree degree-$D$ polynomial $\tilde{P}(T,\vecx,z)= \prod_{j \in [m]} \tilde{G}_j(T,\vecx,z)$, where the $\tilde{G}_i$s are the irreducible factors of $\tilde{P}$. Further, $\tilde{P}(T,\vecx,z)$ is monic in $z$, $T$-regularized with respect to $\vecx$ and $C_{\tilde{P}}(0,\vecx,z) = C_{\tilde{P}}(0,\veczero,z)$ is squarefree.
    \item For each $j \in [m]$, the univariate polynomial $\tilde{G}_j(0,\vecx,z) = \tilde{G}_j(0,\veczero,z)$.
\end{itemize}
}
\Output{A list $L = \{C_{\tilde{G}_1}(T,\vecx,z), \dots, C_{\tilde{G}_m}(T,\vecx,z)\}$, such that each $C_{\tilde{G}_i}(T,\vecx,z)$ is a circuit of size $\poly(s,D)$ and depth $\poly(D)$, which computes the irreducible factor $\tilde{G}_i(T,\vecx,z)$ of $\tilde{P}(T,\vecx,z)$.}
\Prereq{$\ki_{g_\sigma,\mathcal{G}}$ constructed in \cref{alg:factors-squarefree} using \cref{lem:explicit-designs} and \cref{LST-hardness}}

\BlankLine

Initialize $L \leftarrow \emptyset$

{
\ForAll{$j \in [m]$}{ 

Factorize the degree-$D_z$ polynomial $\tilde{G}_j(0,\veczero,z)$ (\cref{thm:-LLL-univariate-factorization}) to get $\tilde{G}_j(0,\veczero,z) = \prod_{l=1}^{r_j} H_l(z)$, where every $H_l(z)$ is distinct, monic and irreducible. \label{alg3-line:univariate-factorization}

Let $H(z)$ be an arbitrary irreducible factor of $\tilde{G}_j(0,\veczero,z)$. Let $\mathbb{K}$ be the field $\frac{\Q[u]}{H(u)}$. \label{alg3-line:number-field-adjoin-root}

Use Newton iteration (\cref{lem:newton-iteration-linear}) to compute a truncated, non-degenerate, approximate $z$-root $\varphi(T,\vecx) \in \mathbb{K}[T,\vecx]$ of $\tilde{G}_j(T,\vecx,z)$ of order $2D\cdot D_z+1$ such that $\varphi(0,\vecx) = \varphi(0,\veczero) = u \in \mathbb{K}$. \label{alg3-line:newton-iteration}

Use the algorithm from \cref{thm:algo-circuit-min-poly-of-approx-root} (with $\ki_{g_\sigma,\mathcal{G}}$) on $\varphi(T,\vecx)$ to compute a circuit $C_{\tilde{G_j}}$ for $\tilde{G_j}$. \label{alg3-line:lin-system-div-free-ckt}

Add $C_{\tilde{G_j}}$ to $L$.
}
}

\Return{$L$}
\end{algorithm}
}

\section{Technical building blocks} \label{sec:technical}

\subsection{Properties of roots modulo {$T^k$}}

The following lemma observes that certain homomorphisms defined by polynomial maps preserve approximate roots for a polynomial. This will be used to argue that an approximate root remains an approximate root (with some nice properties) even after we plug in the Kabanets-Impagliazzo generator.

\begin{lemma}[Preserving root properties under homomorphisms]\label{lem:preserving-root-properies-under-homomorphism}
Let $k > 1$ be any natural number and let $P(T, \vecx, z) \in \F[T, \vecx, z]$ and $R(T, \vecx) \in \F[T, \vecx]$ be polynomials satisfying the following properties. 
\begin{itemize}
    \item $P(T,\vecx,z)$ is $T$-regularized and it is monic in $z$.  
    \item $R(T,\vecx)$ is a truncated, non-degenerate, approximate $z$-root of $P(T,\vecx,y)$ of order $k$.
\end{itemize}

For a new tuple $\vecw$  of $\mu$ variables distinct from $T, \vecx, z$, polynomials $h_1, h_2, \ldots, h_n \in \F[\vecw]$ and a non-zero field constant $\gamma \in \F$, let $\Lambda:\F[T,\vecx,z] \to \F[T,\vecw,z]$ be the ring homomorphism defined by $T \mapsto \gamma T$, $x_i \mapsto h_i(\vecw)$ and $z \mapsto z$. 
Then, the following are true.  
\begin{itemize} 
    \item $\Lambda(R)(0,\vecw) = R(0,\vecx) \in \F$
    \item $\Lambda(P)(T,\vecw,z)$ is $T$-regularized, and it is monic in $z$.
    \item $\Lambda(R)(T,\vecw)$ is a truncated, non-degenerate, approximate $z$-root of $\Lambda(P)(T,\vecw,y)$ of order $k$    
\end{itemize}
\end{lemma}
\begin{proof}
    The lemma follows essentially because $\Lambda$ is a homomorphism.
    \begin{enumerate}
        \item By \cref{obs:approx-root-simple-obs}, $R(0,\vecx) \in \F$. Since $\Lambda$ is a homomorphism with the property that $\deg_T(f) = \deg_T(\Lambda(f))$ for any $f \in \F[T,\vecx,z]$, it follows that $\Lambda(R)(0,\vecw) = R(0,\vecx)$.
        \item If a monomial $\vecm \in \F[T,\vecx,z]$, is $T$-regularized, then so is $\Lambda(\vecm)$. Since $\Lambda$ is a homomorphism, this property extends to all polynomials. Similarly, since $\Lambda(z) = z$, $\Lambda(P)$ remains monic in $z$.
        \item $P(T,\vecx,R(T,\vecx)) \equiv 0 \mod T^k$, or equivalently, every monomial in $P(T,\vecx,R(T,\vecx))$ has $T$-degree at least $k$. As observed in the first point, $\Lambda$ is a homomorphism with the property that $\deg_T(f) = \deg_T(\Lambda(f))$ for any $f \in \F[T,\vecx,z]$. Thus, $\Lambda(P)(T,\vecw,\Lambda(R)(T,\vecw)) \equiv 0 \mod T^k$. 
    \end{enumerate}
    
\end{proof}

The following lemma observes that truncated, non-degenerate, approximate roots of a polynomial are unique once the value of the root at $T=0$ is decided; it follows almost immediately from the uniqueness of approximate roots computed via Newton iteration.
\begin{lemma}[Uniqueness of truncated, non-degenerate, approximate roots]
    \label{lem:unique-truncated-nondeg-approx-roots}
    Let $k \in \N$. Suppose $\Psi_1(T, \vecx)$ and $\Psi_2(T, \vecx)$ are truncated, non-degenerate, approximate $z$-roots of order $k$ with respect to $T$ for a $T$-regularized polynomial $P(T, \vecx, z)$. If $\Psi_1(0,\vecx) = \Psi_2(0,\vecx) = \alpha \in \F$, then $\Psi_1(T,\vecx) \equiv \Psi_2(T,\vecx)$. 
\end{lemma}

\begin{proof}
    Since $\Psi_1(T,\vecx)$ is a non-degenerate root of $P(T,\vecx,z)$ which is $T$-regularized, $\frac{\partial P}{\partial z}(0,\vecx,\Psi_1(0,\vecx)) = \frac{\partial P}{\partial z}(0,\vecx,\Psi_1(0,\vecx)) = \frac{\partial P}{\partial z}(0,\veczero,\alpha) = \beta \in \F$ is nonzero. Similarly,  $\frac{\partial P}{\partial z}(0,\vecx,\Psi_2(0,\vecx)) = \frac{\partial P}{\partial z}(0,\vecx,\Psi_2(0,\vecx)) = \frac{\partial P}{\partial z}(0,\veczero,\alpha) = \beta$. Since $\alpha$ satisfies $P(T,\vecx,\alpha) \equiv 0 \bmod T$ and $\frac{\partial P}{\partial z}(0,\vecx,\alpha) = \beta \in \F$ for $\beta \neq 0$, \cref{lem:newton-iteration-linear} tells us that there is a unique truncated, approximate $z$-root $\Phi(T,\vecx)$ of order $k$ for $P(T,\vecx,z)$, satisfying $\Phi(0,\vecx) = \alpha$. Thus, $\Phi(T,\vecx) \equiv \Psi_1(T,\vecx) \equiv \Psi_2(T,\vecx)$. 
\end{proof}

\subsection{Complexity of low degree homogeneous components of roots}

The following is our main technical lemma where we argue that the low-degree homogeneous components of an approximate root for a constant-depth circuit can be computed by a small constant-depth circuit.
\begin{lemma}\label{lem:low-deg-w-components-of-roots-mod-T}
Let $k > 1$ be any natural number and let $P(T, \vecw, z) \in \F[T, \vecw, z]$ and $R(T, \vecw) \in \F[T, \vecw]$ be polynomials satisfying the following properties. 
\begin{itemize}
    \item $P(T, \vecw, z)$ is computable by a depth $\Delta$ circuit of size $s$. 
    \item $P(T, \vecw, z)$ is $T$-regularized and monic in $z$. 
    \item $R(T, \vecw)$ is a truncated, non-degenerate, approximate $z$-root of $P(T,\vecx,z)$ of order $k$ with respect to $T$. 
\end{itemize}
Then, for every $\ell \in \N$, there is an algebraic circuit $C_{\ell} \in \F[T, \vecw]$ of depth at most $(\Delta + O(1)) $, size at most $\left(\poly(s,\deg(P))\cdot (\log k)^{\poly(\ell)}\right)$ 
\[
C_{\ell}(T,\vecw) = R(T, \vecw) \trunc \langle \vecw\rangle^{\ell} \, .
\]
\end{lemma}
\begin{proof}
    As seen in \cref{obs:approx-root-simple-obs}, $R(0, \vecw)$ is a root of $P(0, \mathbf{0}, z)$, and hence must be a field element $\alpha$ (possibly from an extension of $\F$) and does not depend on $\vecw$. Similarly, we also get that $\frac{\partial P}{\partial z}\left(0, \vecw, R(0, \vecw)\right)$ is a non-zero field element that we denote by $\beta$. Thus, $R(T, \vecw)$ is the \emph{unique} lift of the root $\alpha$ of $P(T,\vecw,z)$ modulo $T$ to a root modulo the ideal $T^k$ (\cref{lem:unique-truncated-nondeg-approx-roots}). In particular, we can view $R$ as an outcome of Newton Iteration (and then eventual truncation modulo $T^k$). We prove the lemma by induction on this iteration, and maintaining the following inductive claim. Let $2^{m} \in [k, 2k]$ be the smallest power of $2$ greater than $k$. From \autoref{lem:newton-iteration-quadratic-div-free}, we get that the sequence of polynomials $\phi_0, \phi_1, \ldots, \phi_{m}, \psi_0, \psi_1, \ldots, \psi_{m} \in \F[T,\vecw]$ defined as 
  \begin{align*}
    \phi_0 & = \alpha, & \psi_0 &= (1/\beta),\\
    \text{For $i \geq 0$,}\quad \phi_{i+1} &= \phi_i - P(T, \vecw, \phi_i) \cdot \psi_i, & \psi_{i+1} &= 2\psi_i - \psi_i^2 \cdot \frac{\partial P}{\partial z}(T, \vecw, \phi_{i+1}).
    \end{align*}
    satisfy \[
  R(T,\vecw) = \phi_{m}(T,\vecw) \trunc T^{k} \, .
  \]
  
  We note that while we are dealing with polynomials in both $T$ and $\vecw$ variables, the lifting is happening only with respect to $T$. In this sense, we are really viewing the polynomial $P(T,\vecw, z)$ as a polynomial in  $\F[\vecw][T,z]$ for the purpose of this lifting. We now use the following claim, whose proof we defer to the end of this section, to complete the proof of the lemma. 
  \begin{claim}\label{clm:induction-over-NI}
    For every $i \geq 0$, there exists a set $\mathcal{G}_i$ of at most $\tau_i \leq 2(\ell+1)i$ polynomials $\{g_1, g_2, \ldots, g_{\tau_i}\}$ in $\F[T,\vecw]$ and two $(\tau_i + 1)$-variate polynomials $Q_i(T, u_1, \ldots, u_{\tau_i}), \Hat{Q}_i(T, u_1, \ldots, u_{\tau_i}) \in \F[T,\vecu]$ such that  
    \[
        \phi_{i} \equiv Q_i(T, g_1, \ldots, g_{\tau_i}) \mod \langle \vecw \rangle^{\ell}\, ,
    \]
    and
    \[
        \psi_{i} \equiv \Hat{Q}_i(T, g_1, \ldots, g_{\tau_i}) \mod \langle \vecw \rangle^{\ell}\, .
    \]
    Moreover, each polynomial $g$ in $\mathcal{G}_i$ is of the form $\frac{\partial P}{\partial z^j}(T, \vecw, \gamma)$ for some $j \leq (\ell+1)$ and $\gamma \in \F[T]$.
  \end{claim}
    From the fact that $R(T,\vecw) = \phi_{m}(T,\vecw) \trunc T^{k}$, we get that 
    \[
    R(T,\vecw) \trunc \inangle{\vecw}^{\ell} = \left(\phi_{m}(T,\vecw) \trunc  T^{k} \right) \trunc \langle \vecw \rangle^{\ell} \, .
    \]
    Now, since $T$ and $\vecw$ are disjoint variables, we can exchange the order of the operations of truncating modulo $T$ and truncating modulo $\langle \vecw \rangle^{\ell}$. So, we have that
    \[
    R(T,\vecw) \trunc \inangle{\vecw}^{\ell} = \left(\phi_{m}(T,\vecw)   \trunc \langle \vecw \rangle^{\ell} \right) \trunc  T^{k} \, .
    \]
    From \autoref{clm:induction-over-NI}, we get that 
    \[
    R(T,\vecw) \trunc \inangle{\vecw}^{\ell} = \left( Q_m(T, g_1, g_2, \ldots, g_{\tau_m})  \trunc \langle \vecw \rangle^{\ell} \right) \trunc  T^{k} \, .
    \]
   We would like to show that the RHS of the above equation has a small constant-depth circuit. Since we eventually truncate to $T$-degree $k-1$, we can assume without loss of generality that the $T$-degree of $Q_m$ and each $g_i$ is at most $k-1$. In particular, if $g_i$ equals $\frac{\partial P}{\partial z^j}(T, \vecw, \gamma)$ for some $j \in \N$ and $\gamma \in \F[T]$, this $\gamma$ can be assumed to be of $T$-degree at most $(k-1)$. Moreover, since $P(T, \vecw, z)$ is a polynomial with a depth-$\Delta$ circuit of size $s$, we get, from \cref{cor:interpolation-consequences} that $\frac{\partial P}{\partial z^j}(T, \vecw, z)$ has a depth $(\Delta + O(1))$ circuit of size at most $s \cdot \poly(\deg(P))$. From the bound on the degree of $\gamma \in \F[T]$, we get that each $g_i(T,\vecw)$ has a depth $(\Delta + O(1))$ circuit of size at most $s \cdot \poly(\deg(P))$. 

   Finally, we note that we can view $Q_m(T, g_1, g_2, \ldots, g_{\tau_m})$ as 
   \[
    \tilde{Q}_m(T, g_1(T,\vecw) - g_1(T,\mathbf{0}), g_2(T,\vecw) - g_2(T, \mathbf{0}), \ldots, g_{\tau_m}(T,\vecw) - g_{\tau_m}(T, \mathbf{0})) ,\]
   for some polynomial $\tilde{Q}_m$. Since we are only interested in working with the above polynomial modulo $\langle \vecw \rangle^{\ell}$ and every monomial in each of the polynomials $g_j(T,\vecw) - g_j(T,\mathbf{0})$ has $\vecw$-degree at least one, we get that $\tilde{Q}_m(T, u_1, \ldots, u_{\tau_m})$ can be assumed to have degree at most $\ell$ in the $\vecu$ variables. Thus, $\tilde{Q}_m(T, u_1, \ldots, u_{\tau_m})$ can be computed by a depth-$2$ circuit of size at most $k\cdot \binom{\tau_m + \ell}{\ell} \leq k \cdot (\log k)^{\poly(\ell)}$. Combining this circuit with the constant-depth circuits for $g_j(T, \vecw)$ (and hence $g_j(T,\vecw) - g_j(T,\mathbf{0})$) gives us a depth-$(\Delta + O(1))$ circuit $\tilde{C}$ of size at most $\poly(s,\deg(P))\cdot (\log k)^{\poly(\ell)}$ satisfying \[ \tilde{C}(T,\vecw) \equiv Q_m(T,g_1, \dots, g_{\tau_m}) \bmod T^k \bmod \inangle{\vecw}^{\ell+1}.\]
   The $T$-degree of this circuit is at most $\deg_T(\tilde{Q}_m) \cdot \ell \cdot k \leq \poly(k, \ell)$. The $\vecw$-degree of this circuit is at most $\ell \cdot \deg(P)$. Thus, we can apply the operations $\trunc \inangle{\vecw}^{\ell + 1}$ and $\trunc T^k$ by an application of \cref{cor:interpolation-consequences} to get the circuit $C_{\ell}(T,\vecw)$ in the conclusion of the lemma. 
\end{proof}
We now discuss the proof of \autoref{clm:induction-over-NI}. 
\begin{proof}[Proof of \autoref{clm:induction-over-NI}]
    We prove the claim via an induction on $i$. For $i = 0$, we already know that $\phi_0 = \alpha$ and $\psi_0 = \frac{1}{\beta}$ are polynomials in $\F[T]$ and hence the claim immediately holds. We now assume that the claim holds for $\phi_i$ and $\psi_i$ i.e.
    \[
        \phi_{i} \equiv Q_i(T, g_1, \ldots, g_{\tau_i}) \mod \langle \vecw \rangle^{\ell}\, ,
    \]
    and
    \[
        \psi_{i} \equiv \Hat{Q}_i(T, g_1, \ldots, g_{\tau_i}) \mod \langle \vecw \rangle^{\ell}
    \]
    and show that it must hold for $\phi_{i+1}$ and $\psi_{i+1}$.  

    From the inductive definitions of $\phi_{i+1}$ and $\psi_{i+1}$, we get that 
    \[
    \phi_{i+1} \equiv \phi_i - P(T, \vecw, \phi_i) \cdot \psi_i
    \]
    and 
    \[
    \psi_{i+1} \equiv 2 \psi_i - \psi_i^2 \cdot \frac{\partial P}{\partial z}(T, \vecw, \phi_{i+1})
    \]
    
    If $\gamma_i := \phi_i(T, \mathbf{0})$, i.e. $\gamma_i$ is the $\vecw$-free part of $\phi_i$, we get via a Taylor expansion that 
    \begin{equation} \label{eqn:taylor-expansion-in-quadratic-NI}
        P(T,\vecw,{\phi}_i) = P(T,\vecw,\gamma_i + ({\phi}_i-\gamma_i)) = \sum_{j = 0}^{\deg(P)} \frac{1}{j!} \cdot \frac{\partial P}{\partial z^j}(T, \vecw, \gamma_i) \cdot   ({\phi}_i-\gamma_i)^j \, .
    \end{equation}
    From the definition of $\gamma_i$, we have that every monomial in    $({\phi}_i-\gamma_i)$ has $\vecw$-degree at least one. Thus, for every $j \geq \ell$, we have that $({\phi}_i-\gamma_i)^j \equiv 0$ modulo $\langle \vecw \rangle^{\ell}$. Thus, we have that 
    \[
    P(T,\vecw,{\phi}_i) \equiv \sum_{j = 0}^{\ell} \frac{1}{j!} \cdot \frac{\partial P}{\partial z^j}(T, \vecw, \gamma_i) \cdot   ({\phi}_i-\gamma_i)^j  \mod \langle \vecw \rangle^{\ell} \, .
    \]
    Similarly, 
    \[
    \frac{\partial{P}}{\partial z}(T,\vecw,{\phi}_{i+1}) \equiv \sum_{j = 0}^{\ell} \frac{1}{j!} \cdot \frac{\partial P}{\partial z^{j+1}}(T, \vecw, \gamma_{i+1}) \cdot   ({\phi}_{i+1}-\gamma_{i+1})^j  \mod \langle \vecw \rangle^{\ell} \, .
    \]
    Thus, $\varphi_{i+1}$ can be written as a polynomial in $Q_i, \Hat{Q}_i$ and polynomials in the set $\{\frac{\partial P}{\partial z^j}(T, \vecw, \gamma_i) : j \in \{0, 1, \ldots, \ell\}\}$, where the coefficients of this polynomial are from the ring $\F[T]$ and all equalities hold modulo $\langle \vecw \rangle^{\ell}$. Similarly, $\psi_{i+1}$ can be written as a polynomial in $Q_{i+1}, \Hat{Q}_i$ and polynomials in the set $\{\frac{\partial P}{\partial z^j}(T, \vecw, \gamma_{i+1}) : j \in \{1,2, \ldots, \ell+1\}\}$, where the coefficients of this polynomial are from the ring $\F[T]$ and all equalities hold modulo $\langle \vecw \rangle^{\ell}$. Moreover, in going from $i$ to $(i+1)$, we have increased the size of the \emph{generating set} additively by at most  $2(\ell + 1)$. Furthermore, each of the new elements of this generating set is again of the form $\frac{\partial P}{\partial z^j}(T, \vecw, \gamma)$ for some $j \leq (\ell+1)$ and $\gamma \in \F[T]$. 
    
 \end{proof}

An immediate consequence of \autoref{lem:low-deg-w-components-of-roots-mod-T} is the following corollary. 
\begin{corollary}\label{cor:low-deg-w-components-of-roots-mod-T-composed}
    Let $k \in \N$,  $P(T, \vecw, z) \in \F[T, \vecw, z]$ and $R(T, \vecw) \in \F[T, \vecw]$ be polynomials that satisfy the hypothesis of \autoref{lem:low-deg-w-components-of-roots-mod-T}. 
    
    Then, for every $\ell \in \N$, there is an algebraic circuit ${C}_{\ell} \in \F[T, \vecw]$ of depth at most $2\Delta + O(1)$ and size at most $\left(\poly(s,\deg(P))\cdot (\log k)^{\poly(\ell)}\right)$ such that 
    \[
        {C}_{\ell}(T,\vecw) \equiv P\left(T, \vecw, R(T, \vecw)\right) \trunc \langle \vecw\rangle^{\ell} \, .
    \]
    
\end{corollary}

\subsubsection{Why do we need the quadratic convergence version of Newton iteration?} \label{rmk:why-quadratic-convergence}

It would be instructive to step back and see why Newton iteration with quadratic convergence was required in the above argument. A slight modification of \cref{lem:CKS-updated} will let us argue that low-degree components of approximate roots of a polynomial have small constant-depth circuits. However, there are subtle differences in the statement of \cref{lem:CKS-updated} and the statement of \cref{lem:low-deg-w-components-of-roots-mod-T}. The key difference is that the approximate root $\varphi_k$ is with respect to the variable $T$, whereas we are extracting homogeneous components of degree up to $\ell$ with respect to a \emph{different set of variables}. Therefore, we may have contributions of $\vecw$-degree at most $\ell$ from terms with $T$-degree up to $k$. So, we cannot replace $\varphi_k$ by a root of lower accuracy such as $\varphi_\ell$. 

In the above proof of \cref{lem:low-deg-w-components-of-roots-mod-T}, each level of the iteration adds to the number of ``generators'' used in the composition. Using the standard Newton iteration of $k$ steps for $\varphi_k$ would result in $O(k)$ generators and would not yield the required size bounds needed for our proof. The version of Newton iteration with quadratic convergence ensures that we obtain $\phi_k$ from $\log k$ iterations, and results in the eventual number of ``generators'' as $O(\log k)$ instead (which was crucial for the final circuit size bound). 

\section{Warm-up: preserving true roots under variable reduction}\label{sec:root-preservation}

In this section, we use the techniques of \autoref{sec:technical} to show a variable reduction map that will help us identify whether a given truncated power series root of a constant-depth circuit is a true root (and not just a sufficiently good truncation of a power series root). As described in the proof overview, the idea is to work with a PIT instance formed by plugging-in a candidate root into the input polynomial. In \cref{lem:low-deg-y-roots-of-hybrid-poly}, we show that low-degree roots of such PIT instances have relatively small constant-depth circuits when the input polynomial has a small constant-depth circuit. 

\begin{lemma}\label{lem:low-deg-y-roots-of-hybrid-poly}
Let $k > 1$ be any natural number and let $P(T, \vecw, y, z) \in \F[T, \vecw, z]$ and $R(T, \vecw, y) \in \F[T, \vecw, y]$ be polynomials satisfying the following properties. 
\begin{itemize}
    \item $P(T, \vecw, y, z)$ is computable by a depth $\Delta$ circuit of size $s$. 
        \item $P(T,\vecw,y,z)$ is $T$-regularized and monic in $z$.
    \item $R(T, \vecw, y)$ is a truncated, non-degenerate, approximate $z$-root of $P(T,\vecw,y,z)$ of order $k$.
\end{itemize}
Let $F(\vecw) \in \F[\vecw]$ be a polynomial of degree at most $\ell$ such that $(y-F(\vecw))$ divides $P\left(T,\vecw, y, R(T, \vecw, y) \right)$. Then, $F$ can be computed by an algebraic circuit of depth $(\Delta + O(1))$ and size at most 
\[
\poly(s,\deg(P))\cdot (\log k)^{\poly(\ell)} \, .
\]
\end{lemma}
\begin{proof}[Proof of \autoref{lem:low-deg-y-roots-of-hybrid-poly}]
We start by setting ourselves up to invoke \cref{lem:CKS-updated-rational-in-T-new}, a version of \autoref{lem:CKS-updated} from \cite{ChouKS19}. To do this, we first need to ensure that the hypothesis of the lemma holds in our case.

\paragraph{Setting up to invoke \autoref{lem:CKS-updated-rational-in-T-new}: }
The first issue is that $F(\vecw)$ could be a $y$-root of high multiplicity of $P\left(T,\vecw, y, R(T, \vecw, y) \right)$. To work around this, we work with an appropriately high order derivative of $P$ with respect to $y$. To this end, we start by viewing $P\left(T,\vecw, y, R(T, \vecw, y)\right)$ as a univariate polynomial in $y$ with coefficients from the ring $\F[T, \vecw]$. Thus, $P\left(T,\vecw, y, R(T, \vecw, y)\right)$ can be decomposed as 
\[
P\left(T,\vecw, y, R(T, \vecw, y)\right) = \sum_{i = 0}^D P_i(T, \vecw) y^i \, ,
\]
where each $P_i$ is a polynomial in only $T$ and $\vecw$ variables. Let $(m+1)$ be the largest integer such that $(Y - F(\vecw))^{m+1}$ divides $P\left(T,\vecw, y, R(T, \vecw, y)\right)$. So, to reduce the multiplicity of the $y$-root $F(\vecw)$ to $1$, we consider the polynomial $\Hat{P}$ defined as  
\[
\Hat{P}(T, \vecw, y) := \frac{\partial P\left(T,\vecw, y, R(T, \vecw, y)\right)}{\partial y^m} \, .
\]
Thus, we have that $(y - F(\vecw))$ divides $\Hat{P}$ and does not divide $\frac{\partial\Hat{P} }{\partial y}$. As a consequence, we get that $\frac{\partial \Hat{P}}{\partial y}(T, \vecw, F(\vecw))$ is a non-zero polynomial in $\F[T, \vecw]$. By shifting the $\vecw$ variables if needed to $\veca + \vecw$ for some $\veca \in \F^{|\vecw|}$, we get that 
\[
\frac{\partial \Hat{P}}{\partial y}(T, \veczero, F(\veczero)) \not\equiv 0 \, .
\]

We are now in a position to invoke \Cref{lem:CKS-updated-rational-in-T-new}, which tells us that there exists a $\kappa \in \F$ and a $(\ell+1)$-variate polynomial $Q(\vecu) \in \F[\vecu]$ of degree at most $\ell$ such that 
\begin{align}
    F(\vecw) \equiv  Q(h_0(\kappa,\vecw), h_1(\kappa,\vecw), \ldots, h_{\ell}(\kappa,\vecw)) \mod \langle \vecw \rangle^{\ell+1} \, , \label{eqn:CKS-on-F(w)}
\end{align}
    where for every $i \in \{0, 1, \ldots, \ell\}$, $h_i(T, \vecw)$ is defined as 
    \[
        h_j(T, \vecw) := \inparen{\frac{\partial \Hat{P}}{\partial y^j}(T, \vecw, F(\mathbf{0})) - \frac{\partial \Hat{P}}{\partial y^j}(T, \mathbf{0}, F(\mathbf{0}))} \trunc \langle \vecw \rangle^{\ell+1} \, .
    \]

Given the bounds on the number of variables and $\vecu$-degree of $Q(\vecu)$, we get that $Q(\vecu)$ is computable by an algebraic circuit of depth-$2$ and size at most $\exp(O(\ell))$. At this point, if we can somehow replace each $h_i(T, \vecw)$ by a small constant-depth circuit, we can obtain a small constant-depth circuit for $F(\vecw)$ by combining these with the depth-$2$ circuit for $Q$, and then extracting certain homogeneous components of interest. 

\paragraph{Constant-depth circuits for $h_i(T, \vecw)$: }
Recall that $\Hat{P}(T, \vecw, y)$ is a derivative of $P(T, \vecw, y, R(T, \vecw,y))$ with respect to $y^m$. Thus, from \cref{cor:interpolation-consequences}, we get that $\Hat{P}$ can be written as an $\F[y]$-weighted linear combination of the polynomials in the set 
\[
\{P(T, \vecw, \beta_i, R(T, \vecw, \beta_i)) : i \in \{0, 1, \ldots, \deg(P)\} \} \, 
\]
where $\beta_i$'s are distinct field constants and every weight in the linear combination has $y$-degree at most $\deg(P)$. 
\begin{align*}
    h_j(T,\vecw) &= \inparen{\frac{\partial \Hat{P}}{\partial y^j}(T, \vecw, F(\mathbf{0})) - \frac{\partial \Hat{P}}{\partial y^j}(T, \mathbf{0}, F(\mathbf{0}))} \trunc \langle \vecw \rangle^{\ell+1} \\
    &= \inparen{\inparen{\frac{\partial}{\partial y^{j+m}}P(T,\vecw,y,R(T,\vecw,y))}\middle|_{y=F(\veczero)} - \frac{\partial \Hat{P}}{\partial y^j}(T, \mathbf{0}, F(\mathbf{0}))} \trunc \langle \vecw \rangle^{\ell+1}\\ 
    &= \inparen{\inparen{\sum_{i=0}^{D} \Lambda_i(F(\veczero))P(T,\vecw,\beta_i,R(T,\vecw,\beta_i))} - \frac{\partial \Hat{P}}{\partial y^j}(T, \mathbf{0}, F(\mathbf{0}))} \trunc \langle \vecw \rangle^{\ell+1}\\
    &= \inparen{\sum_{i=0}^{D} \Lambda_i(F(\veczero))P(T,\vecw,\beta_i,R(T,\vecw,\beta_i)) \trunc \inangle{\vecw}^{\ell + 1}} - \frac{\partial \Hat{P}}{\partial y^j}(T, \mathbf{0}, F(\mathbf{0}))
\end{align*}

Thus, to show that  $h_i(T, \vecw)$ have small constant-depth circuits, it suffices (up to multiplicative factors of $\poly(\deg(P))$) to show that the polynomials $(P(T, \vecw, \beta_i, R(T, \vecw, \beta_i)) \trunc \langle \vecw \rangle^{\ell + 1} )$ have small constant-depth circuits. But, this is exactly the content of \autoref{cor:low-deg-w-components-of-roots-mod-T-composed}\footnote{There is a slight subtlety here: $P(T,\vecw,y,R(T,\vecw,y))$ must still satisfy the hypothesis of \Cref{lem:low-deg-w-components-of-roots-mod-T}/\Cref{cor:low-deg-w-components-of-roots-mod-T-composed} after replacing each $y$ by a field constant $\beta_i$. But this is true because of \Cref{lem:preserving-root-properies-under-homomorphism}.} , which shows that $(P(T, \vecw, \beta_i, R(T, \vecw, \beta_i)) \trunc \langle \vecw \rangle^{\ell + 1} )$ has a circuit of depth $\Delta+O(1)$ and size 
\[
\left(\poly(s, \deg(P)) \cdot (\log k)^{\poly(\ell)} \right) \, .
\]
This gives a circuit $C_i(T,\vecw)$ of roughly the same size with a constant additive increase in depth for each $h_i(T, \vecw)$, and thus a circuit $\tilde{C}_i(\vecw):= C_i(\kappa,\vecw)$ for each $h_i(\kappa, \vecw) $, where $\kappa$ is the field constant in \Cref{eqn:CKS-on-F(w)}. 

\paragraph{Putting things together:}
Plugging in these circuits as inputs to the depth-$2$ circuit for $Q(\vecu)$, we get that there is a circuit $C(\vecw) = Q(\tilde{C}_0(\vecw), \dots, \tilde{C}_{\ell}(\vecw))$ of size $\left(\poly(s, \deg(P)) \cdot (\log k)^{\poly(\ell)} \right)$ and depth $(\Delta + O(1))$ such that 
\[
    F(\vecw) \equiv C(\vecw) \bmod \langle \vecw \rangle^{\ell + 1} \, .
\]
The $\vecw$-degree of $C(\vecw)$ is at most $\ell \cdot \deg(Q)$, which is at most $\ell^2$. By using \autoref{cor:interpolation-consequences} to compute the truncation $C(\vecw) \trunc \inangle{\vecw}^{\ell+1}$, we get a depth-$(\Delta + O(1))$ circuit for $F(\vecw)$, of size at most \[\left(\poly(s, \deg(P)) \cdot (\log k)^{\poly(\ell)} \right).\] 
\end{proof}

Now, we will use \cref{lem:low-deg-y-roots-of-hybrid-poly} to prove the main theorem of this section: the Kabanets-Impagliazzo generator, instantiated with an appropriate design and a low-degree polynomial, preserves the roots of constant-depth circuits, while ensuring that no new roots are created. This is a special case of \cref{thm:irreducibility-preservation}. 

\begin{theorem}[Checking validity of a root]\label{thm:technical-theorem-1}
Let $k > 1$ be any natural number and let $A(T, \vecx, z) \in \F[T, \vecx, z]$ and $\Phi(T, \vecx) \in \F[T, \vecx]$ be polynomials satisfying the following properties. 
\begin{itemize}
    \item $A(T, \vecx, z)$ is computable by a depth $\Delta$ circuit of size $s$. 
    \item  $A(T, \vecx, z)$ is $T$-regularized and monic in $z$.
    \item $\Phi(T, \vecx)$ is a truncated, non-degenerate, approximate $z$-root of  $A(T, \vecx, z)$ of order $k$
    \item $A(T, \vecx, \Phi(T, \vecx))$ is not \emph{identically zero} in $\F[T, \vecx]$.
\end{itemize}
Let $\mathcal{S}$ be a $(n,\sigma,\mu,\rho)$-design and let $g$ be a $\sigma$-variate polynomial of degree $d$. For a new tuple $\vecw$  of  $\mu$ variables distinct from $T, \vecx, z$, let  us define the polynomial map $\ki_{g,\mathcal{S}}:\F^\mu \to \F^n$ using \cref{def:KI-generator}.
Then, the following statement is true. 

If $A(T,\ki_{g,\mathcal{S}}(\vecw),\Phi(T,\ki_{g,\mathcal{S}}(\vecw)))$ is identically zero, then $g$ can be computed by an algebraic circuit of depth $\Delta' = (\Delta + O(1))$ and size $s'$ which is at most
\[
\poly(s, \deg(A), k) \cdot (\rho \log k)^{\poly(d)} \, .
\]

In particular, if $g$ cannot be computed by a size $s'$ depth $\Delta'$ circuit, then $A(T,\ki_{g,\mathcal{S}}(\vecw),\Phi(T,\ki_{g,\mathcal{S}}(\vecw)))$ is identically zero if and only if $A(T, \vecx, \Phi(T, \vecx))$ is identically zero. 
\end{theorem}
The proof of this theorem is based on some of the same principles as that of the proof of \autoref{thm:irreducibility-preservation}. However, the underlying technical details are considerably shorter and cleaner. We now discuss the proof. 

\begin{proof}[Proof of \autoref{thm:technical-theorem-1}]
    Let $\Hat{B}_j(T, \vecx, \vecw)$ be the polynomial obtained by replacing the variables $x_1, \ldots, x_{j-1}$ in the polynomial $A(T, \vecx, \Phi(T, \vecx))$ by the polynomials $g_1, \ldots, g_{j-1} \in \F[\vecw]$ respectively. From the hypothesis of the theorem, we know that $\Hat{B}_0$ is not identically zero and $\Hat{B}_n$ is identically zero. Thus, there must be a $j$ such that $\Hat{B}_{j}$ is not identically zero but $\Hat{B}_{j+1}$ is identically zero. We focus on this $j$ for the rest of the proof. 

    Note that $\Hat{B}_j$ depends on the variables $T, \vecw$ and $x_{j}, x_{j+1}, \ldots, x_{n}$ but does not depend on the variables $x_1, x_2, \ldots, x_{j-1}$. Moreover, $\Hat{B}_{j+1}$ is obtained from $\Hat{B}_j$ by substituting $x_j$ by $g_j$. For ease of notation, we refer to the variable $x_j$ as $y$ for the rest of this argument. Let $B_j$ be the polynomial obtained from $\Hat{B}_j$ by setting the variables $x_{j+1}, \ldots, x_{n}$ and $\{w_i : w_i \notin S_j\}$ to constants from $\F$ such that $B_j$ remains non-zero. Since $\Hat{B}_j$ is a non-zero polynomial, a random substitution of these variables from $\F$ (assuming it is large enough) has this property. Thus, $B_j$ only depends on the variables $T, y, \{w_i : w_i \in S_j\}$. For ease of notation, we continue to refer to the tuple $\vecw_{S_j}$ as $\vecw$.
    
    We know that when we substitute $y$ by $g_j$ in $B_j$, we end up with the identically zero polynomial. In other words, $(y-g_j)$ divides $B_j$. At this point, we would like to invoke \autoref{lem:low-deg-y-roots-of-hybrid-poly} to conclude the proof. In order to do that, we need to make sure that the hypothesis of the lemma holds, which we do now. 

    From its definition, we know that  
    \[
    \Hat{B}_j(T, \vecx, \vecw)  :=  A(T, g_1(\vecw), \ldots, g_{j-1}(\vecw), x_j, x_{j+1}, \ldots, x_n) \, .
    \]
    Now, to obtain $B_j$ from $\Hat{B_j}$, we set the variables $x_{j+1}, \ldots, x_n$ and the $\vecw$ variables outside the set $S_j$ to field constants, and rename $x_j$ as $y$. For every $i < j$, the above setting of $\vecw$ variables outside the set $S_j$ to field constants reduces each $g_i$ to a polynomial $\hat{g}_i$ that has degree at most $d$ and only depends on at most $\rho$ $\vecw$ variables in the set $|S_i \cap S_j|$. Thus, $\hat{g}_i$ can be written as a sum of at most $\binom{\rho + d}{d} \leq \rho^d$ many monomials, and hence is a depth $2$ algebraic circuit of size at most  $\rho^{O(d)}$. 
    
    For $i \in [j+1, n]$, let the variable $x_i$ be set to the field element $\alpha_i$, and let $\Hat{A}(T, \vecw, y)$ and $\Hat{\Phi}(A, \vecw, y)$ be the polynomials defined as 
    \[
        \Hat{A}(T, \vecw, y, z) := A\left(T, \hat{g}_1(\vecw), \ldots, \hat{g}_{j-1}(\vecw), y, \alpha_{j+1}, \ldots, \alpha_n, z\right) \, ,
    \]
    and 
    \[
        \Hat{\Phi}(T, \vecw, y) := \Phi\left(T, \hat{g}_1(\vecw), \ldots, \hat{g}_{j-1}(\vecw), y, \alpha_{j+1}, \ldots, \alpha_n\right) \, .
    \]
    Thus, we also get that $B_j(T, \vecw, y) = \Hat{A}(T, \vecw, y, \Hat{\Phi})$. A direct application of \autoref{lem:preserving-root-properies-under-homomorphism} tells us that $\Hat{A}(T, \vecw, y, z)$ is $T$-regularized and monic in $z$, and $\Hat{\Phi}$ is a truncated, non-degenerate, approximate $z$-root of $\Hat{A}(T,\vecw,y,z)$ of order $k$.
    Also, since the degree of each $\Hat{g}_i$ is at most $d$, we have that the total degree of $\Hat{A}$ is at most $d \cdot \deg(A)$. Since $A$ has a circuit of depth $\Delta$ and size $s$, from the discussion above, we can plug in a small depth-$2$ circuit of size at most $\rho^{O(d)}$ for each $\Hat{g}_i$, and get that $\Hat{A}$ has a circuit of depth $\Delta+2$ and size $\left(s\rho^{O(d)}\right)$. 
    
    So, the hypothesis of \autoref{lem:low-deg-y-roots-of-hybrid-poly} continue to hold if we set the polynomials $P$ and $R$ in \autoref{lem:low-deg-y-roots-of-hybrid-poly} to $\Hat{A}$ and $\Hat{\Phi}$ respectively, with size and depth bound on the circuit complexity of $\Hat{A}$ being  $\left(s\rho^{O(d)}\right)$ and $\Delta+2$ respectively. 

    Finally, since $(y-g_j(\vecw))$ divides $B_j(T, \vecw, y)$, which equals $(\Hat{A}(T, \vecw, y, \Hat{\Phi}(T, \vecw, y)))$, we get from an application of \autoref{lem:low-deg-y-roots-of-hybrid-poly} that $g_j(\vecw)$ has a depth $\Delta' = (\Delta + O(1))$ circuit of size $s'$ which is at most 
    \[
    \left(\poly(s\rho^d, \deg(\Hat{A}), k) \cdot (\log k)^{\poly(d)}\right) \leq     \left(\poly(s, \deg(A), k) \cdot (\rho \log k)^{\poly(d)}\right)\, .
    \]
    The contrapositive tells us that if every depth-$\Delta'$ circuit for $g$ requires size greater than $s'$ obtained as an upper bound above, then $A(T,\ki_{g,\mathcal{S}}(\vecw),\Phi(T,\ki_{g,\mathcal{S}}(\vecw)))$ is identically zero if and only if $A(T, \vecx, \Phi(T, \vecx))$ is identically zero. This completes the proof of the theorem. 
\end{proof}

\section{Preserving irreducibility under variable reduction} \label{sec:irreducibility-preservation}
We prove our main technical result (\autoref{thm:irreducibility-preservation}) in this section. Together with the ideas discussed in \autoref{sec:analysis-of-algorithms}, this is sufficient to complete the proof of correctness of our algorithms. To this end, we first start by showing that the irreducibility of factors of multivariate polynomials can be characterized by certain divisibility tests involving approximate power series roots of the polynomial. We then show that these divisibility tests can be reduced to polynomial identity tests for certain circuits, where the instances are built using constant-depth circuits and their approximate power series roots. A reduction from divisibility testing to PIT was shown by Forbes \cite{Forbes15} and was later used in factorization algorithms in \cite{KRS23, KRSV}. However,  our proof for this reduction here is different from that in \cite{Forbes15} and goes via a recent result of Andrews \& Wigderson \cite{AW24} that shows that the transformation between elementary symmetric polynomials and power symmetric polynomials can be done efficiently using constant-depth circuits. This alternative proof offers a clearer insight into the structure of the PIT instances and this structure if helpful in completing our proof. Finally, we show that these PIT instances can be solved deterministically in subexponential time using the Kabanets-Impagliazzo generator invoked with the explicit low degree hard polynomials in \cite{LST21}. This analysis is the core technical part of the proof and is perhaps a little surprising since while we are unable to show that the PIT instances we work with are themselves constant-depth circuits. Nevertheless, we manage to show that the generator can still be analyzed and shown to work for such circuits.  

\subsection{From irreducibility testing to divisibility tests}

Throughout this section, we assume that $P(T, \vecx, z)$ is a nonzero squarefree polynomial that is monic in $z$, $T$-regularised with respect to $z$, and it satisfies that $P(0,\vecx,z)=P(0,\veczero,z)$ is squarefree. Let $D = \deg_z P$ and let $P(T, \vecx, z) = \prod_{i \in [D]} (z - \varphi_i(T,\vecx))$ be the factorization\footnote{Such a factorization exists because each root of the univariate $P(0,\vecx,z)$ can be extended to a power series root via Newton iteration; see \cite[Theorem 4]{DSS22-closure}.} of $P(T, \vecx, z)$ where $\varphi_i(T, \vecx)$ are from the ring of power series $\overline{\F}[\vecx]\indsquare{T}$. 

\begin{lemma}
    \label{lem:factors-as-prod-of-roots}
Let $P(T, \vecx, z)$ and $\phi_1(T, \vecx), \ldots, \phi_D(T, \vecx)$ be defined as above, and let $k > \deg_T(P)$. For any $G(T, \vecx, y) \in \overline{\F}(T, \vecx, z)$ that divides $P(T, \vecx, z)$, there is some subset $U \subseteq [D]$ such that 
\[
G(T, \vecx, y)  = \prod_{i\in U} (z - \phi_i(T, \vecx)) \trunc T^k. 
\]
and $\prod_{i\in U} (z - \phi_i(T, \vecx))$ is in fact a polynomial (and not just a power series) that equals $G(T, \vecx, y)$. 
\end{lemma}
\begin{proof}
Since $P(T, \vecx, z)$ is squarefree and we have the factorisation $P(T, \vecx, z) = \prod_{i\in [D]} (z - \phi_i(T, \vecx))$ in $\F[\vecx, z]\indsquare{T}$ and hence any polynomial factor $G(T, \vecx, z)$ must be $\prod_{i\in U} (z - \phi_i(T, \vecx))$ for some $U \subseteq [D]$. Since $\deg_T G \leq \deg_T P < k$, we have
\begin{align*}
    G(T, \vecx, y) &= G \trunc T^k \\
        & = \prod_{i\in U} (z - \phi_i(T, \vecx)) \trunc T^k\qedhere
\end{align*}
\end{proof}

For our algorithm, we would be interested in divisibility over $\F$ and not over $\overline{\F}$. The following lemma gives some sufficient conditions for when the above polynomial $Q_S(T, \vecx, z)$ has coefficients in $\F$. 

\begin{lemma}
    \label{lem:coefficients-of-QS-polynomial}
    Let $P(T, \vecx, z)$ be a polynomial and $\phi_1(T,\vecx), \ldots, \phi_D(T, \vecx)$ be power series as defined above. For each $i \in [D]$, let $\alpha_i = \phi_i(0, \veczero)$, and let $\mathcal{Z} = \set{\alpha_1,\ldots, \alpha_D}$ (which are the roots of $P(0, \vecx, z) = P(0, \veczero, z)$ in the algebraic closure $\overline{\F}$).  Fix some $k \geq 0$. For a subset $S \subset [D]$, let
    \[
    Q_{S}(T, \vecx, z) := \prod_{i \in S}(z-\varphi_i(T,\vecx))\trunc T^k
    \]
    for some $k \geq 0$. 

    If the polynomial $g_S(z) = \prod_{i \in S} (z - \alpha_i) = Q_S(0, \veczero, z) = Q_S(0, \vecx, z)$
    has all its coefficients in the base field $\F$, then so does $Q_S(T, \vecx, z)$. 
\end{lemma}
\begin{proof}
Define the following rational functions $\Gamma^{(i)}(T, \vecx, \zeta)$ for $i \geq 0$ as follows:
\begin{align*}
\Gamma^{(0)}(\zeta) & := \zeta\\
\Gamma^{(i+1)}(\zeta) & := \Gamma^{(i)}(\zeta) - \frac{P(T, \vecx, \Gamma^{(i)}(\zeta))}{\partial_z P(0, \veczero, \zeta)}\quad\text{for $i > 0$}.
\end{align*}
Note that $\Gamma^{(i)}(\alpha)$ is in fact a polynomial in $T, \vecx$ for every $\alpha$ such that $\partial_z P(0, \veczero, \alpha) \neq 0$. By \cref{lem:newton-iteration-linear}, for any $\alpha_i \in \mathcal{Z}$, we have that $\Gamma^{(k)}(\alpha_i) = \phi_i \bmod{T^k}$ and therefore
\[
Q_S(T, \vecx, z) = \prod_{i \in S}(z-\Gamma^{(k)}(\alpha_i))\trunc T^k 
\]
Therefore, every coefficient of $Q_S(T, \vecx, z)$ is a symmetric function of the set $\setdef{\alpha_i}{i \in S}$. Since $g(z) = \prod_{i \in S} (z - \alpha_i)$ is assumed to have coefficients in the base field $\F$, this implies that every elementary symmetric function of the set $\setdef{\alpha_i}{i \in S}$ is in $\F$. By the fundamental theorem \RPnote{add a ref?} of symmetric polynomials, this implies that every coefficient of $Q_S(T, \vecx, z)$ is also in $\F$. 
\end{proof}

With the above two lemmas, we can now characterise the factors (over $\F$) of $P(T, \vecx, z)$ using the polynomials $Q_S(T, \vecx, z)$ defined above. 

\begin{lemma}\label{lem:identifying-conjugates-to-divisibility-testing-II}
    Let $P(T, \vecx, z)$ and $\phi_1(T, \vecx), \ldots, \phi_D(T, \vecx)$ be defined as above, and let $k > \deg_T P$. Let $\alpha_i = \phi_i(0, \veczero)$ for each $i \in [D]$. For a subset $S \subseteq [D]$, define the polynomial 
    \[
    Q_S(T, \vecx, z) = \prod_{i \in S} (z - \phi_i(T, \vecx)) \trunc T^k.
    \]
    Then, the set of factors of $P$ over the base field $\F$ is the same as 
    \[
        \setdef{Q_U(T, \vecx, z) \in \overline{\F}[T, \vecx, z]}{\begin{array}{c}U \subseteq [D] \text{ with } Q_U(T, \vecx, z) \text{ dividing }P\\\text{and } Q_U(0, \vecx, z) = Q_U(0, \veczero, z) \in \F[z] \end{array}}
    \]
\end{lemma}
\begin{proof}
    Any $G(T, \vecx, z) \in \F[T, \vecx, z]$ that is a factor of $P(T, \vecx, z)$, by \cref{lem:factors-as-prod-of-roots}, is equal to $Q_U(T, \vecx, z) = \prod_{i\in U} (z - \phi_i(T, \vecx))$ for some $U \subseteq [D]$. Additionally, if $G(T, \vecx, z) \in \F[T, \vecx, z]$, then $G(0, \vecx, z) = Q_U(0, \vecx, z) \in \F[z]$. 

    For the other direction, consider any $Q_U(T, \vecx, z) \in \overline{\F}[T, \vecx, z]$ that divides $P(T, \vecx, z)$ and also satisfies that $Q_U(0, \veczero, z) = \prod_{i \in U}(z - \alpha_i)$ is a polynomial with coefficients in $\F$. By \cref{lem:coefficients-of-QS-polynomial}, we have that $Q_U(T, \vecx, z)$ must also be a polynomial in $\F[T, \vecx, z]$ and hence is a factor of $P(T, \vecx, z)$ over $\F$. 
\end{proof}

\noindent
This immediately implies the following \emph{certificate} for irreducibility. 

\begin{corollary} \label{cor:identifying-conjugates-to-divisibility-testing-III}
Let $P(T, \vecx, z) \in \F[T, \vecx, z]$ and $\phi_1(T, \vecx), \ldots, \phi_D(T, \vecx)$ be defined as above, and let $k > \deg_T(P)$. Then $P(T, \vecx, z)$ is irreducible over $\F$ if and only if for every $\emptyset \neq S \subsetneq [D]$ such that $\prod_{i\in S}(z - \alpha_i)$ has all coefficients in $\F$, we have that $Q_S(T, \vecx, z) \in \F[T, \vecx, z]$ does not divide $P(T, \vecx, z)$ where
\[
Q_S(T, \vecx, z) = \prod_{i \in S} (z - \phi_i(T, \vecx)) \trunc T^k.
\]
\end{corollary}

\begin{remark}
    Essentially the same proof as \cref{lem:identifying-conjugates-to-divisibility-testing-II} also characterizes the set of factors of $P(T,\vecx,z)$ over the algebraic closure $\overline{\F}$, and thus gives a certificate for \emph{absolute irreducibility} using divisibility tests over $\overline{\F}$. 
\end{remark}

\begin{corollary} \label{cor:identifying-conjugates-to-divisibility-testing-III-algebraic-closure}
Let $P(T, \vecx, z) \in \F[T, \vecx, z]$ and $\phi_1(T, \vecx), \ldots, \phi_D(T, \vecx)$ be defined as above, and let $k > \deg_T(P)$. Then $P(T, \vecx, z)$ is absolutely irreducible if and only if for every $\emptyset \neq S \subsetneq [D]$, we have that $Q_S(T, \vecx, z) \in \overline{\F}[T, \vecx, z]$ does not divide $P(T, \vecx, z)$ where
\[
Q_S(T, \vecx, z) = \prod_{i \in S} (z - \phi_i(T, \vecx)) \trunc T^k.
\]
\end{corollary}

\autoref{cor:identifying-conjugates-to-divisibility-testing-III} together with the reduction from divisibility testing to PIT in the next section gives us a reduction from irreducibility testing to PIT. However, we note that \label{cor:identifying-conjugates-to-divisibility-testing-III} gives exponentially many (in degree $D$) divisibility testing instances, and thus it isn't immediately clear if this reduction can be useful for obtaining subexponential time irreducible testing algorithms. It is for this reason that our algorithm for factorization does not proceed directly using this reduction and only uses the corollary (and the previous lemmas) in its analysis.

\subsection{From divisibility testing to polynomial identity tests}\label{sec:divtest-to-pit}

Let us recall the standard definitions of the power sum symmetric polynomials and the elementary symmetric polynomials, particularly in the context of roots of a univariate polynomial. 

\begin{definition}
For any monic univariate polynomial $P = \prod_{j\in [d]} (z-\alpha_j) \in \F[z]$ and natural number $i$, $\psym_i(P)$ and $\esym_i(P)$ are the power sum symmetric polynomial $\sum_j \alpha_j^i$ and the elementary symmetric polynomial $\sum_{S\subseteq [d]: |S|=i} \prod_{j \in S} \alpha_j$ of degree $i$ respectively in the multiset of roots of $P$.     
\end{definition}

The following important lemma from a recent beautiful work of Andrews and Wigderson \cite{AW24} (but also present in earlier works such as \cite{SW01}) shows that there is a constant-depth circuit that computes the power sum symmetric polynomials of a multiset from the elementary symmetric polynomials of the multiset; similarly, there is a constant-depth circuit that computes the elementary symmetric polynomials of a multiset from the power sum symmetric polynomials of the multiset.
\begin{lemma}[{\cite[Lemma 3.4, Lemma 3.6]{AW24}}]\label{lem:esym-psym-transform-constant-depth}
    Let $\F$ be any field of characteristic zero and $n \in \N$ be any natural number. Then, there is a constant-depth circuit of size and degree $\poly(n)$ that takes the $n$-variate polynomials $\{\esym_i(\vecx) : i \in [n]\}$ as inputs and outputs the polynomials $\{\psym_i(\vecx) : i \in [n]\}$. 

    Similarly, there is a constant-depth circuit of size and degree $\poly(n)$ that takes the $n$-variate polynomials $\{\psym_i(\vecx) : i \in [n]\}$ as inputs and outputs the polynomials $\{\esym_i(\vecx) : i \in [n]\}$. 
\end{lemma}

In \cite[Section 7.2]{Forbes15}, Forbes showed that the question of deterministic divisibility testing for polynomials from a certain complexity class can be reduced to deterministic PIT for polynomials from a related complexity class. This reduction was extensively used in \cite{KRS23} and \cite{KRSV}. The following version offers a different approach for reducing divisibility testing to PIT, using \cref{lem:esym-psym-transform-constant-depth}.
\begin{lemma}
    \label{lem:divisibility-test-to-PIT}
    Let $D \geq t \geq 0$ be integer parameters. Let $\F$ be any field of characteristic zero. Then, there is a constant-depth $\poly(D,t)$-sized circuit $\DivTest_{D,t}$ on $D+t+1$ variables, that takes $(D+t)$ inputs labelled $f_0, \ldots, f_{D-1} \in \F$ and $g_0,\ldots, g_{t-1} \in \F$ respectively, such that 
    \[
        \DivTest_{D,t}(z,f_0,\ldots, f_{D-1}, g_0, \ldots, g_{t-1}) = 0
    \]
    if and only if the polynomial $f(z) = f_0 + f_1 z + \cdots + f_{t-1} z^{t-1} + z^t$ divides the polynomial $g(z) = g_0 + g_1 z + \cdots + g_{D-1} z^{D-1} + z^D$. 
\end{lemma}
\begin{proof}
    By \cref{lem:esym-psym-transform-constant-depth}, we can construct multi-output circuits $\PSymToESym_n$ and $\ESymToPSym_n$ of size $\poly(n)$ and depth $O(1)$ such that for every choice of $\alpha_1,\ldots, \alpha_n$ we have
    \begin{align*}
        \ESymToPSym_n(E_1,\ldots, E_n) & = (P_1,\ldots, P_n)\\
        \PSymToESym_n(P_1,\ldots, P_n) & = (E_1,\ldots, E_n)\\
        \text{ where } E_i & = \esym_i(\alpha_1, \ldots, \alpha_n)\\
        \text{ and } P_i   & = \psym_i(\alpha_1,\ldots, \alpha_n).
    \end{align*}

    Let $f(z) = \prod_{i=1}^D (z - \alpha_i)$ where $S_f = \set{\alpha_1,\ldots, \alpha_D}$ is the multiset of roots of $f$ (in the algebraic closure of $\F$), and let $S_g$ be the multiset of roots of $g(z)$. 

    Note that $e^{(f)}_i := \esym_{i}(S_f) = (-1)^{i} f_{D-i}$ for $i = 1,\ldots, D$ and $e^{(g)}_i = \esym_{i}(S_g) = (-1)^i g_{t-i}$ for $i = 1\ldots, t$. 
    Define the \emph{pseudo-quotient} of $f$ and $g$, referred to by $\tilde{h}$ as follows:
    \begin{align*}
        \tilde{h}(z) &= z^{D-t} - z^{D-t-1}e_{1}^{(h)} + \cdots + (-1)^{D-t} e_{D-t}^{(h)}\\
        \text{where } e^{(h)}_i &= \PSymToESym_{D-t}\inparen{r_1, \ldots, r_{D-t}}\\
        \text{where } r_j & = \inparen{\ESymToPSym_D\inparen{e_1^{(f)}, \ldots, e^{(f)}_D} - \ESymToPSym_t\inparen{e_1^{(g)}, \ldots, e^{(g)}_t}}_j
    \end{align*}

    If $g \mid f$, then the multiset $S_f$ is the union of the multiset $S_g$ and the multiset $S_h$ of roots of $h = f/g$. In that case, for each $i = 1,\ldots, D-t$, we have 
    \begin{align*}
        \psym_i(S_f) & = \psym_i(S_g) + \psym_i(S_h)\\
        \implies \psym_i(S_h) & = \psym_i(S_f) - \psym_i(S_g).
    \end{align*}
    Therefore, $\tilde{h}(z) = h(z)$ if $g \mid f$. \\

    The circuit $\DivTest_{D,t}(f_0, \ldots, f_{D-1}, g_0,\ldots, g_{t-1})$ outputs the polynomial $f(z) - \tilde{h}(z) \cdot g(z)$ where $\tilde{h}(z)$ is computed as stated above. By construction, this is a circuit of size $\poly(D, t)$ and depth $O(1)$. Furthermore, as argued above, if $g \mid h$ then $\tilde{h}(z) = h(z) = f(z)/g(z)$ and hence the above computes the zero polynomial. If $g(z) \nmid f(z)$, then $f(z) - \tilde{h}(z) g(z)$ is nonzero for every choice of $\tilde{h}(z)$ and hence the above is nonzero polynomial. 
\end{proof}

The following lemma uses \cref{lem:divisibility-test-to-PIT} to show that we can reduce the question of checking whether a candidate factor (constructed using various approximate roots from Newton iteration) is an actual factor, to a PIT instance.
\begin{lemma}\label{lem:constant-depth-PIT-instance-for-conjugates}
Let $\F$ be a field of characteristic zero. Let $P(T, \vecx, z) \in \F[T, \vecx, z]$ be a polynomial that is monic in $z$. Let $R_1(T, \vecx), R_2(T, \vecx), \ldots, R_{\ell}(T, \vecx) \in \F[T, \vecx]$ be polynomials of $T$-degree at most $(k-1)$. 

Then, there exists a circuit $\DivTest_{\deg_z(P),\ell}$ on at most $(\deg_z(P) + \ell + 1)$-variables and size and degree at most $\poly(\deg(P), \ell, k)$ that is computable by a depth-$(\Delta + O(1))$  such that \[
Q(T,\vecx,z):= \left(\prod_{i \in [\ell]}(z - R_i(T, \vecx)) \right) \trunc T^k
\]
divides $P(T,\vecx, z)$ if and only if 
\[
\DivTest_{\deg_z(P),\ell}\left(z,P_0(T,\vecx), \ldots, P_{\deg_z(P)}(T,\vecx), (\esym_1(\mathcal{R}) \trunc T^k), \ldots, (\esym_{\ell}(\mathcal{R}) \trunc T^k)\right) \equiv 0 \, , 
\]
where, for every $i$, $P_i(T,\vecx)$ is the coefficient of $z^i$ in $P$, when viewing it as a univariate in $z$ with coefficients in the ring $\F[T,\vecx]$ and $\mathcal{R} = \{R_1, R_2, \ldots, R_{\ell}\}$.  
\end{lemma}

\begin{proof}
    This almost immediately follows from \autoref{lem:divisibility-test-to-PIT} since each of the polynomials \\ $\esym_1(\mathcal{R}) \trunc T^k, \dots,\esym_{\ell}(\mathcal{R}) \trunc T^k$ are the coefficients of  $\left(\prod_{i \in [\ell]}(z - R_i(T, \vecx)) \right) \trunc T^k$. As it is, the $\DivTest$ circuit would capture divisibility when the coefficients of the input (univariate) polynomials (in the variable $z$) are from a field, but since we are dealing with monic polynomials, Gauss' Lemma (\cref{lem:gauss}) ensures that $Q(T,\vecx,z)$ divides $P(T,\vecx,z)$ in $\F(T,\vecx)[z]$ if and only if $Q(T,\vecx,z)$ divides $P(T,\vecx,z)$ in $\F(T,\vecx)[z]$. The theorem now follows. 
\end{proof}

\begin{lemma}\label{lem:complexity-of-modular-esym-of-R}
Let $\F$ be a field of characteristic zero. Let $R_1(T, \vecx), R_2(T, \vecx), \ldots, R_{\ell}(T, \vecx) \in \F[T, \vecx]$ be polynomials of $T$-degree at most $(k-1)$.

Then, there is a multi-output algebraic circuit $\Hat{C}$ with depth $(\Delta + O(1))$ and size $\poly(\ell, k, \deg(P))$ that takes as input the polynomials of the form $\{R_i(\alpha_j T, \vecx) : i \in [\ell], j \in [\poly(\ell,k,\deg(P))]\}$ and outputs $\left(\esym_i(R_1, R_2, \ldots, R_{\ell}) \trunc T^k\right)$ for every $i \in [\ell]$. Here $\{\alpha_j : j \in [\poly(\ell,k,\deg(P))]\}$ are  elements of $\F$. 
\end{lemma} 
\begin{proof}
    From \autoref{lem:esym-psym-transform-constant-depth}, we get that there is a multi-output constant-depth circuit $C$ on $\ell$ variables that can be computed by a constant-depth circuit of size and degree $\poly(\ell)$ such that 
    \[
    C(\psym_1(\mathcal{R}), \ldots, \psym_n(\mathcal{R})) = (\esym_1(\mathcal{R}), \ldots, \esym_n(\mathcal{R})) \, ,
    \]
    where, $\mathcal{R} = (R_1, R_2, \ldots, R_{\ell})$. For ease of notation, we just focus on the computation of one of the $\esym_j(\mathcal{R})$ and consider the equality 
    \[
    \esym_j(\mathcal{R}) =     C(\psym_1(\mathcal{R}), \ldots, \psym_n(\mathcal{R})) \, .
    \]
    The above equality immediately gives us a constant-depth circuit that takes $\mathcal{R}$ as input and outputs their elementary symmetric polynomials. However, the goal is to compute these things modulo $T^k$. 
    
    To recover $\esym_j(\mathcal{R}) \trunc T^k$ from the above equality, we think of each $R_i(T,\vecx)$ as a polynomial in $\F[\vecx][T]$ and apply \autoref{cor:interpolation-consequences} with respect to the variable $T$ since we want to extract all the monomials that have $T$-degree at most $k$. This only incurs a polynomial blow up in size and an additive constant increase in depth. Moreover, the new circuit can be seen to be taking as inputs polynomials of the form $\{R_i(\alpha_j T, \vecx) : i \in [\ell], j \in [\poly(\ell,k,\deg(P))]\}$. 
\end{proof}
The following is an immediate consequence of \autoref{lem:complexity-of-modular-esym-of-R} and \autoref{lem:constant-depth-PIT-instance-for-conjugates}. 

\begin{corollary}\label{cor:constant-depth-PIT-instance-for-conjugates-refined}
Let $\F$ be a field of characteristic zero. Let $P(T, \vecx, z) \in \F[T, \vecx, z]$ be a polynomial that is monic in $z$. Let $R_1(T, \vecx), R_2(T, \vecx), \ldots, R_{\ell}(T, \vecx) \in \F[T, \vecx]$ be polynomials of $T$-degree at most $(k-1)$. 

Then, there exists a circuit $C$ on at most $(\poly(\ell, k, \deg(P)))$-variables and size and degree at most $\poly(\deg(P), \ell, k)$ that is computable by a depth $(\Delta + O(1))$ circuit such that \[
\left(\prod_{i \in [\ell]}(z - R_i(T, \vecx)) \right) \trunc T^k
\]
divides $P(T,\vecx, z)$ if and only if 
\[
C\left(z,P_0(T,\vecx), \ldots, P_{\deg_z(P)}(T,\vecx), \mathcal{R}\right) \equiv 0 \, , 
\]
where, for every $i$, $P_i(T,\vecx)$ is the coefficient of $z^i$ in $P$, when viewing it as a univariate in $z$ with coefficients in the ring $\F[T,\vecx]$ and $\mathcal{R} = \left(R_i(\alpha_j T, \vecx) : i \in [\ell], j \in [\poly(\ell,k,\deg(P))]\right)$, with each $\alpha_j$ being an element of $\F$.  
\end{corollary}

\subsection{Building up to the proof of \autoref{thm:irreducibility-preservation}}
\cref{thm:technical-theorem-2} is our main technical theorem building up to \cref{thm:irreducibility-preservation}. It shows that for the kind of circuits that show up when testing the divisibility of candidate factors, if composing such a circuit, say $C$, with the KI generator -- instantiated with a polynomial $g$ and a design $\mathcal{S}$ -- breaks the nonzeroness of $C$, then the polynomial $g$ must be ``easy''.
\begin{theorem}\label{thm:technical-theorem-2}
Let $k > 1$ be any natural number and for every $i \in [\ell]$, let $A_i(T, \vecx, z) \in \F[T, \vecx, z]$ and $\Phi_i(T, \vecx) \in \F[T, \vecx]$ be polynomials satisfying the following properties. 
\begin{itemize}
    \item Each $A_i(T, \vecx, z)$ is computable by a depth $\Delta$ circuit of size $s$
    \item Each $A_i(T, \vecx, z)$ is $T$-regularized and is monic in $z$. 
    \item For every $i\in [\ell]$, $\Phi_i(T,\vecx)$ is a truncated, non-degenerate, approximate $z$-root of $A_i(T,\vecx,z)$ of order $k$.
\end{itemize}

Let $\mathcal{S}$ be a $(n,\sigma,\mu,\rho)$-design and let $g$ be a $\sigma$-variate polynomial of degree $d$. For a new tuple $\vecw$  of  $\mu$ variables distinct from $T, \vecx, z$, let  us define the polynomial map $\ki_{g,\mathcal{S}}$ using \cref{def:KI-generator}. Then the following is true.

For any circuit $C$ of depth $\Delta$ and size $s$ such that the polynomial 
\[
C'(T, \vecx) := C(T, \vecx, \Phi_1, \Phi_2, \ldots, \Phi_{\ell}) 
\]
is non-zero, if $C'(T,\ki_{g,\mathcal{S}}(\vecw))$ is identically zero, then $g$ can be computed by an algebraic circuit of depth $O(\Delta)$ and size 
\[
\poly(s, \deg(C'), k) \cdot (\rho \log k)^{\poly(d)} \, .
\]
\end{theorem}
\begin{proof}
    The proof proceeds along the lines of the proof of \autoref{thm:technical-theorem-1}. In particular, we again apply the hybrid argument to the sequence of substitutions for $x_i$ by $g_i$ and arrive at the point $j$ where the substituted polynomial becomes identically zero for the first time. We again rename the variable $x_j$ as $y$ and set any remaining $\vecx$ variables and $\vecw$ variables outside the set $S_j$ to field constants while preserving the non-zeroness of the polynomial before the $j^{th}$ substitution. We use $\Hat{C}$ to denote the result of applying the sequence of substitutions to the circuit $C$. We use $\Hat{C'}$ to denote the result of applying the sequence of substitutions to the circuit $C'$. Clearly, it is of the form 
    \[
    \Hat{C'}(T,\vecw, y) = C(T, \Hat{g}_1(T,\vecw), \ldots, \Hat{g}_{j-1}(T,\vecw), y, \vecb, \Hat{\Phi}_1, \ldots, \Hat{\Phi}_{\ell}) = \Hat{C}(T,\vecw,y,\Hat{\Phi}_1, \ldots, \Hat{\Phi}_{\ell}) \, ,
    \]
where each $\Hat{\Phi}_i(T,\vecw,y)$ is a polynomial in $\F[T,\vecw_{S_j},y]$ that satisfies \[
\Hat{\Phi}_i(T,\vecw,y) = \Phi_i(T, \Hat{g}_1(T,\vecw), \ldots, \Hat{g}_{j-1}(T,\vecw), y, \vecb) \, 
\]
for field constants $\vecb$ and each $\Hat{g}_i$ is obtained from $g_i$ by setting the $\vecw$ variables outside the set $S_j$ to field constants according to $\vecb$ and hence depends on only $|S_i \cap S_j| \leq \rho$ variables. Thus, each $\Hat{g_i}$ has a depth-$2$ circuit of size at most $\rho^{O(d)}$.  We also know that $y - g_j(\vecw)$ divides $\Hat{C'}$. From this, we would like to conclude that $g_j$ has a small constant-depth circuit. 

We again follow the proof of \autoref{thm:technical-theorem-1}, and reduce to the case that $(y-g_j)$ is a factor of multiplicity one. To this end, we will have to work with an appropriately high enough order derivative of $\Hat{C'}$ with respect to $y$ (depending on the multiplicity). From \cref{cor:interpolation-consequences}, we get that this derivative is in the $\F[y]$-span of polynomials of the form $\Hat{C'}(T, \vecw, \beta_i)$ with $\beta_i \in \F$ and $i \in [\deg_y(\Hat{C'})]$, and the $y$-degree of the weights in this linear combination is at most $\deg_y(\Hat{C'})$. Let us denote the derivative by $B$. Thus, we have that 
\[
B(T, \vecw, y) = \sum_{i = 0}^{\deg_y(\Hat{C'})} \Hat{C'}(T, \vecw, \beta_i) \times \Lambda_i(y) \, , 
\]
for univariate polynomials $\Lambda_i(y)$ of degree at most $\deg_y(\Hat{C'})$, $B$ is non-zero and satisfies
\[
B(T, \vecw, g_j(\vecw)) \equiv 0 \, , 
\]
and 
\[
\frac{\partial B}{\partial y}(T, \vecw, g_j(\vecw)) \neq 0 \, . 
\]
By shifting the $\vecw$ variables if needed, we can assume without loss of generality that 
\[
\frac{\partial B}{\partial y}(T, \mathbf{0}, g_j(\mathbf{0})) \not\equiv 0 \, . 
\]
Thus, from \autoref{lem:CKS-updated-rational-in-T-new}, we get that there is a $\kappa \in \F$ and a polynomial $Q$ of degree at most $d$ on $(d+1)$ variables, that satisfies 
\[
g_j(\vecw) \equiv Q(h_0(\kappa,\vecw), h_1(\kappa,\vecw), \ldots, h_d(\kappa,\vecw)) \mod \langle \vecw \rangle^{d+1} \, ,
\]
where for every $i$, 
\[
h_i(T,\vecw) = \frac{\partial B}{\partial y^i}(T,\vecw, g_j(\mathbf{0})) - \frac{\partial B}{\partial y^i}(T,\veczero, g_j(\mathbf{0})) \trunc \langle \vecw \rangle^{d+1} \, .
\]

Given the degree of $Q$ and the number of variables, we get that it has a depth-$2$ circuit of size at most $\exp(O(d))$. The next claim shows that each $h_i$ has a constant-depth circuit of small size. 
\begin{claim}\label{clm:size-of-hi}
Each $h_i$ can be computed by a circuit of depth $O(\Delta)$ and size 
\[
s' \leq \poly(s, \deg(C'), k) \cdot (\rho \log k)^{\poly(d)}
\]
\end{claim}
We first use the claim to complete the proof of the lemma, and discuss the proof of the claim. 

We take the circuits for $h_i(T,\vecw)$ --- given by \autoref{clm:size-of-hi}, denoted by $V_i(T,\vecw)$ --- and consider the circuit $Q(V_0(\kappa,\vecw), V_1(\kappa,\vecw),\dots, V_d(\kappa,\vecw))$. Clearly, it is of depth $O(\Delta)$ and size $s' \leq \poly(s, \deg(C'), k) \cdot (\rho \log k)^{\poly(d)}$ and satisfies 
\[
g_j(\vecw) \equiv Q(V_0(\kappa,\vecw), V_1(\kappa,\vecw),\dots, V_d(\kappa,\vecw)) \mod \langle \vecw \rangle^{d+1} \, . 
\]

We can now apply \autoref{cor:interpolation-consequences} to obtain the constant-depth circuit for $g_j(\vecw)$ with only an additive increase in depth and a polynomial blow up in the size, thereby giving us the theorem. 
\end{proof}

\begin{proof}[Proof of \autoref{clm:size-of-hi}]
    From \autoref{lem:preserving-root-properies-under-homomorphism}, we know that for any $\beta \in \F$, the polynomials $\Hat{A}_i(T,\vecw, \beta)$ and $\Hat{\Phi}_i(T,\vecw, \beta)$, obtained from $A$ and $\Phi$ using the substitutions in the above discussion and setting $y = \beta$, continue to satisfy that properties satisfied by $A_i, \Phi_i$ in the hypothesis of the lemma. 

    Thus, for any field constant $\beta$, we can invoke \autoref{lem:low-deg-w-components-of-roots-mod-T} for each $\Hat{A}_i(T,\vecw,\beta)$ and $\Hat{\Phi}_i(T,\vecw,\beta)$ to get that $(\Hat{\Phi}_i(T,\vecw,\beta) \trunc \langle \vecw \rangle^{d+1})$ can be computed by a circuit of depth $(\Delta + O(1))$ and size $s'$. We now plug these circuits into the input gates of the constant-depth circuit $\Hat{C}$ to get a circuit of depth  at most $O(\Delta)$ and size $O(s')$ that computes the polynomial 
    \[
    \Hat{C}(T, \vecw, \beta, (\Hat{\Phi}_1(T,\vecw,\beta) \trunc \langle \vecw \rangle^{d+1}), \ldots,  (\Hat{\Phi}_{\ell }(T,\vecw,\beta) \trunc \langle \vecw \rangle^{d+1})) \, .
    \]
    Now, applying \autoref{cor:interpolation-consequences} with respect to the $\vecw$ variables gives us a circuit of depth $O(\Delta)$ and size $\poly(s')$ that computes $\left(\Hat{C'}(T,\vecw,\beta) \trunc \langle \vecw \rangle^{d+1}\right)$, since by definition $\Hat{C'}(T,\vecw,y)$ equals $\Hat{C}(T, \vecw, y, \Hat{\Phi}_1, \ldots, \Hat{\Phi}_{\ell})$.  

    Similar to \autoref{lem:low-deg-y-roots-of-hybrid-poly}\footnote{Refer to the part of the proof where we show that $h_i(T,\vecw)$ can be computed by constant-depth circuits.}, the claim now follows from the definition of $B$, \cref{cor:interpolation-consequences} and the fact that the degree of $B$ in $T$ is at most $\deg(C')$. 
\end{proof}
We now combine the machinery in \autoref{sec:divtest-to-pit} with \autoref{thm:technical-theorem-2} to get the following theorem that is directly used in the proof of \autoref{thm:irreducibility-preservation}.

\begin{theorem} \label{thm:ki-preserves-divisibility}
Let $n \in \N$ be sufficiently large, and let $\vecx = (x_1, \dots, x_n)$. Let $\F$ be a field of characteristic zero. Let $k>1$ be any natural number. Let $P(T, \vecx, z) \in \F[T, \vecx, z]$  and $R_1(T, \vecx), R_2(T, \vecx), \ldots, R_{\ell}(T, \vecx) \in \F[T, \vecx]$ be polynomials with the following properties. 
\begin{itemize}
    \item $P(T, \vecx, z)$ is computable by a depth $\Delta$ circuit of size $s\leq \poly(n)$
    \item $P(T, \vecx, z)$ is $T$-regularized and is monic in $z$. 
    \item For every $i\in [\ell]$, $R_i(T,\vecx)$ is a truncated, non-degenerate, approximate $z$-root of $P(T,\vecx,z)$ of order $k$.
\end{itemize}

Let $Q(T,\vecx,z) := \left(\prod_{i \in [\ell]}(z - R_i(T, \vecx)) \right) \trunc T^k$.

Let $g$ be a $\sigma$-variate degree $d$ polynomial and $\mathcal{S}$ be a $(n,\sigma,\mu,\rho)$-design. Let $\ki_{g,\mathcal{S}}: \F^{\mu} \to \F^{n}$ be the polynomial map in \autoref{def:KI-generator} defined using $g$ and $\mathcal{S}$. Then, the following is true.

If $Q(T,\vecx,z)$ does not divide $P(T,\vecx, z)$ but $Q(T,\ki_{g,\mathcal{S}}(\vecw),z)$ divides $P(T,\ki_{g,\mathcal{S}}(\vecw),z)$, then $g$ can be computed by an algebraic circuit of depth $\Delta' = O(\Delta)$ and size at most
\[
    s' = \poly(s,\ell, k, \deg(P)) \cdot (\rho \log k)^{\poly(d)} \, .   
\]

In particular, if any depth-$\Delta'$ circuit for $g$ requires size greater than $s'$, then $Q(T,\vecx,z)$ divides $P(T,\vecx, z)$ if and only if $Q(T,\ki_{g,\mathcal{S}}(\vecw),z)$ divides $P(T,\ki_{g,\mathcal{S}}(\vecw),z)$
\end{theorem}
\begin{proof} 
    If $Q(T,\vecx,z)$ divides $P(T,\vecx,z)$, then $Q(T,\ki_{g,\mathcal{S}}(\vecw),z)$ certainly divides $P(T,\ki_{g,\mathcal{S}}(\vecw),z)$. 
    
    Suppose $Q(T,\vecx,z)$ does not divide $P(T,\vecx,z)$. \Cref{cor:constant-depth-PIT-instance-for-conjugates-refined} tells us that equivalently, there exists a circuit $C$ on at most $\poly(\ell, k, \deg(P))$ variables, and size and degree at most $\poly(\ell, k, \deg(P))$, that is computable by a depth $(\Delta+O(1))$ circuit such that 
    \[
        C\left(z,P_0(T,\vecx), \ldots, P_{\deg_z(P)}(T,\vecx), \mathcal{R}\right) \not\equiv 0 \, , 
    \]
    where, for every $i$, $P_i(T,\vecx)$ is the coefficient of $z^i$ in $P$ when viewing it as a univariate in $z$ with coefficients in the ring $\F[T,\vecx]$, and $\mathcal{R} = \left(R_i(\alpha_j T, \vecx) : i \in [\ell], j \in [\poly(s,k,\deg(P))]\right)$, with each $\alpha_j$ being an element of $\F$.

    If $Q(T,\ki_{g,\mathcal{S}}(\vecw),z)$ divides $P(T,\ki_{g,\mathcal{S}}(\vecw),z)$, then 
    \[
        C\left(z,P_0(T,\vecx)\circ \ki_{g,\mathcal{S}}(\vecw), \ldots, P_{\deg_z(P)}(T,\vecx)\circ \ki_{g,\mathcal{S}}(\vecw), \mathcal{R}\circ \ki_{g,\mathcal{S}}(\vecw)\right) \equiv 0 \, , 
    \]
    where $\mathcal{R}\circ \ki_{g,\mathcal{S}}(\vecw) := \{R\circ \ki_{g,\mathcal{S}}(\vecw): R \in \mathcal{R}\}$. Observe that \[ C\left(z,P_0(T,\vecx)\circ \ki_{g,\mathcal{S}}(\vecw), \ldots, P_{\deg_z(P)}(T,\vecx)\circ \ki_{g,\mathcal{S}}(\vecw), \mathcal{R}\circ \ki_{g,\mathcal{S}}(\vecw)\right) \] is the same as \[ C\left(z,P_0(T,\vecx), \ldots, P_{\deg_z(P)}(T,\vecx), \mathcal{R}\right) \circ \ki_{g,\mathcal{S}}(\vecw). \]
    Moreover, \cref{lem:preserving-root-properies-under-homomorphism} tells us that each $R_i(\alpha_jT,\vecx)$ is a truncated, non-degenerate, approximate $z$-root of $P(\alpha_jT,\vecx,z)$ of order $k$, where $P(\alpha_jT,\vecx,z)$ is $T$-regularized and monic in $z$.

    Thus, we can apply \cref{thm:technical-theorem-2} to $C\left(z,P_0(T,\vecx), \ldots, P_{\deg_z(P)}(T,\vecx), \mathcal{R}\right)$ to conclude that $g$ can be computed by an algebraic circuit of depth $\Delta' = O(\Delta)$ and size $s'$ which is at most 
    \[
        s' := \poly(s,\ell, k, \deg(P)) \cdot (\rho \log k)^{\poly(d)} \, .   
    \]

    The contrapositive tells us that if every depth-$\Delta'$ circuit for $g$ requires size greater than $s'$ obtained as an upper bound above, then $Q(T,\vecx,z)$ divides $P(T,\vecx, z)$ if and only if $Q(T,\ki_{g,\mathcal{S}}(\vecw),z)$ divides $P(T,\ki_{g,\mathcal{S}}(\vecw),z)$.
\end{proof}

\subsection{Proof of \autoref{thm:irreducibility-preservation}}
We are now ready to prove \autoref{thm:irreducibility-preservation}. We start by recalling some notation that we use. 

Let $\mathcal{G}=\{g_m\}_{m\in \N}$ be a family of polynomials such that for every $m\in \N$, $g_m \in \F[x_1, \dots, x_m]$, $d_m := \deg(g_m) \leq O(\log \log (m))$. Further, $\mathcal{G}$ has the property that for any depth $\Delta \in \N$, if $\mathcal{C} = \{C_m\}_{m\in \N}$ is a family of depth-$\Delta$ circuits computing $\mathcal{G}$, then $\mathcal{C}$ requires size $m^{d_m^{\exp(-O(\Delta))}}$, which is $m^{\omega(1)}$. \cref{LST-hardness} gives us such a family of explicit low-degree polynomials that are hard for constant-depth circuits.

\irredpreserve* 
\begin{proof}
     Suppose $P(T,\vecx,z) = \prod_{i=1}^D(z-\varphi_i(T,\vecx))$, where each $\varphi_i(T,\vecx)$ is a power series root from $\F[\vecx]\llbracket T \rrbracket$. Then, by \autoref{cor:identifying-conjugates-to-divisibility-testing-III}, we have that for any subset $S\subseteq [D]$, $F_S(T,\vecx,z) := \prod_{i\in S} (z - \varphi_i(T,\vecx))$ is an irreducible factor of $P(T,\vecx,z)$ in $\F[T,\vecx,z]$ if and only if 
    \begin{itemize}
         \item $Q_{S}(T,\vecx,z)$ divides $P(T,\vecx,z)$, where $Q_S = F_S \trunc T^k$ and
         \item for every $\emptyset \neq U \subsetneq S$ such that $\prod_{i\in U}(z - \alpha_i)$ has all coefficients in $\F$, we have that $Q_U(T, \vecx, z) \in \F[T, \vecx, z]$ does not divide $P(T, \vecx, z)$
     \end{itemize}
     For the given parameters, one can verify that the upper bound $s'$ given in \cref{thm:ki-preserves-divisibility} is at most $\poly(n)$, whereas $g_\sigma$ requires size $n^{\omega(1)}$ for any constant-depth circuit\footnote{We chose $n$ sufficiently large enough in the theorem statement ($n>C$), which implies that for this particular $\sigma$, $g_\sigma$ requires constant-depth circuits larger than $s'$.}. Thus, we can apply \autoref{thm:ki-preserves-divisibility} on both the conditions above, to equivalently state that for any subset $S\subseteq [D]$, $F_S(T,\vecx,z) = \prod_{i\in S} (z - \varphi_i(T,\vecx))$ is an irreducible polynomial factor of $P(T,\vecx,z)$ if and only if 
     \begin{itemize}\label{list:irreducible-factors}
         \item $Q_{S}(T,\ki_{g_{\sigma},\mathcal{S}}(\vecw),z)$ divides $P(T,\ki_{g_{\sigma},\mathcal{S}}(\vecw),z)$, and
         \item for every $\emptyset \neq U \subsetneq S$ such that $\prod_{i\in U}(z - \alpha_i)$ has all coefficients in $\F$, we have that $Q_U(T,\ki_{g_{\sigma},\mathcal{S}}(\vecw),z)$ does not divide $P(T,\ki_{g_{\sigma},\mathcal{S}}(\vecw),z)$
     \end{itemize}
     \autoref{lem:preserving-root-properies-under-homomorphism} tells us that $P(T, \ki_{g_{\sigma},\mathcal{S}}(\vecw), z)$ will remain a nonzero squarefree\footnote{Squarefreeness is equivalent to each root being non-degenerate, which is maintained by \autoref{lem:preserving-root-properies-under-homomorphism}.} polynomial that is monic in $z$ and $T$-regularized, with $\deg_z P = D$. Furthermore, $P(0,\ki_{g_{\sigma},\mathcal{S}}(\vecw),z) = P(0,\mathbf{0},z)$ will be squarefree. Thus, we can apply \autoref{cor:identifying-conjugates-to-divisibility-testing-III} to the above two conditions to get that for any subset $S\subseteq [D]$, $F_S(T,\vecx,z) = \prod_{i\in S} (z - \varphi_i(T,\vecx))$ is an irreducible polynomial factor of $P(T,\vecx,z)$ if and only if
     \begin{itemize}
         \item $\prod_{i\in S} (z - \varphi_i(T,\ki_{g_{\sigma},\mathcal{S}}(\vecw))) \in Q[T,\vecw,z]$, and
         \item for every $\emptyset \neq U \subsetneq S$ such that $\prod_{i\in U}(z - \alpha_i)$ has all coefficients in $\F$, we have that $\prod_{i\in T} (z - \varphi_i(T,\ki_{g_{\sigma},\mathcal{S}}(\vecw))) \not\in \Q[T,\vecw,z]$
     \end{itemize} 
     These conditions are true if and only if $F_S(T,\ki_{g_{\sigma},\mathcal{S}}(\vecw),z) = \prod_{i\in S} (z - \varphi_i(T,\ki_{g_{\sigma},\mathcal{S}}(\vecw)))$ is an irreducible factor of $P(T,\ki_{g_{\sigma},\mathcal{S}}(\vecw),z)$; this concludes the theorem.
\end{proof}

\section{Proof of correctness of the algorithm} \label{sec:analysis-of-algorithms}

We will analyze the correctness of our algorithms in a bottom-up fashion, starting with \cref{alg:univariate-to-factor} and ending with \cref{alg:all-factors}, which will prove \cref{thm:main-algorithm}.

\subsection{Analysis of {\Cref{alg:univariate-to-factor}}}

Suppose $P(T,\vecx,z) = \prod_{i \in [m]}G_i(T,\vecx,z)$, where $P(T,\vecx,z)$ is monic in $z$, $T$-regularized with $\deg_T(P) \leq \deg(P) := D$ and $\deg_z(P) = D_z$. $P$ is squarefree and $P(0,\veczero,z)$ is squarefree as well. The $G_i$s are the irreducible factors of $P$. 

Let $G(T,\vecx,z)$ denote an arbitrary $G_i$. Suppose $G(0,\veczero,z) = \prod_{j \in [r]} H_j(z)$ is the factorization of $G(0,\veczero,z)$ into its irreducible factors. Since $P(0,\veczero,z)$ is squarefree, each $H_j$ is distinct. If any of the $H_j(z)$ had degree 1, we could've lifted that root via Newton iteration and proceeded with the rest of the algorithm. We will now describe how we deal with the absence of roots of $G(0,\veczero,z)$ in $\Q$ using ideas that are standard in the literature (for instance, see \cite[Section 6.2]{DSS22-closure}).

Let $H(z)$ be an arbitrary irreducible factor of $G(0,\veczero,z)$ and define the field $\mathbb{K} := \frac{\Q[u]}{H(u)}$ for a new variable $u$. Thus, we are artificially adding a root of $H(u)$ to our field. Every operation in this field is polynomial addition or multiplication over $\Q$ with the variable $u$, followed by taking the reminder $\mod H(u)$. 

We can now use Newton iteration (\cref{lem:newton-iteration-linear}) to compute a truncated, non-degenerate, approximate $z$-root $\varphi(T,\vecx) \in \mathbb{K}[T,\vecx]$ satisfying $G(T,\vecx,\varphi(T,\vecx)) \equiv 0 \mod \inangle{T}^{2D\cdot D_z + 1}$ and $\varphi(0,\vecx) = u \in \mathbb{K}$. The following lemma, standard in the factorization literature (see \cite[Lemma 3.6]{B04}), tells us that $G$ is the unique minimal polynomial of $\varphi$ with the above constraints.

\begin{lemma}\label{lem:unique-minimal-polynomial}
    Let $G(T,\vecx,z) \in \Q[T,\vecx,z]$ be an irreducible polynomial, monic in $z$ with $T$-degree at most $D$ and $z$-degree exactly $D_z$. Let $\varphi(T,x) \in \mathbb{K}[T,\vecx]$ be an approximate root of $G$ of order $2D\cdot D_z + 1$, satisfying $G(T,\vecx,\varphi(T,\vecx)) \equiv 0 \mod \inangle{T}^{2D\cdot D_z + 1}$.
    
    Now, suppose $B(T,z) \in \mathbb{Q}(\vecx)[T,z]$ is monic in $z$ with $T$-degree at most $D$ and $z$-degree exactly $D_z$ satisfying:
    \[
        B(T,\varphi(T,\vecx)) \equiv 0 \mod \inangle{T}^{2D\cdot D_z + 1}
    \]
    Then $B \equiv G$. 
\end{lemma}
\begin{proof}
    Consider $R(T) := \Res{z}{B(T,z)}{G(T,z)} \in \Q(\vecx)[T]$. By \cref{thm:resultants-gcd}, there exist polynomials $A_B(T,z), A_G(T,z) \in \Q(\vecx)[T,z]$ such that 
    \[
        R(T) \equiv A_B(T,z)B(T,z) + A_G(T,z)G(T,z) \, . 
    \]
    If we plug in $z = \varphi$, both $B(T,\varphi)$ and $G(T,\varphi)$ vanish $\bmod \inangle{T}^{2D\cdot D_z + 1}$. Thus, $R(T) \equiv 0 \mod  \inangle{T}^{2D\cdot D_z + 1}$. But the definition of the Resultant tells us that $R(T)$ has $T$-degree at most $2D\cdot D_z$, which means that $R(T)$ must be identically zero. By \cref{thm:resultants-gcd}, this implies that $\gcd(B,G)$ is non-trivial in $\Q(T, \vecx)[z]$, but since both $B$ and $G$ are monic in $z$, \cref{lem:gauss} tells us that $\gcd(B,G)$ is non-trivial in $\Q[T, \vecx,z]$. Since $G$ is irreducible, $G$ must divide $B$. Moreover, $\deg_z(B) = \deg_z(G)$, which implies that $B \equiv G$. 
\end{proof}

\subsubsection{Linear system for computing the unique minimal polynomial}

Let $\mathcal{G}=\{g_m\}_{m\in \N}$ be a family of polynomials such that for every $m\in \N$, $g_m \in \F[x_1, \dots, x_m]$, $d_m := \deg(g_m) \leq O(\log \log (m))$. Further, $\mathcal{G}$ has the property that for any depth $\Delta \in \N$, if $\mathcal{C} = \{C_m\}_{m\in \N}$ is a family of depth-$\Delta$ circuits computing $\mathcal{G}$, then $\mathcal{C}$ requires size $m^{d_m^{\exp(-O(\Delta))}}$, which is $m^{\omega(1)}$. \cref{LST-hardness} gives us such a family of explicit low-degree polynomials that are hard for constant-depth circuits.

\begin{theorem}[Small division-free circuit for the minimal polynomial of an approximate root]\label{thm:algo-circuit-min-poly-of-approx-root}
    Fix any $\Delta \in \N$ and $\varepsilon \in (0,0.5)$. For an absolute constant $C_{\Delta,\varepsilon} \in \N$, let $n > C_{\Delta,\varepsilon}$ and $\vecx = (x_1, \dots, x_n)$. Let $P(T,\vecx,z)$ be a polynomial with the following properties.
    \begin{itemize}
        \item $P$ is $T$-regularized and monic in $z$ with total degree $D$.
        \item $G(T,\vecx,z)$ is an irreducible factor of $P$ with $z$-degree $D_z$.
        \item $H(z)$ is an irreducible factor of $G(0,\vecx,z)$, and $\K$ is the field $\frac{\Q[u]}{H(u)}$.
    \end{itemize}
    Then, there is a deterministic algorithm $\mathcal{A}_{\Delta,\varepsilon}$ (with access to $H(z)$) which 
    \begin{itemize}
        \item takes as input a size $s$ circuit computing $\varphi(T,\vecx) \in \K[T,\vecx]$, a truncated, approximate $z$-root of $P(T,\vecx,z)$ of order $(2D\cdot D_z + 1)$ such that $\varphi(0,\vecx)=u$ is a root of $G(0,\vecx,z)$ over $\K$;
        \item outputs a division-free circuit over $\Q$ for $G(T,\vecx,z)$ with size $\poly(s,D)$;
        \item and runs in time $\poly(s,D)^{O(n^{2\varepsilon})}$.
    \end{itemize} 
\end{theorem}
\begin{proof}
We would like to compute a polynomial $B(T,z) \in \Q(\vecx)[T,z]$ that is monic in $z$ with $T$-degree at most $D$ and $z$-degree exactly $D_z$ satisfying:
\[
    B(T,\varphi(T,\vecx)) \equiv 0 \mod \inangle{T}^{2D\cdot D_z + 1} \,
\]
where $\varphi(T,\vecx) \in \mathbb{K}[T,\vecx]$ is computed by a circuit with constants from $\K$. \Cref{lem:unique-minimal-polynomial} guarantees that such a polynomial $B$ has to be the irreducible factor $G$.  

We can separate $B$ into its bihomogeneous components with respect to $T$ and $z$ as \[ B(T,z) = \sum_{i=0}^{d}\sum_{j=0}^{D} B_{i,j} z^i T^j \, \]
where each $B_{i,j}$ can take a value in the field $\Q(\vecx)$. Since $B$ is monic in $z$, $B_{d,0} = 1$ and $B_{d,j} = 0$ for all $j>0$. 
The minimal polynomial satisfies
\[
    B(T,\varphi(T,\vecx)) = \sum_{i=0}^{d}\sum_{j=0}^{D} B_{i,j} (\varphi(T,\vecx))^i T^j \equiv 0 \mod \inangle{T}^{2D\cdot D_z + 1} 
\]
which can equivalently be written as $2D\cdot D_z + 1$ many linear constraints, where each constraint expresses that the coefficient of $T^l$ in $B(T,\varphi(T,\vecx))$ is zero, for some $l\in \{0, 1, \dots, 2D\cdot D_z\}$. Let $\varphi^{(i,j)}$ denote the coefficient of $T^j$ in $\varphi^i$. Thus, we have the following linear system in the variables $B_{i,j}$ for $i \in \{0, \dots, D\}$ and $j \in \{0, \dots, D_z\}$.
\begin{align}
    & B_{d,0} = 1 \\
    \forall j\in\{1, \dots, D\}: \, & B_{d,j} = 0 \\ 
    \forall l \in \{0,1,\dots,2D\cdot D_z + 1\}: \,  & \sum_{i=0}^{d}\sum_{j=0}^{l} B_{i,j} (\varphi)^{(i,l-j)} = 0 \label{eqn:linear-system-constraint-with-u}
\end{align}
At this point, the coefficients of each constraint come from $\mathbb{K}[\vecx]$, which means that we can interpret each constraint as a degree $(\deg(H)-1)$ polynomial in the variable $u$, with coefficients from $\Q[\vecx]$. Since the minimal polynomial of $u$ is $H$, it follows that $1, u, \dots, u^{\deg(H)-1}$ are linearly independent over $\Q$ and, in fact, over\footnote{If $\sum_i p_i(\vecx)u^i \equiv 0$, then $\sum_i p_i(\veczero)u^i = 0$.} $\Q[\vecx]$. Thus, for a constraint to be equal to zero, the coefficient of each $u^r$ (for $r \in \{0,1,\dots, \deg(H)-1\})$ must be identically zero in $\Q[\vecx]$. To compute the coefficient of $u^r$ in $(\varphi)^{(i,j)}$, we can apply \cref{cor:interpolation-consequences} and get a circuit of size $\poly(s)$; we will denote this circuit by $(\varphi)^{(i,j,r)}$. Thus, we can express the constraints of the form \cref{eqn:linear-system-constraint-with-u} using equivalent constraints of the form:
\[
    \sum_{i=0}^{d}\sum_{j=0}^{l} B_{i,j} (\varphi)^{(i,l-j,r)} = 0
\]
for every $r \in \{0, \dots, (\deg(H)-1)D_z\}$ and for every $l \in \{0,1,\dots,2D\cdot D_z + 1\}$. The linear system is now over $\Q[\vecx]$, and thus, any solution to the linear system is going to be over $\Q(\vecx)$.

We can express the linear system as $M_{\varphi} \vec{B} = \vec{c_\varphi}$ for a matrix $M_{\varphi}$ with entries in $\Q[\vecx]$, the vector of variables $\vec{B} = (B_{i,j})$ and a vector $\vec{c_\varphi} \in (\Q[\vecx])^n$. Since the linear system has a unique solution (\cref{lem:unique-minimal-polynomial}), $M_\varphi$ has full column rank. Using some basic linear algebra (for instance, see \cite[Lemma B.6]{KRSV}), we can rewrite the linear system as $M_{\varphi}^T M_{\varphi} \vec{B} = M_{\varphi}^T \vec{c_\varphi}$, where $M_{\varphi}^T M_{\varphi}$ is an invertible square matrix. Thus, we can express the solution to this linear system as $\vec{B} = (M_{\varphi}^T M_{\varphi})^{-1} (M_{\varphi}^T \vec{c_\varphi})$. In particular, each $B_{i,j}$ can be expressed as $\frac{N_{i,j}(\vecx)}{\det(M_{\varphi}^T M_{\varphi})}$, where both the numerator and the denominator have circuits of size $\poly(s,D)$.\footnote{Here, we use the well-known fact that the Determinant can be computed efficiently.}

At this point, we would like to use Strassen's division elimination (\cref{lem:div-elimination-algorithm}) to write each $B_{i,j}$ as a circuit without division. Strassen's method requires that we find a point $\vec{\gamma} \in \Q^n$ such that $\det(M_\varphi^T M_\varphi)$ is non-zero at $\vec{\gamma}$, but this seems to require a hitting set for circuits since the naive upper bound that we can give for $\det(M_\varphi^T M_\varphi)$ is a small circuit. We will now show that the variable reduction map from \cref{thm:irreducibility-preservation} also preserves the nonzeroness of $\det(M_\varphi^T M_\varphi)$. 

Let $\sigma = O(n^\varepsilon),\mu = O(\frac{n^{2\varepsilon}}{\log(n)}),\rho = O(\log(n))$, and let $\mathcal{S}$ be an $(n,\sigma ,\mu ,\rho)$-design. Let $\ki_{g_{\sigma},\mathcal{S}}: \F^{\mu} \to \F^{n}$ be the polynomial map in \autoref{def:KI-generator} defined using the design $\mathcal{S}$ and the polynomial $g_{\sigma}$ from the family of hard polynomials  $\mathcal{G}$.

\begin{claim}
    $\det(M_\varphi^T M_\varphi) \circ \ki_{g,\mathcal{S}}(\vecw) \not \equiv 0$
\end{claim}
\begin{proof}
    Recall that $\det(M_{\varphi}^T M_{\varphi})$ is non-zero precisely because the irreducible factor $G(T,\vecx,z)$ is the unique solution for the linear system. \cref{lem:preserving-root-properies-under-homomorphism} tells us that $\varphi(T,\ki_{g_\sigma,\mathcal{S}}(\vecw))$ is the unique truncated, non-degenerate, approximate $z$-root of $P(T,\ki_{g_\sigma,\mathcal{S}}(\vecw),z)$ of order $2D\cdot D_z+1$, satisfying $\varphi(T,\ki_{g_\sigma,\mathcal{S}}(\vecw)) = u = \varphi(0,\vecx)$. By \cref{thm:irreducibility-preservation}, we now know that $\ki_{g_\sigma,\mathcal{S}}(\vecw)$ preserves the irreducibility of factors, so $G(T,\ki_{g_\sigma,\mathcal{S}}(\vecw),z)$ is irreducible. Thus, if we had first applied $\ki_{g_\sigma,\mathcal{S}}(\vecw)$ to the polynomial $P(T,\vecx,z)$ and computed the truncated, non-degenerate, approximate root $\varphi(T,\ki_{g_\sigma,\mathcal{S}}(\vecw))$ of order $2D\cdot D_z +1$ starting from the same point $u \in \mathbb{K}$, \cref{lem:unique-minimal-polynomial} tells us that we would've still maintained uniqueness of solution for the linear system that computes a monic minimal polynomial of $\varphi(T,\ki_{g_\sigma,\mathcal{S}}(\vecw))$ of $T$-degree at most $D$ and $z$-degree exactly $D_z$. Moreover, the new linear system is going to be the same as the old linear system, except each $(\varphi)^{(i,l-j,r)}$ will be replaced by $(\varphi \circ \ki_{g_\sigma,\mathcal{S}}(\vecw))^{(i,l-j,r)}$; this follows because the transformation $\vecx \mapsto \ki_{g_\sigma,\mathcal{S}}(\vecw)$ preserves the $(T,z)$-degree and the $u$-degree of every monomial; in other words, the matrix $M_{\varphi \circ \ki_{g_\sigma,\mathcal{S}}(\vecw)}$ will be equal to $M_\varphi \circ \ki_{g_\sigma,\mathcal{S}}(\vecw)$. Thus, uniqueness of the linear system's solution implies that $\det(M_\varphi^T M_\varphi)(\vecx) \circ \ki_{g_\sigma,\mathcal{S}}(\vecw) \equiv \det(M_{\varphi \circ \ki_{g_\sigma,\mathcal{S}}(\vecw)}^T M_{\varphi \circ \ki_{g_\sigma,\mathcal{S}}(\vecw)}) \not\equiv 0$.
\end{proof}

Thus, to find a point $\vec{\gamma} \in \Q^n$ satisfying $\det(M_\varphi^T M_\varphi)(\vec{\gamma}) \neq 0$, we can first compose it with $\ki_{g_\sigma,\mathcal{S}}(\vecw)$ to get a degree $\poly(n)$ polynomial on $\mu = O(n^{2\varepsilon})$ variables. For this low-variate polynomial, we can use the brute force derandomization of the Polynomial Identity Lemma (\cref{lem:SZ}) to find a point $\Hat{\vec{\gamma}} \in \Q^\mu$ such that $\det(M_\varphi^T M_\varphi) \circ \ki_{g_{\sigma},\mathcal{S}}$ is nonzero at $\Hat{\vec{\gamma}}$; this takes time $\poly(s,D)^{O(n^{2\varepsilon})}$. Thus, $\vec{\gamma} := \ki_{g_\sigma,\mathcal{S}}(\Hat{\vec{\gamma}}) \in \Q^n$ will be a point where $\det(M_\varphi^T M_\varphi)(\vec{\gamma}) \neq 0$. We can then use \cref{lem:div-elimination-algorithm} to give a circuit of size $\poly(s,D)$ for each $B_{i,j}$, which implies that $G(T,\vecx,z) = B(T,z) =  \sum_{i=0}^{d}\sum_{j=0}^{D} B_{i,j} z^i T^j$ has a circuit of size $\poly(s,D)$.

\end{proof}

\begin{theorem}[Correctness of {\Cref{alg:univariate-to-factor}}]\label{thm:analysis-univariate-to-factor}
    For any $\Delta \in \N$ and $\varepsilon \in (0,0.5)$, \cref{alg:univariate-to-factor}, for sufficiently large $n$, satisfies the following.
    \begin{itemize}
        \item It takes as input a depth-$\Delta$ size-$s$ circuit $C_{\tilde{P}}$ computing a squarefree degree-$D$ polynomial $\tilde{P}(T,\vecx,z)= \prod_{j \in [m]} \tilde{G}_j(T,\vecx,z)$, where the $\tilde{G}_i$s are the irreducible factors of $\tilde{P}$. Further, $\tilde{P}(T,\vecx,z)$ is monic in $z$, $T$-regularized with respect to $\vecx$ and $C_{\tilde{P}}(0,\vecx,z) = C_{\tilde{P}}(0,\veczero,z)$ is squarefree.
        \item It also takes as input the univariate polynomial $\tilde{G}_j(0,\vecx,z) = \tilde{G}_j(0,\veczero,z)$ for each $j \in [m]$.
        \item It outputs a list $L = \{C_{\tilde{G}_1}(T,\vecx,z), \dots, C_{\tilde{G}_m}(T,\vecx,z)\}$, such that each $C_{\tilde{G}_i}(T,\vecx,z)$ is a circuit of size $\poly(s,D)$ and depth $\poly(D)$, which computes the irreducible factor $\tilde{G}_i(T,\vecx,z)$ of $\tilde{P}(T,\vecx,z)$.
        \item It runs in time $\poly(s,D)^{O(n^{2\varepsilon})}$.
    \end{itemize}
\end{theorem}
\begin{proof}
    The algorithm iterates over each $\tilde{G}_j$, so let us focus on one such $\tilde{G}_j$. By \cref{thm:-LLL-univariate-factorization}, \cref{alg3-line:univariate-factorization} of the algorithm gives the correct factorization of $\tilde{G_j}(0,\veczero,z)$ into its irreducible factors, and it runs in time $\poly(\deg(\tilde{G_j}(0,\veczero,z)))$. By \cref{lem:newton-iteration-linear-circuit-and-uniqueness}, \cref{alg3-line:newton-iteration} will correctly compute a truncated, non-degenerate, approximate $z$-root $\varphi(T,\vecx) \in \mathbb{K}[T,\vecx]$ of $\tilde{G}_j(T,\vecx,z)$ of order $2D\cdot D_z+1$ such that $\varphi(0,\vecx) = \varphi(0,\veczero) = u \in \mathbb{K}$. Moreover, it runs in time $\poly(s,D)$ and outputs a circuit for $\varphi(T,\vecx)$ of size $\poly(s,D)$. Finally, by \cref{thm:algo-circuit-min-poly-of-approx-root}, \cref{alg3-line:lin-system-div-free-ckt} will correctly compute a circuit $C_{\tilde{G_j}}$ for $\tilde{G_j}$, of size $\poly(s,D)$, in time $\poly(s,D)^{O(n^{2\varepsilon})}$.
\end{proof}

\subsection{Analysis of {\Cref{alg:factors-squarefree}}}
\begin{theorem}[Correctness of {\Cref{alg:factors-squarefree}}]\label{thm:analysis-factors-squarefree}
    For any $\Delta \in \N$ and $\varepsilon \in (0,0.5)$, \cref{alg:factors-squarefree}, for sufficiently large $n$, satisfies the following.
    \begin{itemize}
        \item The input to the algorithm is a depth-$\Delta$ size-$s$ circuit $C_{P}$ computing the squarefree polynomial $P(\vecx)= \prod_{j \in [m]} G_j(\vecx)$, where the $G_j$s are the irreducible factors of $P$.
        \item The output of the algorithm is a list $L = \{C_{G_1}(\vecx), \dots, C_{G_m}(\vecx)\}$, such that each $C_{G_j}(\vecx)$ is a circuit of size $\poly(s,D)$ and depth $\poly(D)$, which computes the irreducible factor $G_j(\vecx)$.
        \item The algorithm runs in time $\poly(s,D)^{O(n^{2\varepsilon})}$.
    \end{itemize}
\end{theorem}
\begin{proof}    
    The coefficient of $z^D$ in $\Hat{P}(\vecx,z) := P(\vecx + (\veca \cdot z))/\delta$ is $\homog_D[P]/\delta = 1$, thus $\Hat{P}$ is monic in $z$, and $\Hat{P}$ retains the squarefreeness of $P$. Moreover, \cref{cor:interpolation-consequences} tells us that $\homog_D[P]$ has a size $\poly(s,D)$ and depth $(\Delta + O(1))$ circuit, so \cref{thm:lst-PIT} outputs $\veca$ satisfying $\homog_D[P](\veca) \neq 0$, in time $\poly(s,D)^{O(n^{\varepsilon})}$. 
    
    Since $\Hat{P}$ is monic in $z$ and squarefree, \cref{thm:discriminant-squarefreeness} along with \cref{lem:gauss} tells us that $\operatorname{Disc}_z(\Hat{P})(\vecx) \not\equiv 0$. In \cref{algo2-line:discriminant-shift-and-T-regularize}, $\vecb$ is chosen to ensure that $\operatorname{Disc}_z(\Hat{P})(\vecb) \neq 0$. For $\tilde{P}(T,\vecx,z) := \Hat{P}((T\cdot \vecx)+\vecb,z)$, $\operatorname{Disc}_z(\tilde{P}(0,\vecx,z)) =\operatorname{Disc}_z(\tilde{P}(0,\veczero,z)) = \operatorname{Disc}_z(\Hat{P})(\vecb) \neq 0$, which implies that $\tilde{P}(0,\vecx,z) \in \Q[z]$ is squarefree. Moreover, \cref{thm:resultant-constant-depth-andrews-wigderson} (along with an application of \cref{cor:interpolation-consequences} for the complexity of $\frac{\partial \Hat{P}}{\partial z}$) tells us that $\operatorname{Disc}_z(\Hat{P}) = \Res{z}{\Hat{P}}{\frac{\partial \Hat{P}}{\partial z}}$ has a size $\poly(s,D)$ and depth $(\Delta + O(1))$ circuit. Thus, we can again use \cref{thm:lst-PIT} to find $\vecb$ in time $\poly(s,D)^{O(n^{\varepsilon})}$.\\
    In \cref{algo2-line:ki-map-construct}, the map $\ki_{g_\sigma,\mathcal{S}}$ can be constructed in time $\exp(O(n^{\varepsilon}))$ by using \cref{lem:explicit-designs}.\\
    In \cref{algo2-line:factorize-n^eps-variate}, $\tilde{P}(T,\ki_{g_\sigma, \mathcal{S}}(\vecw),z)$ is an $O(n^{2\varepsilon})$-variate polynomial. Using \cref{thm:dense-rep-factorization}, we can factorize $C_{\tilde{P}}(T,\ki_{g,\mathcal{S}}(\vecw),z)$ in time $\poly(s,D)^{O(n^{2\varepsilon})}$. For each $j \in [m]$, let $\tilde{G}_j(T,\vecx,z)$ denote $G_j((T\cdot \vecx) + (\veca \cdot z) + \vecb)$; $\tilde{G_j}$ must be irreducible (this is standard in the literature; for a proof, see \cite[Lemma B.7]{KRSV}). 
    
    By \cref{thm:irreducibility-preservation}, for each $j \in [m]$, $\tilde{G}_j(T,\ki_{g,\mathcal{S}}(\vecw),z)$ is an irreducible factor of $\tilde{P}(T,\ki_{g_\sigma,\mathcal{S}}(\vecw),z)$. Thus, for $j \in [m]$, the factorization in \cref{algo2-line:factorize-n^eps-variate} will compute a circuit $C_{\tilde{G}_j}(T,\ki_{g,\mathcal{S}}(\vecw),z)$ that computes $\tilde{G}_j(T,\ki_{g,\mathcal{S}}(\vecw),z)$, satisfying the property that $\tilde{G}_j(0,\ki_{g,\mathcal{S}}(\vecw),z) = \tilde{G}_j(0,\veczero,z)$. 
    
    By \cref{thm:analysis-univariate-to-factor}, \cref{alg:univariate-to-factor} correctly computes the list $L' = \{C_{\tilde{G}_1}(T,\vecx,z), \dots, C_{\tilde{G}_m}(T,\vecx,z)\}$ in time $\poly(s,D)^{O(n^{\varepsilon})}$. \cref{algo2-line:undo-preprocessing} computes $C_{G_j}(\vecx) = C_{\tilde{G}_1}(1,\vecx-\vecb,0)$ for each $G_j(\vecx)$ correctly. \\
    Overall, the algorithm takes time $\poly(s,D)^{O(n^{\varepsilon})}$ and outputs circuits of size $\poly(s,D)$ for each irreducible factor $G_j(\vecx)$ of $P(\vecx)$. 
\end{proof}

\subsection{Analysis of {\Cref{alg:all-factors}}}

\mainalgotheorem* 
\begin{proof}
    \cref{alg:all-factors}, instantiated with parameters $\Delta$ and $\varepsilon$, is the algorithm $\mathcal{A}_{\Delta, \varepsilon}$ claimed in the theorem statement. \cref{thm:andrews-wigderson-squarefree-decomposition} ensures the correctness of \cref{algo1-line:andrews-wigderson-squarefree} so that for each $i \in [r]$, $P_i(\vecx)$ corresponds to the product of irreducible factors of $P$ that have multiplicity $i$ in the factorization (and $r$ is the maximum multiplicity of an irreducible factor). The algorithm from \cref{thm:andrews-wigderson-squarefree-decomposition} runs in time $\poly(s,D)$ with oracle access to a PIT algorithm for constant-depth circuits; replacing each oracle access by the algorithm in \cref{thm:lst-PIT} results in a running time of $\poly(s,D)^{O(n^{\varepsilon})}$.  For each $P_i(\vecx)$, \cref{thm:analysis-factors-squarefree} guarantees that the list $L_i$ output by \cref{alg:factors-squarefree} in \cref{algo1-line:run-algo2-squarefree} is exactly a list of circuits computing the irreducible factors of $P_i$; moreover, it runs in time $\poly(s,n^D)^{O(n^{2\varepsilon})}$. Since factors of $P_i$ have multiplicity $i$ in $P$, \cref{algo1-line:add-multiplicity-info} adds the right multiplicity information along with each circuit for factors of $P_i$. Every irreducible factor of $P(\vecx)$ must occur as an irreducible factor in $P_i(\vecx)$ for some $i \in [r]$; thus, every irreducible factor along with its multiplicity will be included in the output list. 
\end{proof}

\ifblind
\else
\paragraph*{Acknowledgements:} The discussions leading to this work started when a subset of the authors were at the workshop on Algebraic and Analytic Methods in Computational Complexity (Dagstuhl Seminar 24381) at Schloss Dagstuhl. We are thankful to the organisers of the workshop and to the staff at Schloss Dagstuhl for the wonderful collaborative atmosphere that facilitated these discussions. 

\fi

{\let\thefootnote\relax
\footnotetext{\textcolor{\gitinfonotecolour}{\gitinfonote \easteregg}
}}
\bibliographystyle{customurlbst/alphaurlpp}
\bibliography{crossref,references}

\end{document}